\documentclass[11pt,draftcls,onecolumn]{IEEEtran}

\usepackage{latexsym}
\usepackage{graphicx}
\usepackage{array}
\usepackage{amsmath}
\usepackage{amsfonts}
\usepackage{amssymb}
\usepackage{amsthm}
\usepackage{epsfig}
\usepackage{epstopdf}
\usepackage{cite}
\usepackage{mathtools}
\usepackage{etoolbox}
\usepackage{lipsum}
\usepackage{caption}
\usepackage{subcaption}
\usepackage{algorithm}
\usepackage{algpseudocode}
\usepackage{setspace}
\usepackage[colorlinks=true]{hyperref}

\let\bbordermatrix\bordermatrix
\patchcmd{\bbordermatrix}{8.75}{4.75}{}{}
\patchcmd{\bbordermatrix}{\left(}{\left[}{}{}
\patchcmd{\bbordermatrix}{\right)}{\right]}{}{}
 
\numberwithin{equation}{section}

\hypersetup{
     colorlinks = true,
     linkcolor = [rgb]{0,0,1},
     anchorcolor = [rgb]{0,0,1},
     citecolor = [rgb]{0.9,0.5,0},
     filecolor = [rgb]{0,0,1},
     pagecolor = [rgb]{1,1,0},
     urlcolor = [rgb]{0.9,0.5,0},
     bookmarks,
        bookmarksopen = true,
        bookmarksnumbered = true,
        breaklinks = true,
        linktocpage,
        pagebackref,
        colorlinks = false,
        linkcolor = [rgb]{0.2,0.6,0.2},
        urlcolor  = [rgb]{0,0,1},
        citecolor = [rgb]{0.9,0.5,0},
        anchorcolor = [rgb]{0.2,0.6,0.2},
        hyperindex = true,
        hyperfigures
}

\newtheorem{theorem}{Theorem}[section]
\newtheorem{lemma}[theorem]{Lemma}

\newtheorem{corollary}[theorem]{Corollary}

\newtheorem{remark}[theorem]{Remark}

\newcommand{\sr}{\stackrel}

\newcommand{\rar}{\rightarrow}

\newcommand{\tri}{\sr{\triangleq}{=}}

\makeatletter
\def\rightharpoonupfill@{\arrowfill@\relbar\relbar\rightharpoonup}
\newcommand{\overrightharpoonup}{%
   \mathpalette{\overarrow@\rightharpoonupfill@}}
\makeatother
\newcommand{\be}{\begin{equation}}
\newcommand{\ee}{\end{equation}}
\newcommand{\bea}{\begin{eqnarray}}
\newcommand{\eea}{\end{eqnarray}}
\newcommand{\bes}{\begin{eqnarray*}}
\newcommand{\ees}{\end{eqnarray*}}
\newcommand{\bi}{\begin{itemize}}
\newcommand{\ei}{\end{itemize}}
\newcommand{\ben}{\begin{enumerate}}
\newcommand{\een}{\end{enumerate}}
\newcommand{\bp}{\begin{problem}}
\newcommand{\ep}{\end{problem}}
\newcommand{\hso}{\hspace{.1in}}
\newcommand{\hst}{\hspace{.2in}}

\newcommand{\noi}{\noindent}

\begin{document}

%

\title{Sequential Necessary and Sufficient Conditions for Capacity Achieving Distributions of Channels with Memory and Feedback}
%
%
\author{ Photios~A.~Stavrou, Charalambos~D.~Charalambous and Christos K. Kourtellaris
\thanks{Part of this paper is accepted for publication in the proceedings of the IEEE International Symposium on Information Theory (ISIT), Barcelona Spain, July 10--15 2016 \cite{stavrou-charalambous-kourtellaris2016isit}.}
\thanks{The authors are with the Department of Electrical and Computer Engineering (ECE), University of Cyprus, 75 Kallipoleos Avenue, P.O. Box 20537, Nicosia, 1678, Cyprus, e-mail: $\{stavrou.fotios,chadcha,kourtellaris.christos\}$@ucy.ac.cy}}

%

\markboth{Submitted to IEEE Transactions on Information Theory}{}

\maketitle

\begin{abstract}
\noi We derive sequential necessary and sufficient conditions for any channel input conditional distribution ${\cal P}_{0,n}\triangleq\{P_{X_t|X^{t-1},Y^{t-1}}:~t=0,\ldots,n\}$ to maximize the finite-time horizon directed information defined by 
\begin{align*}
C^{FB}_{X^n \rar Y^n} \tri \sup_{ {\cal P}_{0,n}}
I(X^n\rightarrow{Y^n}), \hso I(X^n \rar Y^n) =\sum_{t=0}^n{I}(X^t;Y_t|Y^{t-1})
\end{align*}
for channel distributions $\{P_{Y_t|Y^{t-1},X_t}:~t=0,\ldots,n\}$ and $\{P_{Y_t|Y_{t-M}^{t-1},X_t}:~t=0,\ldots,n\}$, where $Y^t\triangleq\{Y_0,\ldots,Y_t\}$ and $X^t\triangleq\{X_0,\ldots,X_t\}$ are the channel input and output random processes, and $M$ is a finite nonnegative integer. 

\noi We apply the necessary and sufficient conditions to application examples of time-varying channels with memory and we derive recursive closed form expressions of the optimal distributions, which maximize the finite-time horizon directed information. Further, we derive the feedback capacity from the asymptotic properties  of the optimal distributions by investigating  the limit  
\begin{align*}
C_{X^\infty \rar Y^\infty}^{FB} \tri \lim_{ n \longrightarrow \infty} \frac{1}{n+1} C_{X^n \rar Y^n}^{FB}
\end{align*}
without any \'a priori assumptions, such as, stationarity, ergodicity or irreducibility of the channel distribution. The necessary and sufficient conditions can be easily extended to a variety of channels with memory, beyond the ones considered in this paper.
\end{abstract}

\begin{IEEEkeywords}
directed information, variational equalities, feedback capacity, channels with memory, sequential necessary and sufficient conditions, dynamic programming.
\end{IEEEkeywords}


\tableofcontents

%
%
%
%

\section{Introduction}\label{section:necessary:sufficient:introduction}

\par Computing feedback capacity for any class of channel distributions with memory, with or without transmission cost constraints, and computing the optimal channel input conditional distribution, which achieves feedback capacity, and determining whether feedback increases capacity, are fundamental and challenging open problems in information theory for half a century. 

\noi Notable exceptions are the Cover and Pombra \cite{cover-pombra1989} characterization of feedback capacity of nonstationary and nonergodic,  Additive Gaussian Noise (AGN) channels  with memory and feedback. The characterization of feedback capacity derived in \cite{cover-pombra1989},  initiated several  investigations for variants of the AGN channel with memory,  such as, the finite alphabet channel with memory investigated by Alajaji in \cite{alajaji1995},  the stationary ergodic version of Cover and Pombra \cite{cover-pombra1989} AGN channel, in which  the channel noise is of   limited memory,  investigated by Kim in \cite{kim2010}, and several generalizations investigated via dynamic programming  by Yang et al. in \cite{yang-kavcic-tatikonda2007ieeeit}. Despite the progress in \cite{cover-pombra1989,alajaji1995,kim2010,yang-kavcic-tatikonda2007ieeeit}, the task of determining the closed form expression of the optimal channel input conditional distribution without any assumptions of stationarity or ergodicity imposed on the AGN channel, remains to this date a challenging problem. Over the last ten years, feedback capacity expressions of certain symmetric channels with  memory, defined on finite alphabets, are  derived in \cite{permuter-cuff-roy-weissman2008,elishco-permuter2014,permuter-asnani-weissman2014ieeeit}, and  in  \cite{kourtellaris-charalambous2015itw},  when  transmission cost constraints are imposed on the channel input distributions. However, the progress has been limited; the fundamental problem  of determining  feedback capacity, and understanding the properties of the optimal channel input distributions for general channels, remains to this date a challenge. Specifically, in  \cite{permuter-cuff-roy-weissman2008,elishco-permuter2014,permuter-asnani-weissman2014ieeeit}, the closed form expressions of feedback capacity are obtained using the symmetry of the channels considered, the capacity achieving input distributions are often not determined, while the methodology is based on an \'a priori assumption of ergodicity of the joint processes.\\
For general channel distributions with memory, the lack of progress in computing feedback capacity is attributed to the absence of a general methodology to solve extremum problems of feedback capacity. In this paper, we utilize recent work found in \cite{charalambous-stavrou2015ieeeit,kourtellaris-charalambous2015aieeeit}, to develop  such a methodology. Specifically, we  derive  sequential necessary and sufficient conditions for channel input distributions to maximize the finite horizon directed information. Then we apply the necessary and sufficient conditions to specific application examples, and we compute recursive expressions for  the finite horizon information feedback capacity and the optimal channel input distributions. We determine the expressions of feedback capacity and the corresponding expressions of the optimal distributions, which achieve it,  from the per unit time limit of the finite time horizon. The application examples include a)
the time-varying Binary Unit Memory Channel Output (BUMCO) channel (defined by \eqref{introduction:section:example:matrix:general:bumco}), b) the time-varying Binary Erasure Unit Memory Channel Output (BEUMCO) channel (defined by \eqref{introduction:example:conjecture:matrix:beumco}), and c) the time-varying Binary Symmetric Two Memory Channel Output (BSTMCO) channel (defined by \eqref{introduction:section:example:matrix:general:bstmco}).
Moreover, we show how to obtain existing results, such as, the POST channel and the Binary State Symmetric Channel (BSSC) investigated in \cite{permuter-asnani-weissman2014ieeeit} and \cite{kourtellaris-charalambous2015itw}, respectively, as degenerated versions of more general channel models.\\
Next, we describe the problem investigated, we give some of the results obtained, and we draw connections to existing literature.

\subsection{Main Problem}
\label{m-prob}
Consider any  channel model  
\begin{align}
&\Big(\Big\{{\cal X}_t:~t=0,\ldots,n\Big\}, \Big\{{\cal Y}_t:~t=0,\ldots,n\Big\}, {\cal C}_{0,n} \triangleq \Big\{{\bf P}_{Y_t|Y^{t-1},X^t}:~t=0,\ldots,n\Big\},\nonumber\\
&\qquad\qquad{\cal P}_{0,n} \triangleq\Big\{ {\bf P}_{X_t|X^{t-1}, Y^{t-1}}:~t=0, \ldots, n\Big\}  \Big)\nonumber
\end{align}
where $X^t\triangleq\{X_0,X_1,\ldots,X_t\}$ and $Y^t\triangleq\{Y_0,Y_1,\ldots,Y_t\}$ are the channel input and output Random Variables (RVs), taking values in ${\cal X}^t=\times_{t=0}^n{\cal X}_t$, ${\cal C}_{0,n}$ is the set of channel distributions,  and ${\cal P}_{0,n}$ is the set of channel conditional distributions.\\
Our objective is to   derived necessary and sufficient conditions for any channel input  conditional distribution from the set ${\cal P}_{0,n}$, to maximize the finite-time horizon directed information from $X^n$ to $Y^n$, defined by 
\begin{align}
 C_{X^n \rightarrow Y^n}^{FB} \triangleq\sup_{{\cal P}_{0,n}}I(X^n\rightarrow Y^n) \label{cap_fb_1}
 \end{align}
where $I(X^n\rightarrow Y^n)$ is the directed information from $X^n$ to $Y^n$, defined by \cite{marko1973,massey1990}
\begin{align}
I(X^n\rightarrow Y^n) \triangleq \sum_{t=0}^n I(X^t;Y_t|Y^{t-1})= \sum_{t=0}^n {\bf E} \Big\{ \log \Big( \frac{ d{\bf P}_{Y_t|Y^{t-1}, X^t}(\cdot|Y^{t-1}, X^t)}{d{\bf P}_{Y_t|Y^{t-1}}(\cdot|Y^{t-1})}(Y_t)\Big)\Big\}. \label{intro_fbc1a}
\end{align}
\noi We prefer to derive  necessary and sufficient conditions for extremum problem (\ref{cap_fb_1}), because these  translate into corresponding  necessary and sufficient conditions for any channel input distribution to maximize its per unit time limiting version, defined by 
\begin{align}
 C_{X^\infty \rar Y^\infty}^{FB}\triangleq\liminf_{n\longrightarrow\infty}\frac{1}{n+1}C_{X^n \rightarrow Y^n}^{FB}. \label{cap_fb_1_IH}
 \end{align}
Moreover, the transition to the per unit time limit provides significant insight on the asymptotic properties of optimal channel input conditional distributions.

\noi We also derived necessary and sufficient conditions for channel input conditional distributions, which satisfies transmission cost constraint of the form 
\begin{align}
{\cal P}_{0,n}(\kappa)\triangleq  \Big\{ {\bf P}_{X_t|X^{t-1}, Y^{t-1}},  t=0, \ldots, n:  \frac{1}{n+1} {\bf E}\Big\{c_{0,n}(X^n,Y^{n-1})\Big\} \leq \kappa\Big\},~ \kappa \in [0, \infty)\label{cap_fb_3}
\end{align}
and  maximize the finite-time horizon directed information defined by 
\begin{align}
 C_{X^n \rightarrow Y^n}^{FB}(\kappa)\triangleq\sup_{{\cal P}_{0,n}(\kappa)}I(X^n\rightarrow Y^n). \label{cap_fb_1_cost}
 \end{align} 
Subsequently, we illustrate via application examples, that feedback capacity and capacity achieving distributions can be obtained from the asymptotic properties of the solution of the finite-time horizon extremum problem of directed information. To the best of our knowledge, this is the first paper which gives necessary and sufficient conditions for any channel input conditional distribution to maximize the finite-time horizon optimization problems $C_{X^n \rightarrow Y^n}^{FB}$, $C_{X^n \rightarrow Y^n}^{FB}(\kappa)$, and gives non-trivial finite alphabet application examples in which the optimal channel input distribution and the corresponding channel output transition probability distribution are computed recursively.

\par Coding theorems for channels with memory with and without feedback are developed extensively over the years, in an anthology of papers, such as, \cite{dobrushin1959,pinsker1964,gallager1968,blahut1987,ihara1993,verdu-han1994,kramer1998,han2003,kramer2003,cover-thomas2006,kim2008ieeeit,tatikonda-mitter2009,permuter-weissman-goldsmith2009,gamal-kim2011}. Under  certain conditions, $C_{X^\infty \rar Y^\infty}^{FB}$ is the supremum of all achievable rates of the sequence of feedback codes  $\{(n, { M}_n, \epsilon_n):~n=0, \dots\}$ (see \cite{tatikonda-mitter2009} for definition). For the convenience of the reader the definition of feedback codes and the sufficient conditions for $C_{X^\infty \rar Y^\infty}^{FB}$ to correspond to feedback capacity are given in Appendix~\ref{section:feedback:codes}.

%
%
%
%

\subsection{Contributions and Main Results}\label{introduction:contributions}
\par In this paper, to avoid excessive notation, we derive {\it sequential necessary and sufficient conditions} for any channel input distribution $\{{\bf P}_{X_t|X^{t-1},Y^{t-1}}:~t=0,\ldots,n\}\in\{{\cal P}_{0,n},~{\cal P}_{0,n}(\kappa)\}$ to maximize directed information $I(X^n \rightarrow Y^n)$, for the following classes of channel distributions and transmission cost functions.
\begin{align}
&\mbox{\bf Channel Distributions:}\nonumber\\
&\mbox{\bf Class A.}~~~{\bf P}_{Y_t|Y^{t-1}, X^{t}}= {\bf P}_{Y_t|Y_{t-M}^{t-1}, X_t}\equiv q_t(dy_t|y_{t-M}^{t-1}, x_t), ~ t=0, \ldots, n, \label{introduction:class:channel:equation1}\\
&\mbox{\bf Class B.}~~~{\bf P}_{Y_t|Y^{t-1}, X^{t}}= {\bf P}_{Y_t|Y^{t-1}, X_t}\equiv q_t(dy_t|y^{t-1}, x_t), ~ t=0, \ldots, n . \label{introduction:class:channel:equation2}\\
&\mbox{\bf Transmission Cost Functions:}\nonumber\\
&\mbox{\bf Class A.}~~~c^{A.N}_{0,n}(X^n,Y^{n-1})\triangleq\sum_{t=0}^n\gamma_t(X_t,Y_{t-N}^{t-1}),~t=0,\ldots, n, \label{TC_1} \\
&\mbox{\bf Class B.}~~~c^{B}_{0,n}(X^n,Y^{n-1})\triangleq\sum_{t=0}^n\gamma_t(X_t, Y^{t-1}),~t=0,\ldots, n. \label{TC_2} 
\end{align}
Here,  $\{M,N\}$ are nonnegative finite integers. We use the following convention.
\begin{align*}
&\mbox{If $M=0$ then}~{\bf P}_{Y_t|Y_{t-M}^{t-1},X_t}|_{M=0}={\bf P}_{Y_t|X_t},~\mbox{i.e., the channel is memoryless},~t=0,\ldots,n.  \\
&\mbox{If $N=0$ then}~{\gamma}_{t}(x_t,y^{t-1}_{t-N})|_{N=0}=\gamma_t(x_t),~t=0,\ldots,n.  
\end{align*}

\subsubsection{Methodology}
\label{meth}
\noi The starting point of our analysis is based on the information structures of the channel input conditional distribution developed in \cite{kourtellaris-charalambous2015aieeeit}, and the convexity property of the extremum problem of feedback capacity derived in \cite{charalambous-stavrou2012isit,charalambous-stavrou2015ieeeit} for abstract alphabet spaces and in \cite{permuter-asnani-weissman2014ieeeit} for finite alphabet spaces. We translate these convexity properties into convexity properties of dynamic programming recursions.   For the reader's convenience, we introduce the main concepts we invoke in the paper in order to explain the methodology and to state some of the main contributions of this paper.\\

\noi{\it Information Structures of Optimal Channel Input Distributions Maximizing $I(X^n \rar Y^n)$.} 
 From \cite{kourtellaris-charalambous2015aieeeit}, we use the following results.\\ 
(a) For any channel distribution of class $A$, the optimal channel input conditional distribution, which maximizes $I(X^n\rightarrow{Y^n})$ satisfies conditional independence\footnote{For finite alphabet channels with $M=1$, i.e. ${\bf P}_{Y_t|Y_{t-1},X_t}$, it is conjectured in \cite{berger2003,berger-ying2003,chen-berger2005} that  \eqref{introduction:class:input:cost:equation1} holds. The authors were unable to locate, in  the literature,  the derivation of this structural result, besides \cite{kourtellaris-charalambous2015aieeeit}.}
\begin{align}
\Big\{{\bf P}_{X_t|X^{t-1}, Y^{t-1}}= {\bf P}_{X_t|Y_{t-M}^{t-1}}\equiv\pi_t(dx_t|y_{t-M}^{t-1}),~ t=0,\ldots, n\Big\} \subset{\cal P}_{0,n}\label{introduction:class:input:cost:equation1}
\end{align}
which implies the corresponding joint process $\{(X_t,Y_t):~t=0,\ldots,n\}$ is $M$-order Markov, and the output process $\{Y_t:~t=0,\ldots,n\}$ is $M$-order Markov, that is, the joint distribution and channel output transition probability distribution are given by  
\begin{align}
{\bf P}_{Y^t, X^t}^{\pi}(dy^t, dx^t) =&\otimes_{i=0}^t\Big(q_i(dy_i|y_{i-M}^{i-1}, x_i)\otimes \pi_i(dx_i|y_{i-M}^{i-1})\Big),\hso t=0, \ldots, n, \\
{\bf P}_{Y_t|Y^{t-1}}^\pi(dy_t|y^{t-1})=&{\bf P}_{Y_t|Y_{t-M}^{t-1}}^\pi(dy_t|y_{t-M}^{t-1}) \\
=& \int_{{\cal X}_t } q_t(dy_t|y_{t-M}^{t-1}, x_t)\otimes \pi_t(dx_t|y_{t-M}^{t-1}) \equiv\nu_t^\pi(dy_t|y_{t-M}^{t-1}).
\end{align}
(b)  The characterization of ${C}_{X^n \rightarrow Y^n}^{FB}$ called  ``Finite Transmissions Feedback Information'' (FTFI) capacity,  is given by the following expression.
  \begin{align}
{C}_{X^n \rightarrow Y^n}^{FB,A.M}= \sup_{{\cal P}^{A.M}_{0,n}} \sum_{t=0}^n {\bf E}^{ \pi}\left\{
\log\Big(\frac{q_t(\cdot|Y_{t-M}^{t-1}, X_t)}{\nu_t^{\pi}(\cdot|Y_{t-M}^{t-1})}(Y_t)\Big)
\right\}  \label{introduction:cor-ISR_B.2}
\end{align}
where the optimization is over the restricted set of distributions 
\begin{align}
{\cal P}^{A.M}_{0,n}=\Big\{\pi_t(dx_t|y_{t-M}^{t-1}):~t=0,\ldots,n\Big\}.
\end{align}
In view of the Markov property of the channel output process, we optimize the characterization of FTFI capacity \eqref{introduction:cor-ISR_B.2} to determine the optimal channel input distribution from the set ${\cal P}^{A.M}_{0,n}$. \\

\noi{\it Convexity of Directed Information.}  From \cite{charalambous-stavrou2015ieeeit}, we use the following results.\\
(c) The extremum problem of the characterization of FTFI capacity ${C}_{X^n \rightarrow Y^n}^{FB,A.M}$ given by \eqref{introduction:cor-ISR_B.2} is a convex optimization problem, over the space of channel input distributions ${\cal P}^{A.M}_{0,n}$.\\
(d) The characterization of FTFI capacity ${C}_{X^n \rightarrow Y^n}^{FB,A.M}$ can be reformulated as a double sequential maximization problem of concave functionals over appropriate convex subsets of probability distributions.\\

\subsubsection{Sequential Necessary and Sufficient Conditions of the Characterization of FTFI Capacity for Class A Channels}
We derive the sequential necessary and sufficient conditions for the extremum problem \eqref{introduction:cor-ISR_B.2} as follows.

\noi {\it Dynamic Programming Recursions.}  In view of (a)-(d),  we apply dynamic programming and standard techniques of optimization of convex functionals defined on the set of probability distributions, to derive sequential necessary and sufficient conditions for any channel input distribution from the set ${\cal P}^{A.M}_{0,n}$ to achieve the supremum in the characterization of FTFI capacity ${C}_{X^n \rightarrow Y^n}^{FB,A.M}$. \\ 
\noi Specifically, let $C_t: {\cal Y}^{t-1}_{t-M}\longmapsto[0,\infty)$ represent the maximum expected total pay-off in (\ref{introduction:cor-ISR_B.2}) on the future time horizon $\{t,t+1,\ldots,n\}$, given $Y^{t-1}_{t-M}=y^{t-1}_{t-M}$ at time $t-1$, defined by 
\begin{align}
&C_t(y^{t-1}_{t-M})=\sup_{\big\{ {\pi}_i(dx_i|y_{i-M}^{i-1}):~i=t,t+1,\ldots,n\big\}}{\bf E}^{\pi}\bigg\{\sum_{i=t}^n\log\Big(\frac{d{q}_i(\cdot|y^{i-1}_{i-M}, x_i)}{d\nu^\pi_t(\cdot|y^{i-1}_{i-M})}(Y_i)\Big)\Big{|}Y^{t-1}_{t-M}=y^{t-1}_{t-M}\bigg\}.\label{introduction:algorithms:generalizations:lmco:equation17}
\end{align}
\noi The dynamic programming recursions for  (\ref{introduction:algorithms:generalizations:lmco:equation17}) are the following.
\begin{align}
C_n(y^{n-1}_{n-M})=&\sup_{{\pi}_n(dx_n|y_{n-M}^{n-1})}\int_{{\cal X}_n\times{\cal Y}_n}\log\Big(\frac{q_{n}(\cdot|y_{n-M}^{n-1}, x_n)}{{\nu}^{\pi}_{n}(\cdot|y_{t-M}^{n-1})}(y_n)\Big)q_{n}(dy_n|y_{n-M}^{n-1}, x_n)\otimes{\pi}_n(dx_n|y_{n-M}^{n-1}),\label{introduction:algorithms:generalizations:lmco:equation18}\\
C_t(y^{t-1}_{t-M})=&\sup_{{\pi}_t(dx_t|y_{t-M}^{t-1})}\int_{{\cal X}_t\times{\cal Y}_t}\Big(\log\Big(\frac{dq_t(\cdot|y_{t-M}^{t-1}, x_t)}{\nu^{\pi}_{t}(\cdot|y_{t-M}^{t-1})}(y_t)\Big)\nonumber\\
&\qquad\qquad+C_{t+1}(y^t_{t+1-M})\Big)q_t(dy_t|y_{t-M}^{t-1}, x_t)\otimes{\pi}(dx_t|y_{t-M}^{t-1}),~t=0,\ldots,n-1.\label{introduction:algorithms:generalizations:lmco:equation19}
\end{align}
Since   \eqref{introduction:algorithms:generalizations:lmco:equation18}, \eqref{introduction:algorithms:generalizations:lmco:equation19} form a convex optimization problem (sequentially backward in time), we prove the following sequential necessary and sufficient conditions.

\begin{theorem}(Sequential necessary and sufficient conditions for channels of class A)\label{introduction:baa:sequential:theorem:lmco:necessary:sufficient}{\ \\}
The necessary and sufficient conditions for any input distribution $\{\pi_{t}(dx_t|y_{t-M}^{t-1}):~t=0,\ldots,n\}$ to achieve the supremum in $C^{FB,A.M}_{X^n\rightarrow{Y^n}}$ defined by \eqref{introduction:cor-ISR_B.2} (assuming it exists) are the following.\\
\noi{(a)} For each $y_{n-M}^{n-1}\in{\cal Y}_{n-M}^{n-1}$, there exist a  ${C}_n(y_{n-M}^{n-1})$ such that the following hold.
\begin{align}
&\int_{{\cal Y}_{n}}\log\Big(\frac{dq_n(\cdot|y_{n-M}^{n-1}, x_n)}{d\nu^{\pi}_n(\cdot|y_{n-M}^{n-1})}(y_n)\Big)q_n(dy_n|y_{n-M}^{n-1}, x_n)=C_n(y_{n-M}^{n-1}),~\forall{x_n\in{\cal X}_n},~\mbox{if}~{\pi}_n(dx_n|y_{n-M}^{n-1})\neq{0},\label{introduction:baa:sequential:theorem:lmco:necessary:sufficient:1}\\
&\int_{{\cal Y}_{n}}\log\Big(\frac{dq_n(\cdot|y_{n-M}^{n-1}, x_n)}{d\nu^{\pi}_n(\cdot|y_{n-M}^{n-1})}(y_n)\Big)q_n(dy_n|y_{n-M}^{n-1}, x_n)\leq{C}_n(y_{n-M}^{n-1}),~\forall{x_n\in{\cal X}_n},~\mbox{if}~{\pi}_n(dx_n|y_{n-M}^{n-1})={0}\label{introduction:baa:sequential:theorem:lmco:necessary:sufficient:2}
\end{align}
and moreover, $C_n(y_{n-M}^{n-1})$ is the value function defined by \eqref{introduction:algorithms:generalizations:lmco:equation17} at $t=n$.\\
\noi{(b)} For each $t$, $y_{t-M}^{t-1}\in{\cal Y}_{t-M}^{t-1}$, there exist a  ${C}_t(y_{t-M}^{t-1})$ such that the following hold.
\begin{align}
&\int_{{\cal Y}_{t}}\Big(\log\Big(\frac{dq_t(\cdot|y_{t-M}^{t-1}, x_t)}{d\nu^{\pi}_t(\cdot|y_{t-M}^{t-1})}(y_t)\Big)+C_{t+1}(y^t_{t+1-M})\Big)\nonumber\\
&\qquad\qquad{q}_t(dy_t|y_{t-M}^{t-1}, x_t)=C_t(y_{t-M}^{t-1}),~\forall{x_t\in{\cal X}_t},~\mbox{if}~{\pi}_t(dx_t|y_{t-M}^{t-1})\neq{0},\label{introduction:baa:sequential:theorem:lmco:necessary:sufficient:3}\\
&\int_{{\cal Y}_{t}}\Big(\log\Big(\frac{dq_t(\cdot|y_{t-M}^{t-1}, x_t)}{d\nu^{\pi}_t(\cdot|y_{t-M}^{t-1})}(y_t)\Big)+C_{t+1}(y^t_{t+1-M})\Big)\nonumber\\
&\qquad\qquad{q}_t(dy_t|y_{t-M}^{t-1}, x_t)\leq{C}_t(y_{t-M}^{t-1}),~\forall{x_t\in{\cal X}_t},~\mbox{if}~{\pi}_t(dx_t|y_{t-M}^{t-1})={0}\label{introduction:baa:sequential:theorem:lmco:necessary:sufficient:4}
\end{align}
for $t\in\{n-1,\ldots,0\}$, and  moreover, $C_t(Y_{t-M}^{t-1})$ is the value function defined by \eqref{introduction:algorithms:generalizations:lmco:equation17} for $t\in\{n-1,\ldots,0\}$.
\end{theorem}

\noi In application examples of time-varying channels with memory (Section~\ref{section:application:examples}), we invoke Theorem~\ref{introduction:baa:sequential:theorem:lmco:necessary:sufficient} to derive recursive expressions of the optimal channel input distributions. Moreover, from these expressions, we derive the optimal channel input distributions for the per unit time limiting expression $C^{FB}_{X^\infty\rightarrow{Y^\infty}}$, and we show it converges to feedback capacity. 

The necessary and sufficient conditions stated in Theorem~\ref{introduction:baa:sequential:theorem:lmco:necessary:sufficient}, are generalizations of the ones obtained by Gallager \cite{gallager1968} and Jelinek \cite{jelinek1968}, for Discrete Memoryless Channels (DMCs). The main point to be made, is that for channels with memory, we derive the dynamic versions of Gallager and Jelinek's necessary and sufficient conditions, and these are sequential necessary and sufficient conditions. 

\noi In Theorem~\ref{theorem:necessary:sufficient:ftfi:lmco:cost} we derive similar necessary and sufficient conditions for  channel distributions of Class $A$ and  transmission cost functions of Class $A$. In Section~\ref{section:generalizations},  we illustrate how to extend the necessary and sufficient conditions of Theorem~\ref{theorem:necessary:sufficient:ftfi:lmco:cost} to channel distributions of Class $B$ and transmission cost functions of Class $A$ or $B$, and to channel distributions of Class $A$ with transmission cost functions of Class $B$.

\subsubsection{Applications Examples of Necessary and Sufficient Conditions}\label{introduction:applications}
In Section~\ref{section:application:examples}, we apply the sequential necessary and sufficient conditions to derive recursive closed form expressions of optimal channel input conditional distributions, which achieve the characterizations of FTFI capacity of the following channels.

\begin{description}
\item[(a)] The time-varying Binary Unit Memory Channel Output (BUMCO) channel (defined by \eqref{introduction:section:example:matrix:general:bumco}).
\item[(b)] The time-varying Binary Erasure Unit Memory Channel Output (BEUMCO) channel (defined by \eqref{introduction:example:conjecture:matrix:beumco}).

\item[(c)] The time-varying Binary Symmetric Two Memory Channel Output (BSTMCO) channel (defined by \eqref{introduction:section:example:matrix:general:bstmco}).

\end{description}
Further, we consider the time-invariant or homogeneous versions of the BUMCO and BEUMCO channels, and we investigate the asymptotic properties of optimal channel input conditional distributions, by analyzing the per unit time limit of the characterizations of FTFI capacity, specifically, $C_{X^\infty\rightarrow{Y^\infty}}^{FB}$. Via this analysis, we derive the ergodic properties of optimal channel input conditional distributions, which achieve feedback capacity without imposing any \'a priori assumptions, such as, stationarity, ergodicity, or information stability. Rather, we show that the optimal channel input conditional distributions, induce ergodicity of the joint process $\{(X_t, Y_t): t=0,1, \ldots\}$.  \\
Next, we discuss one of the application examples of this paper.

\noi{\it The Time-Varying Binary Unit Memory Channel Output (BUMCO) Channel} In Section~\ref{subsection:applications:ftfi:capacity:bumco}, we apply Theorem~\ref{introduction:baa:sequential:theorem:lmco:necessary:sufficient} to  the time-varying BUMCO channel, denoted by $\{BUMCO(\alpha_t,\beta_t,$ $\gamma_t,\delta_t)$:~$t=0,\ldots,n\}$, and defined by the transition matrix
\begin{align}
q_t(dy_t|x_t,y_{t-1})=\bordermatrix{&0,0&0,1&1,0&1,1\cr
            0&\alpha_t&\beta_t&\gamma_t&\delta_t\cr
            1&1-\alpha_t&1-\beta_t&1-\gamma_t&1-\delta_t\cr},~ \alpha_t, \beta_t, \gamma_t, \delta_t\in[0,1],\alpha_t\neq{\gamma_t},\beta_t\neq{\delta_t}.\label{introduction:section:example:matrix:general:bumco}
\end{align}
\noi That is, for channel \eqref{introduction:section:example:matrix:general:bumco}, the characterization of FTFI capacity is $C^{FB,A.1}_{X^n\rightarrow{Y^n}}$, given by \eqref{introduction:cor-ISR_B.2} with $M=1$.\\
We prove the following theorem.
\begin{theorem}(Optimal solution of $BUMCO$)\label{introduction:theorem:bumco:optimal:solutions}{\ \\}
Consider the time-varying $\{BUMCO(\alpha_t,\beta_t,$ $\gamma_t,\delta_t)$:~$t=0,\ldots,n\}$ defined by \eqref{introduction:section:example:matrix:general:bumco}, and denote the optimal channel input distribution and the corresponding channel output transition probability distribution by $\Big\{{\pi}_t^*(x_t|y_{t-1}):  (x_t, y_{t-1})\in \{0,1\} \times \{0,1\}, t=0,\ldots,n\Big\}$, and $\Big\{ {\nu}_t^{\pi^*}(y_t|y_{t-1}):  (y_t, y_{t-1})\in \{0,1\} \times \{0,1\}, t=0,\ldots,n\Big\}$, respectively. Then the following hold.
\begin{itemize} 
\item[(a)] The optimal distributions are given by the following expressions\footnote{Define $H(x)\triangleq-xlog_2(x)-(1-x)\log_2(1-x),~x\in[0,1]$.}.
\begin{subequations}
\begin{align}
 {\pi}^*_t(0|0)&=\frac{1-\gamma_t(1+2^{\mu_0(t)+\Delta{C}_{t+1}})}{(\alpha_t-\gamma_t)(1+2^{\mu_0(t)+\Delta{C}_{t+1}})},~& {\pi}^*_t(0|1)&=\frac{1-\delta_t(1+2^{\mu_1(t)+\Delta{C}_{t+1}})}{(\beta_t-\delta_t)(1+2^{\mu_1(t)+\Delta{C}_{t+1}})},\\
{\pi}^*_t(1|0)&=1-{\pi}^*_t(0|0),~&{\pi}^*_t(1|1)&=1-{\pi}^*_t(0|1),\\
\nu_t^{\pi^*}(0|0)&=\frac{1}{1+2^{\mu_0(t)+\Delta{C}_{t+1}}},~&\nu_t^{\pi^*}(0|1)&=\frac{1}{1+2^{\mu_1(t)+\Delta{C}_{t+1}}},\\
\nu_t^{\pi^*}(1|0)&=1-\nu_t^{\pi^*}(0|0),~&\nu_t^{\pi^*}(1|1)&=1-\nu_t^{\pi^*}(0|1),\\
\mu_0(\alpha_t,\gamma_t)&=\frac{H(\gamma_t)-H(\alpha_t)}{\gamma_t-\alpha_t}\equiv{\mu}_0(t),~&\mu_1(\beta_t,\delta_t)&=\frac{H(\beta_t)-H(\delta_t)}{\beta_t-\delta_t}\equiv\mu_1(t).
\end{align}\label{section:example:bumco:equation1}
\end{subequations}
where $\{{\Delta}C_t\triangleq{C}_t(1)-C_t(0):~t=0,\ldots,n+1\}$, is the difference of the value functions at each time, satisfying the following backward recursions.
\begin{subequations}
\begin{align}
&\Delta{C}_{n+1}=0,\\
&\Delta{C}_{t}=\Big(\mu_1(t)(\beta_t-1)-\mu_0(t)(\alpha_t-1)\Big)+H(\alpha_t)-H(\beta_t)+\log\Big(\frac{1+2^{\mu_1(t)+\Delta{C}_{t+1}}}{1+2^{\mu_0(t)+\Delta{C}_{t+1}}}\Big),~t\in\{n,\ldots,0\}.  
\end{align}\label{section:example:bumco:equation2b}
\end{subequations}
\item[(b)] The value functions are given recursively by the following expressions.
\begin{align}
C_{t}(0)&=\mu_0(t)(\alpha_{t}-1)+C_{t+1}(0)+\log(1+2^{\mu_0(t)+\Delta{C}_{t+1}})-H(\alpha_{t}),~C_{n+1}(0)=0,\label{introduction:theorem:value:function:eq1}\\
~C_{t}(1)&=\mu_1(t)(\beta_{t}-1)+{C}_{t+1}(0)+\log(1+2^{\mu_1(t)+\Delta{C}_{t+1}})-H(\beta_{t}),~C_{n+1}(1)=0,~t\in\{n,\ldots,0\}.\label{introduction:theorem:value:function:eq2}
\end{align}

\item[(c)] The characterization of the FTFI capacity is given by
\begin{align}
C^{FB,A.1}_{X^n\rightarrow{Y^n}}=\sum_{y_{-1}\in\{0,1\}}C_0(y_{-1}){\bf P}_{Y_{-1}}(dy_{-1}),~{\bf P}_{Y_{-1}}(dy_{-1})\equiv\mu(dy_{-1})~\mbox{is fixed}.\label{introduction:ftfi:capacity:bumco}
\end{align}

\item[(d)] If the channel is time-invariant, denoted by BUMCO$(\alpha,\beta,\gamma,\delta)$, then the following hold.\\
The ergodic feedback capacity  $C^{FB,A.1}_{X^\infty\rightarrow{Y^\infty}}$ is given by the following expression.
\begin{align}
C_{X^\infty\rightarrow{Y}^{\infty}}^{FB,A.1}&=\lim_{n\longrightarrow\infty}\frac{1}{n+1}C_{X^n\rightarrow{Y}^{n}}^{FB,A.1}=\nu_0\Big(H(\nu_{0|0})-H(\gamma)\Big)+(1-\nu_0)\Big(H(\nu_{0|1})-H(\delta)\Big)\nonumber\\
&\qquad\qquad+\xi_0\Big(H(\gamma)-H(\alpha)\Big)+\xi_1\Big(H(\delta)-H(\beta)\Big)\label{introduction:example:ergodic:feedback:capacity:bumco}
\end{align}
where
\begin{subequations}
\begin{align}
\nu_0\equiv\nu^{\pi^{*,\infty}}(0)&=\frac{1+2^{\mu_0+\Delta{C}^{\infty}}}{1+2^{\mu_0+\mu_1+2\Delta{C}^{\infty}}+2^{\mu_0+1+\Delta{C}^{\infty}}},~\xi_0=\frac{1-\gamma(1+2^{\mu_0+\Delta{C}^{\infty}})}{(\alpha-\gamma)\big(1+2^{\mu_0+\mu_1+2\Delta{C}^{\infty}}+2^{\mu_0+1+\Delta{C}^{\infty}}\big)},\nonumber\\
\xi_1&=\frac{2^{\mu_0+\Delta{C}^{\infty}}\big(1-\delta(1+2^{\mu_1+\Delta{C}^{\infty}})\big)}{(\beta-\delta)\big(1+2^{\mu_0+\mu_1+2\Delta{C}^{\infty}}+2^{\mu_0+1+\Delta{C}^{\infty}}\big)},~\nu_{0|0}=\nu^{\pi^{*,\infty}}(0|0),~~~\nu_{0|1}=\nu^{\pi^{*,\infty}}(0|1),\nonumber\\
\mu_0(\alpha,\gamma)&=\frac{H(\gamma)-H(\alpha)}{\gamma-\alpha}\equiv{\mu}_0,~~~\mu_1(\beta,\delta)=\frac{H(\beta)-H(\delta)}{\beta-\delta}\equiv\mu_1.\nonumber
\end{align}\label{introduction:example:ergodic:feedback:capacity:bumco:equation11}
\end{subequations}
$\Delta{C}^{\infty}$ is the steady-state solution of the algebraic equation
\begin{align}
\Delta{C}^{\infty}=&\big(\mu_1(\beta-1)-\mu_0(\alpha-1)\big)+H(\alpha)-H(\beta)+\log\Big(\frac{1+2^{\mu_1+\Delta{C}^{\infty}}}{1+2^{\mu_0+\Delta{C}^{\infty}}}\Big), \label{section:example:bumco:time-invariant:equation2b} 
\end{align}
and $\{\nu^{\pi^{*,\infty}}(y): y \in \{0,1\}\}$ is the unique invariant distribution of $\big\{\nu^{\pi^{*,\infty}}(z|y):~(z,y)\in\{0,1\}\times\{0,1\}\big\}$, given by
\begin{subequations}
\begin{align}
 {\pi}^{*,\infty}(0|0)&=\frac{1-\gamma(1+2^{\mu_0+\Delta{C}^\infty})}{(\alpha-\gamma)(1+2^{\mu_0+\Delta{C}^{\infty}})},~&~{\pi}^{*,\infty}(0|1)&=\frac{1-\delta(1+2^{\mu_1+\Delta{C}^{\infty}})}{(\beta-\delta)(1+2^{\mu_1+\Delta{C}^{\infty}})},\\
 {\pi}^{*,\infty}(1|0)&=1-{\pi}^{*,\infty}(0|0),~&~{\pi}^{*,\infty}(1|1)&=1-{\pi}^{*,\infty}(0|1),\\
\nu^{\pi^{*,\infty}}(0|0)&=\frac{1}{1+2^{\mu_0+\Delta{C}^{\infty}}},~&~\nu^{\pi^{*,\infty}}(0|1)&=\frac{1}{1+2^{\mu_1+\Delta{C}^{\infty}}},\\
\nu^{\pi^{*,\infty}}(1|0)&=1-\nu^{\pi^{*,\infty}}(0|0),~&~\nu^{\pi^{*,\infty}}(1|1)&=1-\nu^{\pi^{*,\infty}}(0|1).
\end{align}\label{section:example:bumco:time-invariant:equation1}
\end{subequations}
\end{itemize}
\end{theorem}

\noi The derivation is given in Section~\ref{subsection:applications:ftfi:capacity:bumco}.  To the best of the authors knowledge, the only other reference, where  closed form expressions for  feedback capacity and capacity achieving distributions are derived,  from the solution of the finite-time horizon directed information extremum problem $C_{X^n \rightarrow Y^n}^{FB}(\kappa)$ defined by (\ref{cap_fb_1_cost}),  is \cite{charalambous-kourtellaris-loyka2016ieeeit}, where analogous results are obtained for Multiple Input Multiple Output Gaussian Linear Channels Models with memory.  

\noi In Sections~\ref{subsection:applications:ftfi:capacity:beumco}, \ref{subsection:applications:ftfi:capacity:bstmco}, we derive analogous results for the BEUMCO channel and the BSTMCO channel, respectively.   

\noi These application examples are by no means exhaustive; they are simply introduced and analyzed in order to  illustrate the effectiveness of the sequential necessary and sufficient conditions for any channel input distribution to maximize the characterizations of FTFI capacity, and their application in computing feedback capacity, via  the asymptotic analysis of the per unit time limit of the characterization of FTFI capacity.

\par This paper is structured as follows. In Section~\ref{baa:sequential:section:problem:formulation}, we give the machinery and background material based on which the results in this paper are developed. In Section~\ref{section:generalizations:lmco}, we derive the sequential necessary and sufficient conditions for channels of class $A$ with transmission cost functions of class $A$. In Section~\ref{section:application:examples} we apply the sequential necessary and sufficient conditions to the BUMCO channel, the BEUMCO channel, and the BSTMCO channel. In Section~\ref{generalizations:abstract:alphabets}, we give sufficient  conditions for the results of the paper to extend to abstract alphabet spaces (i.e., countable, continuous,  mixed, etc.). In Section~\ref{section:generalizations}, we illustrate that the main theorems of Section~\ref{section:generalizations:lmco} extend to channels of class $B$ with transmission cost functions of class $A$ or $B$. We draw conclusions and future directions in Section~\ref{conclusion}.

%
%
%

\section{Preliminaries: Extremum Problems of Feedback Capacity and Background Material}\label{baa:sequential:section:problem:formulation}

%
%
%
%
\par In this section, we introduce the notation, the definition of extremum problem of feedback capacity, and we recall the variational equality derived in \cite{charalambous-stavrou2015ieeeit}.

\subsection{Basic Notation}\label{subsection:basic:notions:probability}
We denote the set of nonnegative integers by $\mathbb{N}_0\triangleq\{0, 1,\ldots\}$, and for any $n\in\mathbb{N}_0$, its restriction to a finite set by  $\mathbb{N}^n_0\triangleq\{0, 1,\ldots,n\}$. Given two measurable spaces $({\cal X}, {\cal  B}({\cal X}))$,  $({\cal Y}, {\cal  B}({\cal Y}))$, we denote the Cartesian product of ${\cal X}$ and ${\cal Y}$ by ${\cal X} \times {\cal Y} \triangleq \{(x,y):  x\in {\cal X}, y \in {\cal Y}\}$, and the product measurable space of $({\cal X}, {\cal  B}({\cal X}))$ and $({\cal Y}, {\cal  B}({\cal Y}))$ by $({\cal X} \times {\cal Y}, {\cal  B}({\cal X})\otimes  {\cal  B}({\cal Y}))$, where  ${\cal  B}({\cal X})\otimes   {\cal  B}({\cal Y})$ is the product $\sigma-$algebra generated by $\{A \times B:  A \in {\cal  B}({\cal X}), B\in  {\cal  B}({\cal Y})\}$. We denote by $H(\cdot)$ the binary entropy, and by $card(\cdot)$ the cardinality of the space.  \\
We denote the probability distribution induced by a Random Variable (RV) $X$ defined on a probability space $(\Omega, {\cal F}, {\mathbb P})$, by the mapping $X: (\Omega, {\cal F}) \longmapsto ({\cal X}, {\cal  B}({\cal X}))$, as follows\footnote{The subscript $X$ is often omitted.}.
\begin{align}
{\bf P}(A) \equiv  {\bf P}_X(A)  \triangleq {\mathbb P}\big\{ \omega \in \Omega: X(\omega)  \in A\big\}, \quad  \forall A \in {\cal  B}({\cal X}).
 \end{align}
 We denote the set of all probability distributions on $({\cal X}, {\cal  B}({\cal X}))$ by ${\cal M}({\cal X})$. A RV $X$ is called discrete if there exists a countable set ${\cal S}_X\triangleq \{x_i: ~i \in {\mathbb{N}_0}\}$ such that $\sum_{x_i \in {\cal S}_X} {\mathbb  P} \{ \omega \in \Omega : X(\omega)=x_i\}=1$. In this case, the probability distribution ${\bf P}_X(\cdot)$  is concentrated on  points in ${\cal S}_X$, and it is defined by 
\begin{align}
 {\bf P}_X(A)  \triangleq \sum_{x_t \in {\cal S}_X \bigcap A} {\mathbb P} \big\{ \omega \in \Omega : X(\omega)=x_t\big\}, \hso \forall A \in {\cal  B}({\cal X}).\nonumber 
\end{align}
If the cardinality of ${\cal S}_X$ is finite then the RV is finite-valued, and we call it a finite alphabet RV. \\
Given another RV, $Y: (\Omega, {\cal F}) \longmapsto ({\cal Y}, {\cal  B}({\cal Y}))$,   ${\bf P}_{Y|X}(dy| X)(\omega)$ is the conditional distribution of RV $Y$ given RV $X$. We denote the conditional distribution of RV $Y$ given $X=x$ (i.e., fixed) by ${\bf P}_{Y|X}(dy| X=x)  \equiv {\bf P}_{Y|X}(dy|x)$. Such conditional distributions are  equivalently  described   by  stochastic kernels or transition functions ${\bf K}(\cdot|\cdot)$ on $ {\cal  B}({\cal Y}) \times {\cal X}$, mapping ${\cal X}$ into ${\cal M}({\cal Y})$ (space of distributions), i.e., $x \in {\cal X}\longmapsto {\bf K}(\cdot|x)\in{\cal M}({\cal Y})$, and such that for every $A\in{\cal B}({\cal Y})$, the function ${\bf K}(A|\cdot)$ is ${\cal B}({\cal X})$-measurable.

%
%
%
%

\subsection{FTFI Capacity and Convexity of Feedback Capacity}\label{subsection:ftfi:feedback:capacity:variational:equalities}

\noi The channel input and  channel output alphabets are  sequences of measurable spaces $\{({\cal X}_t,{\cal  B}({\cal X }_t)):~t\in\mathbb{N}_0\}$ and  $\{({\cal  Y}_t,{\cal  B}({\cal  Y}_t)):~t\in\mathbb{N}_0\}$, respectively, with their product spaces ${\cal X}^{\mathbb{N}_0}\triangleq {{\times}_{t\in\mathbb{N}_0}}{\cal X}_t$, ${\cal  Y}^{\mathbb{N}_0}\triangleq {\times_{t\in\mathbb{N}_0}}{\cal  Y}_t$.  These spaces are endowed with their respective product topologies, and  ${\cal  B}({\Sigma}^{\mathbb{N}_0})\triangleq \otimes_{t\in\mathbb{N}_0}{\cal  B}({\Sigma }_t)$,  denotes the $\sigma-$algebras on ${\Sigma }^{\mathbb{N}_0}$, where ${\Sigma}_t \in  \big\{{\cal X}_t, {\cal  Y}_t\big\}$,  ${\Sigma}^{\mathbb{N}_0} \in  \big\{{\cal X}^{{\mathbb N}_0}, {\cal  Y}^{{\mathbb N}_0}\big\}$,  and generated by cylinder sets. We denote points in ${\Sigma}_k^m \triangleq \times_{j=k}^m {\Sigma}_j$ by $z_{k}^m \triangleq \{z_k, z_{k+1}, \ldots, z_m\} \in {\Sigma}_k^m$,   $(k, m)\in   {\mathbb N}_0 \times {\mathbb N}_0$. \\ 
Below, we introduce the elements of the extremum problem we address in this paper, and we establish the notation.\\

\noi{\bf Channel Distribution with Memory.}  A sequence of conditional distributions defined by 
\begin{align}
{\cal C}_{0,n}\triangleq\Big\{{\bf P}_{Y_t|Y^{t-1}, X^t}=q_t(dy_t|y^{t-1},x^{t}) :  t=0,1, \ldots, n \Big\}. \label{channel1}
\end{align}
 At each time instant $t$ the conditional distribution of the channel depends on past channel output symbols $y^{t-1} \in {\cal Y}^{t-1}$ and current and past channel input symbols $x^{t} \in {\cal X}^{t}$, for $t=0,1, \ldots, n$.\\

\noi{\bf Channel Input Distribution with Feedback.}  A  sequence of conditional distributions defined by 
\begin{align}
{\cal P}_{0,n}\triangleq\Big\{{\bf P}_{X_t|X^{t-1}, Y^{t-1}}=  p_t(dx_t|x^{t-1},y^{t-1}) :   t=0,1, \ldots, n \Big\}. \label{rancodedF}
\end{align}
At each time instant $t$ the conditional channel input distribution with feedback depends on past channel inputs and output symbols  $\{x^{t-1}, y^{t-1}\} \in {\cal X}^{t-1} \times {\cal Y}^{t-1}$, for $t=0,1, \ldots, n$. \\

\noi{\bf Transmission Cost.} The set of channel input distributions with feedback and transmission cost is defined by 
\begin{align}
{\cal P}_{0,n}(\kappa) \triangleq \Big\{  p_t(dx_t|x^{t-1}, y^{t-1}), t=0,1, \ldots, n: \frac{1}{n+1} {\bf E}^{p} \Big( c_{0,n}(X^n, Y^{n-1}) \Big)\leq  \kappa\Big\}\subset{\cal P}_{0,n},~ \kappa \in [0,\infty) \label{rc1}
\end{align}
where the superscript notation ${\bf E}^{p}\{\cdot\}$ denotes the dependence of the joint distribution on the choice of  conditional distribution $\{p_t(dx_t|x^{t-1}, y^{t-1}) :~t=0,1 \ldots, n\}$. The cost of transmitting channel input symbols $x^n\in {\cal X}^{n}$ over a channel, and receiving channel output symbol $y^n\in{\cal Y}^n$,  is a  measurable function $c_{0,n}:{\cal X}^{n}\times{\cal Y}^{n-1} \longmapsto [0,\infty)$. \\

\noi{\bf  FTFI Capacity and Feedback Capacity.}  Given any channel input distribution from the set ${\cal P}_{0,n}$ and a channel distribution from the set ${\cal C}_{0,n}$, we can  uniquely define the induced joint distribution  ${\bf P}^{p}(dx^n, dy^n)$ on  the  canonical space $\Big( {\cal X}^{n} \times {\cal  Y}^{n}, {\cal  B}({\cal X}^{n})\otimes  {\cal  B}({\cal  Y}^{n})\Big)$, and we can construct a probability space $\Big(\Omega, {\cal F}, {\mathbb P}\Big)$ carrying the sequence of RVs $\{(X_t, Y_t): t=0,1, \ldots, n\}$, as follows.
\begin{align}
 {\mathbb P}\big\{X^n \in d{x}^n, Y^n \in d{y}^n\big\}  \triangleq  &
{\bf P}^p(dx^n, dy^n), \hso n \in {\mathbb N}_0 \nonumber \\
= & \otimes_{t=0}^n \Big( { \bf P}(dy_t|y^{t-1}, x^t) \otimes {\bf P}(dx_t| x^{t-1}, y^{t-1})\Big) \label{CIS_2gde2new} \\
=&\otimes_{t=0}^n \Big(q_t(dy_t|y^{t-1}, x^t)\otimes p_t(dx_t|x^{t-1}, y^{t-1})\Big). \label{CIS_2gg_new} 
\end{align}
From the joint distribution, we can define the ${\cal Y}^{n}-$marginal distribution, and its conditional distribution\footnote{Throughout the paper the superscript notation ${\bf P}^p(\cdot), \nu_{0,n}^{p}(\cdot)$, etc., indicates the dependence of the distributions on the channel input conditional distribution.} as follows.
\begin{align}
{\mathbb  P}\big\{Y^n \in dy^n\big\} \triangleq \; & {\bf P}^{p}(dy^n) =  \int_{{\cal X}^{n}}  {\bf P}^{ p}(dx^n, dy^n) ,   \hso  n \in {\mathbb N}_0,  \label{CIS_3g}\\
\equiv \; &   \nu_{0,n}^{p}(dy^n) \label{MARGINAL} \\
\nu_t^{p}(dy_t|y^{t-1})= \; &  \int_{{\cal X}^{t}} q_t(dy_t|y^{t-1}, x^t)\otimes p_t(dx_t|x^{t-1}, y^{t-1}) \otimes {\bf P}^{p}(dx^{t-1}|y^{t-1}), \hso  t=0,1, \ldots, n . \label{CIS_3a}
\end{align}
The above joint distributions are parametrized by either a fixed $Y^{-1}=y^{-1} \in {\cal Y}^{-1}$ or a fixed distribution ${\bf P}_{Y^{-1}}(dy^{-1})=\mu(dy^{-1})$. \\
Directed information pay-off $I(X^n \rightarrow Y^n)$, is defined as follows. 
\begin{align}
I(X^n\rightarrow Y^n) \triangleq &\sum_{t=0}^n {\bf E}^{{p}} \Big\{  \log \Big( \frac{dq_t(\cdot|Y^{t-1},X^t) }{d\nu_t^{{p}}(\cdot|Y^{t-1})}(Y_t)\Big)\Big\}\label{CIS_6_a}   \\
=& \sum_{t=0}^n \int_{{\cal X}^{t} \times {\cal  Y}^{t}}\log \Big( \frac{ dq_t(\cdot|y^{t-1}, x^t) }{d\nu_t^{{p}}(\cdot|y^{t-1})}(y_t)\Big) {\bf P}^{p}( dx^t, dy^t). \label{CIS_6}
\end{align}

\noi {\it Our objective is the following.} Given a channel  distribution form the set ${\cal C}_{0,n}$, determine necessary and sufficient conditions for any channel input distribution of the set ${\cal P}_{0,n}$ (assuming it exists) to correspond to the maximizing element of the following extremum problem.
\begin{align}
C_{X^n \rightarrow Y^n}^{FB} \triangleq \sup_{{\cal P}_{0,n}} I(X^n\rightarrow Y^n). \label{prob2}
\end{align}
If a transmission cost constraint is imposed, then we replace \eqref{prob2} by 
 \begin{align}
 C_{X^n \rightarrow Y^n}^{FB}(\kappa) \triangleq \sup_{{\cal P}_{0,n}(\kappa)} I(X^n\rightarrow Y^n).  \label{prob2tc}
\end{align}

\noi Since our objective is to derive sufficient conditions in addition to necessary conditions, we invoke the following convexity results from \cite[Theorems III.2, III.3]{charalambous-stavrou2015ieeeit}.\\ 

\begin{lemma}(Convexity of Directed Information){\ \\}
\label{conv-DI}
(a) Any sequence of channel input conditional distributions from the set ${\cal P}_{0,n}$ and channel distributions from the set ${\cal C}_{0,n}$ uniquely define the following two $(n+1)$-fold compound causally conditioned probability distributions.\\
The family of distributions $\overleftarrow{P}(\cdot|y^{n-1})$ on ${\cal X}^n$ parametrized by $y^{n-1}\in{\cal Y}^{n-1}$ defined by
\begin{align}
\overleftarrow{P}_{0,n}(C|{y}^{n-1})&\triangleq\int_{C_0}p_0(dx_0|x^{-1},y^{-1})\ldots\int_{C_n}p_n(dx_n|x^{n-1},y^{n-1}),~C=\times_{t=0}^n{C}_t\in{\cal B}({\cal X}_{0,n})\label{section:2:equation2}
\end{align}
which is formally represented by
\begin{align}
{\overleftarrow {P}}_{0,n}(dx^n|y^{n-1}) &\triangleq \otimes_{t=0}^n {p}_{t}(dx_t|x^{t-1}, y^{y-1})\in{\cal M}({\cal X}^n)\label{cc_CI}
\end{align}
and similarly, the family of distributions $\overrightarrow{Q}(\cdot|x^{n})$ on ${\cal Y}^n$ parametrized by $x^{n}\in{\cal X}^{n}$, formally represented by
\begin{align}
\overrightarrow{Q}_{0,n}(dy^n|x^n) &\triangleq \otimes_{t=0}^n {q}_{t}(dy_t|y^{t-1}, x^{t})\in{\cal M}({\cal Y}^n)\label{cc_CI_1}
\end{align} 
and vice-versa. That is, \eqref{cc_CI}, \eqref{cc_CI_1} uniquely define any sequence of channel input distributions $\{q_t(dx_t|x^{t-1}, y^{t-1}): ~t=0,1, \ldots, n\} \in {\cal P}_{0,n}$ and channel distributions $\{ {q}_{t}(dy_t|y^{t-1}, x^{t})~:t=0,1, \ldots, n\}$, respectively. The joint distribution is equivalently expressed formally as ${\bf P}^p(x^n,y^n)=(\overleftarrow{P}_{0,n}\otimes\overrightarrow{Q}_{0,n})(x^n,y^n)$.\\
\noi(b) Directed information is equivalent to the following expression.
\begin{align}
I(X^n\rightarrow{Y}^n)=\int_{{\cal X}_{0,n} \times {\cal Y}_{0,n}} \log \Big( \frac{d{\overrightarrow Q}_{0,n}(\cdot|x^n)}{d\nu_{0,n}(\cdot)}(y^n)\Big)({\overleftarrow P}_{0,n}\otimes {\overrightarrow Q}_{0,n})(dx^n,dy^n)\equiv{\mathbb{I}}_{X^n\rightarrow{Y^n}}({\overleftarrow P}_{0,n}, {\overrightarrow Q}_{0,n})\label{eqdi5}
\end{align} 
where the notation ${\mathbb{I}}_{X^n\rightarrow{Y^n}}({\overleftarrow P}_{0,n}, {\overrightarrow Q}_{0,n})$ indicates the dependence of $I(X^n\rightarrow{Y}^n)$ on $\{\overleftarrow{P}_{0,n},\overrightarrow{Q}_{0,n}\}\in{\cal M}({\cal X}^n)\times{\cal M}({\cal Y}^n)$.\\
\noi(c) The set of conditional distributions ${\overleftarrow {P}}_{0,n}(\cdot|y^{n-1})\in {\cal M}({\cal X}^n)$ and $\overrightarrow{Q}_{0,n}(\cdot|x^n) \in {\cal M}({\cal Y}^n)$ are convex.\\
\noi(d) The functional ${\mathbb{I}}_{X^n\rightarrow{Y^n}}({\overleftarrow P}_{0,n}, {\overrightarrow Q}_{0,n})$ is concave with respect to ${\overleftarrow {P}}_{0,n}(\cdot|y^{n-1})\in {\cal M}({\cal X}^n)$ for a fixed $\overrightarrow{Q}_{0,n}(\cdot|x^n) \in {\cal M}({\cal Y}^n)$, and convex with respect to $\overrightarrow{Q}_{0,n}(\cdot|x^n) \in {\cal M}({\cal Y}^n)$ for a fixed ${\overleftarrow {P}}_{0,n}(\cdot|y^{n-1})\in {\cal M}({\cal X}^n)$. 
\end{lemma}

\noi In view of the convexity result stated in Lemma~\ref{conv-DI}, any extremum problem of feedback capacity is a convex optimization problem, and the following holds.

\begin{theorem}(Extremum problem of feedback capacity)\\
\label{thm-pr_fb}
Assume the set $ {\cal P}_{0,n}(\kappa)$ is nonempty and  the supremum in  (\ref{prob2tc}) is achieved in the set  $ {\cal P}_{0,n}(\kappa)$.\\
Then \\
\noi{(a) } $C_{X^n \rar Y^n}^{FB}(\kappa)$ is nondecreasing,  concave  function of $\kappa \in [0, \infty]$.\\
\noi {(b)} An alternative characterization of  $C_{X^n \rightarrow Y^n}^{FB}(\kappa)$ is given by
 \begin{align}
C_{X^n \rightarrow Y^n}^{FB}(\kappa)=\sup_{{\overleftarrow P}_{0,n}(dx^n|y^{n-1}):  {\frac{1}{n{+}1}} {\bf  E} \big\{ c_{0,n}(X^n,Y^{n{-}1}) \big\}  =  \kappa } {\mathbb{I}}_{X^n\rightarrow{Y^n}}({\overleftarrow P}_{0,n}, {\overrightarrow Q}_{0,n}),  \hst \mbox{for} \hso  \kappa \leq \kappa_{max}, 
 \end{align}
where $\kappa_{max}$  is the smallest number belonging to $[0,\infty]$ such that $C_{X^n \rar Y^n}^{FB}(\kappa)$ is constant in $[\kappa_{max}, \infty]$, and ${\bf E}\big\{\cdot \big\}$ denotes expectation with respect to $({\overleftarrow {P}}_{0,n}\otimes\overrightarrow{Q}_{0,n})(dx^n,dy^n)$.  
\end{theorem}
\noi Clearly, $\kappa_{max}$ is the value of $\kappa\in[0,\infty]$ for which $C_{X^n \rightarrow Y^n}^{FB}(\kappa)=C_{X^n \rightarrow Y^n}^{FB}$, i.e., it corresponds to the maximization of $I(X^n \rightarrow Y^n)$ over ${\cal P}_{0,n}$ (without transmission cost constraints).

\subsection{Variational Equality}\label{section:variational:equality}

\par Next, we recall a sequential variational equality of directed information, found in \cite[Section IV]{charalambous-stavrou2015ieeeit}, which is applied to derive necessary and sufficient conditions for extremum problems \eqref{prob2}, \eqref{prob2tc}.

\begin{theorem}\cite[Section IV]{charalambous-stavrou2015ieeeit}(Sequential variational equality of directed information)\label{theorem:variational_equality}{\ \\}
Given a channel input distribution $\big\{p_t(dx_t|x^{t-1},y^{t-1}): t=0, \ldots, n\big\}\in{\cal P}_{0,n}$ and channel distribution  $\big\{q_t(dy_t|y^{t-1},x^t) : t=0, \ldots, n\big\}\in{\cal C}_{0,n}$, let ${\bf P}^p(dx^n,dy^n) \in {\cal M}({\cal X}^{n} \times {\cal  Y}^{n})$, and ${\nu }_{0,n}^p(dy^n) \in {\cal M}({\cal  Y}^{n})$ denote their joint and marginal distributions defined by (\ref{CIS_2gde2new})-(\ref{CIS_3a}).\\
Let ${\cal S}_{0,n}\triangleq\big\{s_t(dy_t|y^{t-1},x^{t-1})\in{\cal M}({\cal Y}_t):~t\in\mathbb{N}_0^n\big\}$ and ${\cal R}_{0,n}\triangleq\big\{r_t(dx_t|x^{t-1},y^t)\in{\cal M}({\cal X}_t):~t\in\mathbb{N}_0^n\big\}$ be arbitrary distributions, and formally define the corresponding joint distribution by
\begin{align*}
\otimes_{t=0}^n\big(s_t(dy_t|y^{t-1},x^{t-1})\otimes {r}_t(dx_t|x^{t-1},y^t)\big)\in{\cal M}({{\cal X}^{n} }\times {\cal Y}^{n}).
\end{align*}
Then the following variational equality holds. 
\begin{align}
I(X^n\rightarrow{Y}^n) =&\sup_{{\cal S}_{0,n}\otimes{\cal R}_{0,n}} \sum^n_{t=0}\int_{{{\cal X}^{t} }\times {\cal Y}^{t}}\log\Bigg(\frac{d{r}_t(\cdot|x^{t-1},y^{t})}{dp_t(\cdot|x^{t-1},y^{t-1})}(x_t)\frac{ds_t(\cdot|y^{t-1},x^{t-1})}{d\nu_{t}^p(\cdot|y^{t-1})}(y_t)\Bigg){\bf P}^p(dx^t, dy^t)   \label{equation15a}
\end{align}
and the supremum in (\ref{equation15a}) is achieved when the following identity holds. 
\begin{align}
\frac{dp_t(\cdot|x^{t-1},y^{t-1})}{d{r}_t(\cdot|x^{t-1},y^t)}(x_t).\frac{d{q}_t(\cdot|y^{t-1},x^{t})}{ds_t(\cdot|y^{t-1},x^{t-1})}(y_t)=1-a.a.~(x^t,y^t),~t\in\mathbb{N}_0^n.\label{equation102}
\end{align}
Equivalently, the supremum in (\ref{equation15a}) is achieved at 
\begin{align}
\otimes_{t=0}^n\Big(s_t(dy_t|y^{t-1},x^{t-1})\otimes {r}_t(dx_t|x^{t-1},y^t)\Big)= {\bf P}^p(dx^n, dy^n).\nonumber
\end{align}
\end{theorem}

\noi To avoid excessive technical issues, we derive the main results of this paper by restricting our attention to finite alphabet spaces $\{({\cal X}_t,{\cal Y}_t):~t=0,1,\ldots\}$. This means that we replace distributions by probability mass functions, and integrals by sums, i.e., $q_t(dy_t|y^{t-1},x^t)\longmapsto{q}_t(y_t|y^{t-1},x^t),~ p_t(dx_t|x^{t-1},y^{t-1})\longmapsto{p}_t(x_t|x^{t-1},y^{t-1})$. However, in Section~\ref{generalizations:abstract:alphabets}, we give sufficient conditions for the results derived for finite alphabet spaces to extend to abstract alphabet spaces (i.e., countable and continuous).

%
%
%
%

\section{Necessary and Sufficient Conditions for Channels of Class $A$ with Transmission Cost of Class $A$}\label{section:generalizations:lmco}

\par Consider the finite alphabet version of channel distributions of class $A$ given by \eqref{introduction:class:channel:equation1}, and a transmission cost function of class $A$ given by \eqref{TC_1}. By \cite{kourtellaris-charalambous2015aieeeit}, the characterization of FTFI capacity with average transmission cost constraint is given by 
\begin{align}
{C}_{X^n \rightarrow Y^n}^{FB,A.J}(\kappa) 
= \sup_{{\cal P}_{0,n}^{A.J}(\kappa)} \sum_{t=0}^n {\bf E}^{ \pi}\left\{
\log\Big(\frac{q_t(Y_t|Y_{t-M}^{t-1},X_t)}{\nu_t^{{\pi}}(Y_t|Y_{t-J}^{t-1})}\Big)
\right\},~J=\max\{M,N\} \label{characterization:ftfi:capacity:cost:lmco1}
\end{align}
where
\begin{align}
{\cal P}_{0,n}^{A.J}(\kappa)\triangleq\Big\{\pi_t(x_t|y_{t-J}^{t-1}), ~t=0, 1, \ldots, n: \frac{1}{n+1} {\bf E}^{\pi} \Big( c^{A.N}_{0,n}(X^n, Y^{n-1}) \Big)\leq  \kappa\Big\},~ \kappa \in [0,\infty)\label{characterization:ftfi:fmco:transmission:cost}
\end{align}
and the joint and transition probabilities are given by
\begin{align}
{\bf P}^{\pi}(y^t, x^t) =&\prod_{i=0}^tq_i(y_i|y_{i-M}^{i-1}, x_i)\pi_i(x_i|y_{i-J}^{i-1}), \label{section:lmco:equation1} \\
\nu_t^{\pi}(y_t|y_{t-J}^{t-1}) =&\sum_{x_t\in{\cal X}_t} q_t(y_t|y_{t-M}^{t-1}, x_t)\pi_t(x_t|y_{t-J}^{t-1}),~t\in\mathbb{N}_0^{n}.\label{section:umco:equation2}
\end{align}
 
\noi In this section, we utilize the characterization of FTFI given by \eqref{characterization:ftfi:capacity:cost:lmco1}, to derive the {\it sequential necessary and sufficient conditions} for any ${\cal P}^{A.J}_{0,n}(\kappa)$ to achieve ${C}_{X^n \rightarrow Y^n}^{FB,A.J}(\kappa)$.

\noi Since we have assumed all spaces $\{({\cal X}_t, {\cal Y}_t):~t\in\mathbb{N}_0^n\}$ have finite cardinality, in the subsequent analysis we use the preliminary results of Section~\ref{baa:sequential:section:problem:formulation}, with distributions replaced by probability mass functions \big(as defined in \eqref{characterization:ftfi:capacity:cost:lmco1}-\eqref{section:umco:equation2}\big).

%
%
%
%

\subsection{Sequential Necessary and Sufficient Conditions}\label{subsubsection:ftfi:lmco:cost}
 
\noi For any $\{\pi_t(x_t|y_{t-J}^{t-1}):~t\in\mathbb{N}_0^{n}\}$, let $C^\pi_t: {\cal Y}_{t-J}^{t-1}\longmapsto[0,\infty)$ represent the expected total pay-off corresponding to (\ref{characterization:ftfi:capacity:cost:lmco1}), without the maximization, on the future time horizon $\{t,t+1,\ldots,n\}$, given $Y^{t-1}_{t-J}=y^{t-1}_{t-J}$ at time $t-1$, defined by
\begin{align}
C^{\pi}_t(y_{t-J}^{t-1})={\bf E}^{\pi}\bigg\{\sum_{i=t}^n\log\Big(\frac{q_i(Y_i|y_{i-M}^{i-1},X_i)}{{\nu}^{\pi}_{i}(Y_i|y_{i-J}^{i-1})}\Big)\Big{|}Y_{t-J}^{t-1}=y_{t-J}^{t-1}\bigg\},~t\in\mathbb{N}_0^n,~\forall{y^{t-1}_{t-J}\in{\cal Y}^{t-1}_{t-J}}.\label{payoff:without:maximization:ftfi:capacity:lmco:equation1}
\end{align}
By invoking Theorem~\ref{theorem:variational_equality}, we can express \eqref{payoff:without:maximization:ftfi:capacity:lmco:equation1} as a variational problem as follows.
\begin{corollary}\label{corollary:application:variational:equalities:lmco}  {\ \\}
Consider the cost-to-go $C^\pi_t(y^{t-1}_{t-J})$, $t\in\mathbb{N}_0^{n}$, $y^{t-1}_{t-J}\in{\cal Y}^{t-1}_{t-J}$, defined by \eqref{payoff:without:maximization:ftfi:capacity:lmco:equation1}.\\
(a) The cost-to-go $C^\pi_t(y^{t-1}_{t-J})$, is the solution of the extremum problem
\begin{align}
C^{\pi}_t(y_{t-J}^{t-1})=&\sup_{\big\{r_i(x_i|y_{i-M}^{i-1},y_i):~i=t,t+1,\ldots,n\big\}}{\bf E}^{\pi}\bigg\{\sum_{i=t}^n\log\Big(\frac{r_i(X_i|y_{i-M}^{i-1},Y_i)}{{\pi}_{i}(X_i|y_{i-J}^{i-1})}\Big)\Big{|}Y_{t-J}^{t-1}=y_{t-J}^{t-1}\bigg\},~t\in\mathbb{N}_0^n\label{applications:variational:equalities:lmco:equation1}
\end{align}
and moreover, the supremum is achieved at
\begin{align}
r^{\pi}_t(x_t|y_{t-M}^{t-1},y_t)=\Big(\frac{q_t(y_t|y_{t-M}^{t-1},x_t)}{\nu^{\pi}_{t}(y_t|y_{t-J}^{t-1})}\Big){\pi}_{t}(x_t|y_{t-J}^{t-1}),~t\in\mathbb{N}_0^n.\label{applications:variational:equalities:lmco:equation2}
\end{align}
(b) The cost-to-go $C^\pi_t(y^{t-1}_{t-J})$, satisfies the following dynamic programming recursions\footnote{For the rest of the paper we use the notation $\sum_{x_t}(\cdot)\equiv\sum_{x_t\in{\cal X}_t}(\cdot)$}. 
\begin{align}
C^{\pi}_n(y_{n-J}^{n-1})&=\sup_{r_n(x_n|y_{n-M}^{n-1},y_n)}\sum_{x_n,y_n}\log\Big(\frac{r_n(x_n|y_{n-M}^{n-1},y_n)}{{\pi}_{n}(x_n|y_{n-J}^{n-1})}\Big)q_n(y_n|y_{n-J}^{n-1},x_n){\pi}_n(x_n|y_{n-J}^{n-1}),~~\forall{y^{n-1}_{n-J}\in{\cal Y}^{n-1}_{n-J}},\label{baa:sequential:applications:lmco:equation10aa}\\
C^{\pi}_t(y_{t-J}^{t-1})&=\sup_{r_n(x_t|y_{t-M}^{t-1},y_t)}\sum_{x_t,y_t}\Big(\log\Big(\frac{r_t(x_t|y_{t-M}^{t-1},y_t)}{{\pi}_{t}(x_t|y_{t-J}^{t-1})}\Big)\nonumber\\
&+C^{\pi}_{t+1}(y_{t+1-J}^t)\Big)q_t(y_t|y_{t-M}^{t-1},x_t){\pi}_t(x_t|y_{t-J}^{t-1}),~t\in\mathbb{N}^{n-1}_{0},~\forall{y^{t-1}_{t-J}\in{\cal Y}^{t-1}_{t-J}}\label{baa:sequential:applications:lmco:equation11aa}
\end{align}
and moreover, the supremum in \eqref{baa:sequential:applications:lmco:equation10aa}, \eqref{baa:sequential:applications:lmco:equation11aa} is achieved at
\eqref{applications:variational:equalities:lmco:equation2}.
\end{corollary}
\begin{proof}
(a) This follows from \cite[Section IV.1]{charalambous-stavrou2015ieeeit} by repeating the derivation if necessary. (b) This follows from dynamic programming \cite{bertsekas-shreve2007,vanschuppen2014} and (a).
\end{proof}

\noi Corollary~\ref{corollary:application:variational:equalities:lmco} illustrates that the variational equality of Theorem~\ref{theorem:variational_equality}, as expected, also holds for a running pay-off over an interval $\{t,t+1,\ldots,n\}$ conditioned on $Y^{t-1}_{t-J}=y^{t-1}_{t-J}$ at time ${t-1}$. Moreover, it is obvious that the functional $C^{\pi}_t(y_{t-J}^{t-1})\equiv\mathbb{C}_t^{\pi}(r_t,r_{t+1},\ldots,r_n;y_{t-J}^{t-1})$ over which the supremum is taken in \eqref{applications:variational:equalities:lmco:equation1}, defined by
\begin{align*}
\mathbb{C}_t^{\pi}(r_t,r_{t+1},\ldots,r_n;y_{t-J}^{t-1})\triangleq{\bf E}^{\pi}\bigg\{\sum_{i=t}^n\log\Big(\frac{r_i(X_i|y_{i-M}^{i-1},Y_i)}{{\pi}_{i}(X_i|y_{i-J}^{i-1})}\Big)\Big{|}Y_{t-J}^{t-1}=y_{t-J}^{t-1}\bigg\},~t\in\mathbb{N}_0^n
\end{align*}
is concave in $\{r_t(x_t|y^{t-1}_{t-M}),\ldots,r_n(x_n|y^{n-1}_{n-M})\}\in{\cal M}({\cal X}_t)\times\ldots\times{\cal M}({\cal X}_n)$.\\

\noi Next, we introduce the dynamic programming recursions, when \eqref{payoff:without:maximization:ftfi:capacity:lmco:equation1} is maximized over channel input distributions  from the set ${\cal P}^{A.J}_{0,n}(\kappa)$.\\
\noi Throughout this section, we assume existence of an interior point of the constraint set ${\cal P}^{A.J}_{0,n}(\kappa)$ and existence of an optimal channel input distribution which maximizes ${C}_{X^n \rightarrow Y^n}^{FB,A.J}(\kappa)$. Hence, in view of the convexity of optimization problem \eqref{characterization:ftfi:capacity:cost:lmco1}, we can apply Lagrange Duality Theorem (see \cite{dluenberger1969}) to convert the problem into an unconstrained optimization problem over the space of probability distributions $\{\pi(x_t|y^{t-1}_{t-J})\in{\cal M}({\cal X}_n):~t\in\mathbb{N}_0^n\}$.\\

\noi Let $C_t: {\cal Y}^{t-1}_{t-J}\longmapsto[0,\infty)$ represent the maximum expected total pay-off in (\ref{characterization:ftfi:capacity:cost:lmco1}) on the future time horizon $\{t,t+1,\ldots,n\}$, given $Y^{t-1}_{t-J}=y^{t-1}_{t-J}$ at time $t-1$, defined by 
\begin{align}
&C_t(y^{t-1}_{t-J})=\sup_{\big\{\pi_i(x_i|y^{i-1}_{i-J}):~i=t,t+1,\ldots,n\big\}}{\bf E}^{\pi}\bigg\{\sum_{i=t}^n\log\Big(\frac{q_i(Y_i|y^{i-1}_{i-M},X_i)}{{\nu}^{\pi}_{i}(Y_i|y^{i-1}_{i-J})}\Big)\nonumber\\
&\qquad-s\Big(\sum_{i=t}^n\gamma_i(x_i,y^{i-1}_{i-N})-(n+1)\kappa\Big)\Big{|}Y^{t-1}_{t-J}=y^{t-1}_{t-J}\bigg\}\label{algorithms:generalizations:lmco:equation17}\\
&\stackrel{(*)}\equiv\sup_{\big\{\pi_i(x_i|y^{i-1}_{i-J}):~i=t,t+1,\ldots,n\big\}}\bigg\{C_t^{\pi}(y^{t-1}_{t-J})-s\Big({\bf E}^{\pi}\Big\{\sum_{i=t}^n\gamma_i(x_i,y^{i-1}_{i-N})\Big{|}Y^{t-1}_{t-J}=y^{t-1}_{t-J}\Big\}-(n+1)\kappa\Big)\bigg\}\label{algorithms:generalizations:lmco:equation17aaa}
\end{align}
where $(*)$ follows from Corollary~\ref{corollary:application:variational:equalities:lmco}, and $s\geq{0}$ is the Lagrange multiplier associated with the constraint.\\ 
\noi By standard dynamic programming arguments \cite{bertsekas-shreve2007,vanschuppen2014}, it follows that (\ref{algorithms:generalizations:lmco:equation17}) satisfies the following dynamic programming recursions.
\begin{align}
&C_n(y^{n-1}_{n-J})=\sup_{\pi_n(x_n|y^{t-1}_{t-J})}\Bigg\{\sum_{x_n,y_n}\log\Big(\frac{q_n(y_n|y^{n-1}_{n-M},x_n)}{{\nu}^{\pi}_{n}(y_n|y^{n-1}_{n-J})}\Big)q_n(y_n|y^{n-1}_{n-M},x_n){\pi}_n(x_n|y^{n-1}_{n-J})\nonumber\\
&-s\Big(\sum_{x_n}\gamma_n(x_n,y^{n-1}_{n-N})\pi_n(x_n|y^{n-1}_{n-J})-(n+1)\kappa\Big)\Bigg\},\label{algorithms:generalizations:lmco:equation18}\\
&C_t(y^{t-1}_{t-J})=\sup_{\pi_t(x_t|y^{t-1}_{t-J})}\Bigg\{\sum_{x_t,y_t}\Big(\log\big(\frac{q_t(y_t|y^{t-1}_{t-M},x_t)}{{\nu}^{\pi}_{t}(y_t|y^{t-1}_{t-J})}\big)+C_{t+1}(y^t_{t+1-J})\Big)\nonumber\\
&\qquad{q}_t(y_t|y^{t-1}_{t-M},x_t){\pi}_t(x_t|y^{t-1}_{t-J})-s\Big(\sum_{x_t}\gamma_t(x_t,y^{t-1}_{t-N})\pi_t(x_t|y^{t-1}_{t-J})-(n+1)\kappa\Big)\Bigg\},~t\in\mathbb{N}_0^{n-1}.\label{algorithms:generalizations:lmco:equation19}
\end{align}

\noi Next, we apply variational equality \eqref{equation15a} to show that the supremum in (\ref{algorithms:generalizations:lmco:equation18}), (\ref{algorithms:generalizations:lmco:equation19}), can be expressed as an extremum problem involving a double maximization problem over specific sets of distributions.

\begin{theorem}(Sequential double maximization with transmission cost)\label{theorem:generalization:ftfi:lmco:cost:double maximization}{\ \\}
Consider the sequence of channel distributions ${\cal C}^{A.M}_{0,n}\triangleq\{q_t(y_t|y^{t-1}_{t-M},x_t):~t\in\mathbb{N}_0^{n}\}$, and $C^{FB,A.J}_{X^n\rightarrow{Y^n}}(\kappa)$ defined by (\ref{characterization:ftfi:capacity:cost:lmco1}), for a fixed ${\mu}(y^{-1}_{-J})$. Assume there exist interior point to the constraint set ${\cal P}^{A.J}_{0,n}(\kappa)$. Then the following hold.\\
\noi{(a)} The dynamic programming recursions (\ref{algorithms:generalizations:lmco:equation18}), (\ref{algorithms:generalizations:lmco:equation19}) are equivalent to the following sequential double maximization dynamic programming recursions.
\begin{align}
C_n(y^{n-1}_{n-J})&=\sup_{\pi_n(x_n|y^{n-1}_{n-J})}\sup_{r_n(x_n|y^{n-1}_{n-M},y_n)}\Bigg\{\sum_{x_n,y_n}\log\Big(\frac{r_n(x_n|y^{n-1}_{n-M},y_n)}{{\pi}_{n}(x_n|y^{n-1}_{n-J})}\Big){q}_n(y_n|y^{n-1}_{n-M},x_n)\pi_n(x_n|y^{n-1}_{n-J})\nonumber\\
&-s\Big(\sum_{x_n}\gamma_n(x_n,y^{n-1}_{n-N})\pi_n(x_n|y^{n-1}_{n-J})-(n+1)\kappa\Big)\Bigg\},\label{algorithms:generalizations:lmco:equation21}\\
C_t(y^{t-1}_{t-J})&=\sup_{\pi_t(x_t|y^{t-1}_{t-J})}\sup_{r_t(x_t|y^{t-1}_{t-M},y_t)}\Bigg\{\sum_{x_t,y_t}\Big(\log\Big(\frac{r_t(x_t|y^{t-1}_{t-M},y_t)}{{\pi}_{t}(x_t|y^{t-1}_{t-J})}\Big)+C_{t+1}(y^t_{t+1-J})\Big)q_t(y_t|y^{t-1}_{t-M},x_t){\pi}_t(x_t|y^{t-1}_{t-J})\nonumber\\
&-s\Big(\sum_{x_t}\gamma_t(x_t,y^{t-1}_{t-N})\pi_t(x_t|y^{t-1}_{t-J})-(n+1)\kappa\Big)\Bigg\},~t\in\mathbb{N}_0^{n-1}\label{algorithms:generalizations:lmco:equation22}
\end{align}
and $C^{FB, A.J}_{X^n\rightarrow{Y}^n}(\kappa)$ is given by
\begin{align}
C^{FB, A.J}_{X^n\rightarrow{Y^n}}(\kappa)=\inf_{s\geq{0}}\sum_{y^{-1}_{-J}}C_0(y^{-1}_{-J})\mu(y^{-1}_{-J}).\label{algorithms:generalizations:lmco:equation20}
\end{align}
In addition, the following hold.\\
\noi {(i)} For a fixed $\pi_n(x_n|y^{n-1}_{n-J})$, the maximum in (\ref{algorithms:generalizations:lmco:equation21}) over $r_n(x_n|y^{n-1}_{n-M},y_n)$ occurs at $r^{*,\pi}_n(x_n|y^{n-1}_{n-M},y_n)$ given by
\begin{align}
r^{*,\pi}_n(x_n|y^{n-1}_{n-M},y_n)=\Big(\frac{q_n(y_n|y^{n-1}_{n-M},x_n)}{\nu^{\pi}_{n}(y_n|y^{n-1}_{n-J})}\Big){\pi}_{n}(x_n|y^{n-1}_{n-J})
\label{algorithms:generalizations:lmco:equation25aa}
\end{align}
and for a fixed $r_n(x_n|y^{n-1}_{n-M},y_n)$, the maximum in (\ref{algorithms:generalizations:lmco:equation21}) over $\pi_n(x_n|y^{n-1}_{n-J})$ is given by
\begin{align}
\pi^{}_n(x_n|y^{n-1}_{n-J})&=\frac{\exp{\Big\{\sum_{y_n}\log\big(r_n(x_n|y^{n-1}_{n-M},y_n)\big)q_n(y_n|y^{n-1}_{n-M},x_n)-s\gamma_n(x_n,y^{n-1}_{n-N})\Big\}}}{\sum_{x_n}\exp{\Big\{\sum_{y_n}\log\big(r_n(x_n|y^{n-1}_{n-M},y_n)\big)q_n(y_n|y^{n-1}_{n-M},x_n)-s\gamma_n(x_n,y^{n-1}_{n-N})\Big\}}},~\forall{x_n}\in{\cal X}_n.\label{algorithms:generalizations:lmco:equation26}
\end{align}
\noi{(ii)} For a fixed $\pi_t(x_t|y^{t-1}_{t-J})$, the maximum in (\ref{algorithms:generalizations:lmco:equation22}) over $r_t(x_t|y^{t-1}_{t-M},y_t)$ occurs at $r^{*,\pi}_t(x_t|y^{t-1}_{t-M},y_t)$ given by
\begin{align}
r^{*,\pi}_t(x_t|y^{t-1}_{t-M},y_t)=\Big(\frac{q_t(y_t|y^{t-1}_{t-M},x_t)}{\nu^{\pi}_{t}(y_t|y^{t-1}_{t-J})}\Big){\pi}_{t}(x_t|y^{t-1}_{t-J}),~t\in\mathbb{N}_0^{n-1}
\label{algorithms:generalizations:lmco:equation28}
\end{align}
and for a fixed $r_t(x_t|y^{t-1}_{t-M},y_t)$, the maximum in (\ref{algorithms:generalizations:lmco:equation22}) over $\pi_t(x_t|y^{t-1}_{t-J})$ is given by
\begin{align}
&\pi^{}_t(x_t|y^{t-1}_{t-J})=\frac{\exp{\Big\{\sum_{y_t}\Big(\log\big(r_t(x_t|y^{t-1}_{t-M},y_t)\big)+C_{t+1}(y^{t}_{t+1-J})\Big)q_t(y_t|y^{t-1}_{t-M},x_t)-s\gamma_t(x_t,y^{t-1}_{t-N})\Big\}}}{\sum_{x_t}\exp{\Big\{\sum_{y_t}\Big(\log\Big(r_t(x_t|y^{t-1}_{t-M},y_t)\Big)+C_{t+1}(y^t_{t+1-J})\Big)q_t(y_t|y^{t-1}_{t-M},x_t)-s\gamma_t(x_t,y^{t-1}_{t-N})\Big\}}},\nonumber\\
&\hspace{9.5cm}~\forall{x_t}\in{\cal X}_t,~t\in\mathbb{N}_0^{n-1}.\label{algorithms:generalizations:lmco:equation29}
\end{align}
\noi(iii) When (\ref{algorithms:generalizations:lmco:equation26}) is evaluated at $r_n(\cdot|\cdot,\cdot)=r^{*,\pi}_n(\cdot|\cdot,\cdot)$ given by (\ref{algorithms:generalizations:lmco:equation25aa}) then
\begin{align}
\pi^{}_n(x_n|y^{n-1}_{n-J})=&\frac{\exp{\Big\{\sum_{y_n}\log\big(\frac{q_n(y_n|y^{n-1}_{n-M},x_n)}{\nu^{\pi}_{n}(y_n|y^{n-1}_{n-J})}\big)q_n(y_n|y^{n-1}_{n-M},x_n)-s\gamma_n(x_n,y^{n-1}_{n-N})\Big\}}\pi_n(x_n|y^{n-1}_{n-J})}{\sum_{x_n}\exp{\Big\{\sum_{y_n}\log\big(\frac{q_n(y_n|y^{n-1}_{n-M},x_n)}{\nu^{\pi}_{n}(y_n|y^{n-1}_{n-J})}\big)q_n(y_n|y^{n-1}_{n-M},x_n)-s\gamma_n(x_n,y^{n-1}_{n-N})\Big\}}\pi_n(x_n|y^{n-1}_{n-J})},\nonumber\\
&\hspace*{10cm}~\forall{x_n}\in{\cal X}_n.\label{algorithms:generalizations:lmco:equation27}
\end{align}
When (\ref{algorithms:generalizations:lmco:equation29}) is evaluated at $r^{*,\pi}_t(x_t|y^{t-1}_{t-M},y_t)=r_t(\cdot|\cdot,\cdot)$ given by (\ref{algorithms:generalizations:lmco:equation28}) then
\begin{align}
&\pi^{}_t(x_t|y^{t-1}_{t-J})\nonumber\\
&=\frac{\exp{\Big\{\sum_{y_t}\Big(\log\big(\frac{q_t(y_t|y^{t-1}_{t-M},x_t)}{\nu^{\pi}_{t}(y_t|y^{t-1}_{t-J})}\big)+C_{t+1}(y^{t}_{t+1-J})\Big)q_t(y_t|y^{t-1}_{t-M},x_t)-s\gamma_t(x_t,y^{t-1}_{t-N})\Big\}}\pi_t(x_t|y^{t-1}_{t-J})}{\sum_{x_t}\exp{\Big\{\sum_{y_t}\Big(\log\big(\frac{q_t(y_t|y^{t-1}_{t-M},x_t)}{\nu^{\pi}_{t}(y_t|y^{t-1}_{t-J})}\big)+C_{t+1}(y^t_{t+1-J})\Big)q_t(y_t|y^{t-1}_{t-M},x_t)-s\gamma_t(x_t,y^{t-1}_{t-N})\Big\}}\pi_t(x_t|y^{t-1}_{t-J})},\nonumber\\
&\hspace{9.5cm}~\forall{x_t}\in{\cal X}_t,~t\in\mathbb{N}_0^{n-1}.\label{algorithms:generalizations:lmco:equation30}
\end{align}
\noi{(b)} The extremum problem $C^{FB,A.J}_{X^n\rightarrow{Y^n}}(\kappa)$ defined by (\ref{characterization:ftfi:capacity:cost:lmco1}) is equivalent to the following sequential double maximization problem.
\begin{align}
&C^{FB,A.J}_{X^n\rightarrow{Y}^n}(\kappa)=\inf_{s\geq{0}}\sup_{\pi_0(x_0|y^{-1}_J)}\sup_{r_0(x_0|y^{-1}_M,y_0)}\ldots\sup_{\pi_n(x_n|y^{n-1}_{n-J})}\sup_{r_n(x_n|y^{n-1}_{n-M},y_n)}\sum_{t=0}^n\Bigg\{{\bf E}\Big\{\log\Big(\frac{r_t(x_t|y^{t-1}_{t-M},y_t)}{\pi_{t}(x_t|y^{t-1}_{t-J})}\Big)\Big\}\nonumber\\
&\hspace{5cm}-s\Big({\bf E}\big\{\gamma_t(x_t,y^{t-1}_{t-N})\big\}-(n+1)\kappa\Big)\Bigg\}.\label{algorithms:generalizations:lmco:equation31}
\end{align}
\end{theorem}
\begin{proof}
The derivation is given in Appendix~\ref{proof:theorem:generalization:ftfi:lmco:cost:double maximization}.
\end{proof}

\noi In the next remark, we make some observations regarding Theorem~\ref{theorem:generalization:ftfi:lmco:cost:double maximization}.

\begin{remark}(Comments on Theorem~\ref{theorem:generalization:ftfi:lmco:cost:double maximization})
\begin{itemize}
\item[(a)] Theorem~\ref{theorem:generalization:ftfi:lmco:cost:double maximization} is a sequential version of the one derived for DMC in \cite[Theorem 8]{blahut1972}, which is crucial for the development of Blahut-Arimoto algorithm, to compute channel capacity of memoryless channels with transmission cost. That is, if we degrade the channel to a memoryless channel, and the transmission cost function to $\gamma_t(x_t,y^{t-1})\equiv\bar{\gamma}(x_t)$,~$t\in\mathbb{N}_0^n$, then Theorem~\ref{theorem:generalization:ftfi:lmco:cost:double maximization} is precisely \cite[Theorem 8]{blahut1972}. However, unlike \cite[Theorem 8]{blahut1972}, since the channel in our case is not memoryless, all equations involve the cost-to-go or value function. 
\item[(b)] The optimal channel input distribution satisfies the implicit nonlinear recursive equations \eqref{algorithms:generalizations:lmco:equation27}, \eqref{algorithms:generalizations:lmco:equation30}. These can be used to develop sequential algorithms to compute feedback capacity of channels with memory, with and without transmission cost constraint.
\end{itemize}
\end{remark}

\noi Next, we derive necessary and sufficient conditions for any input distribution $\{\pi_t(x_t|y_{t-J}^{t-1})\in{\cal M}({\cal X}_t):~t\in\mathbb{N}_0^{n}\}$ to achieve the supremum of the characterization of FTFI capacity with transmission cost given by \eqref{characterization:ftfi:capacity:cost:lmco1}. We obtain these conditions using two different methods. The first method is based on Theorem~\ref{theorem:generalization:ftfi:lmco:cost:double maximization}, while the second method is based on maximizing directly (\ref{algorithms:generalizations:lmco:equation18}), (\ref{algorithms:generalizations:lmco:equation19}). The derivation applies Karush-Kuhn-Tucker (KKT) theorem (see \cite{boyd-vandenberghe2004}), in view of the convexity of the optimization problems (\ref{algorithms:generalizations:lmco:equation18}), (\ref{algorithms:generalizations:lmco:equation19}) over the space of channel input distributions.

\begin{theorem}(Sequential necessary and sufficient conditions)\label{theorem:necessary:sufficient:ftfi:lmco:cost}{\ \\}
The necessary and sufficient conditions for any input distribution $\{\pi_{t}(x_t|y^{t-1}_{t-J}):~t\in\mathbb{N}_0^{n}\}$,~$J=\max\{M,N\}$, to achieve the supremum in $C^{FB,A.J}_{X^n\rightarrow{Y}^n}(\kappa)$  given by \eqref{characterization:ftfi:capacity:cost:lmco1} are the following.\\
\noi{(a)}  For each $y^{n-1}_{n-J}\in{\cal Y}^{n-1}_{n-J}$, there exist a ${K}^s_n(y^{n-1}_{n-J})$, which depends on $s\geq{0}$, such that the following hold.
\begin{align}
&\sum_{y_{n}}\Big(\log\big(\frac{q_n(y_n|y^{n-1}_{n-M},x_n)}{\nu^{\pi}_{t}(y_n|y^{n-1}_{n-J})}\big)\Big)q_n(y_n|y^{n-1}_{n-M},x_n)-s\gamma_n(x_n,y^{n-1}_{n-N})={K}^s_n(y^{n-1}_{n-J}),~\forall{x_n},~\mbox{if}~\pi_n(x_n|y^{n-1}_{n-J})\neq{0},\label{theorem:ftfi:lmco:cost:necessary:sufficient:equation1}\\
&\sum_{y_{n}}\Big(\log\big(\frac{q_n(y_n|y^{n-1}_{n-M},x_n)}{\nu^{\pi}_{n}(y_n|y^{n-1}_{n-J})}\big)\Big)q_n(y_n|y^{n-1}_{n-M},x_n)-s\gamma_n(x_n,y^{n-1}_{n-N})\leq{K}^s_n(y^{n-1}_{n-J}),~\forall{x_n},~\mbox{if}~\pi_n(x_n|y^{n-1}_{n-J})={0}.\label{theorem:ftfi:lmco:cost:necessary:sufficient:equation2}
\end{align}
Moreover, $C_t(y^{t-1}_{t-J})={K}^s_n(y^{n-1}_{n-J})+s(n+1)\kappa$ corresponds to the value function $C_t(y^{t-1}_{t-J})$, defined by \eqref{algorithms:generalizations:lmco:equation17}, evaluated at $t=n$.\\
\noi{(b)} For each $t$, $y^{t-1}_{t-J}\in{\cal Y}^{t-1}_{t-J}$, there exist a ${K}^s_t(y^{t-1}_{t-J})$, which depends on $s\geq{0}$, such that the following hold. \begin{align}
&\sum_{y_{t}}\Big(\log\big(\frac{q_t(y_t|y^{t-1}_{t-M},x_t)}{\nu^{\pi}_{t}(y_t|y^{t-1}_{t-J})}\big)+K^s_{t+1}(y^{t}_{t+1-J})\Big)q_t(y_t|y^{t-1}_{t-M},x_t)\nonumber\\
&\hspace{5cm}-s\gamma_t(x_t,y^{t-1}_{t-N})=K^s_t(y^{t-1}_{t-J}),~\forall{x_t},~\mbox{if}~\pi_t(x_t|y^{t-1}_{t-J})\neq{0},\label{theorem:ftfi:lmco:cost:necessary:sufficient:equation3}\\
&\sum_{y_{t}}\Big(\log\big(\frac{q_t(y_t|y^{t-1}_{t-M},x_t)}{\nu^{\pi}_{t}(y_t|y^{t-1}_{t-J})}\big)+K^s_{t+1}(y^{t}_{t+1-J})\Big)q_t(y_t|y^{t-1}_{t-M},x_t)\nonumber\\
&\hspace{5cm}-s\gamma_t(x_t,y^{t-1}_{t-N})\leq{K}^s_t(y^{t-1}_{t-J}),~\forall{x_t},~\mbox{if}~\pi_t(x_t|y^{t-1}_{t-J})={0}\label{theorem:ftfi:lmco:cost:necessary:sufficient:equation4}
\end{align}
for $t=n-1,\ldots,0$. Moreover, $C_t(y^{t-1}_{t-J})={K}^s_t(y^{t-1}_{t-J})+s(n+1)\kappa$ corresponds to the value function $C_t(y^{t-1}_{t-J})$, defined by \eqref{algorithms:generalizations:lmco:equation17}, evaluated at $t=n-1,\ldots,0$.
\end{theorem}
\begin{proof}
See Appendix~\ref{proof:theorem:necessary:sufficient:ftfi:lmco:cost}.
\end{proof}

\noi Before we proceed, we make the following comments about Theorem~\ref{theorem:necessary:sufficient:ftfi:lmco:cost}.

\begin{remark}(Comments on Theorem~\ref{theorem:necessary:sufficient:ftfi:lmco:cost})
\begin{itemize}
\item[(a)] An alternative derivation of Theorem~\ref{theorem:necessary:sufficient:ftfi:lmco:cost} based on Theorem~\ref{theorem:generalization:ftfi:lmco:cost:double maximization} is given in Appendix~\ref{proofs:section:ftfi:lmco},~Remark~\ref{remark:necessary:sufficient:alternative:proof}.
\item[(b)] Theorem~\ref{theorem:necessary:sufficient:ftfi:lmco:cost}  degenerates to Theorem~\ref{introduction:baa:sequential:theorem:lmco:necessary:sufficient} given in Section~\ref{section:necessary:sufficient:introduction} if there is no transmission cost constraint. 
\item[(c)] The sequential necessary and sufficient conditions derived in Theorem~\ref{theorem:necessary:sufficient:ftfi:lmco:cost}  are important for the following reasons.\\
{(i)} They characterize explicitly any input distribution that achieves the supremum of the characterization of FTFI capacity, in extremum problems of feedback capacity of channels with finite memory with and without transmission cost.\\
{(ii)} They can be used to develop sequential algorithms to facilitate numerical evaluation of feedback capacity problems \cite{stavrou-charalambous-tzortzis2016ieeeit}. 
\end{itemize}
\end{remark}

\noi Chen and Berger in the seminal paper \cite{chen-berger2005}, gave sufficient conditions for Unit Memory Channel Output (UMCO) channels\footnote{channels of class $A$ given by \eqref{introduction:class:channel:equation1}, with $M=1$.} to obtain the ergodic feedback capacity. We summarize the main one in the following remark. 
\begin{remark}(Conditions for ergodic feedback capacity of UMCO)\label{remark:conditions:ergodic:capacity}{\ \\}
Suppose the channel is time-invariant, i.e., $\{q_t(y_t|y_{t-1},x_t)\equiv{q}(y_t|y_{t-1},x_t):~t\in\mathbb{N}_0^n\}$. If the channel is strongly indecomposable and strongly aperiodic, as defined by Chen and Berger \cite[Definitions 2, 4]{chen-berger2005} the following  hold.
\begin{itemize}
\item[(a)] The optimal channel input distributions $\{\pi_t(x_t|y_{t-1}):~t\in\mathbb{N}_0^n\}$ converge asymptotically to time-invariant distributions denoted by $\pi^\infty(x|y), x\in{\cal X}, y\in{\cal Y}$, and the corresponding channel output transition probabilities converges to time-invariant transition probabilities $\nu^{\pi^\infty}(z|y),~z\in{\cal Y},~y\in{\cal Y}$. Moreover, there is a unique invariant distribution $\nu^{\pi^\infty}(y)$ corresponding to $\nu^{\pi^\infty}(z|y)$.
\item[(b)] The ergodic feedback capacity is given by
\begin{subequations}
\begin{align}
C^{FB, A.1}
=&\lim_{n\longrightarrow\infty}\sup_{\pi_t(x_t|y_{t-1}):~t\in\mathbb{N}_0^n}\frac{1}{n+1}{\bf E}^{\pi}\bigg\{\sum_{t=0}^n\log\Big(\frac{q(Y_t|Y_{t-1},X_t)}{\nu_t^{\pi}(Y_t|Y_{t-1})}\Big)\bigg\}\\
=&\sup_{\pi^\infty(x_t|y_{t-1}):~t=0,\ldots,\infty}\lim_{n\longrightarrow\infty}\frac{1}{n+1}{\bf E}^{\pi^{\infty}}\bigg\{\sum_{t=0}^n\log\Big(\frac{q(Y_t|Y_{t-1},X_t)}{\nu_t^{\pi}(Y_t|Y_{t-1})}\Big)\bigg\}\\
=&\sup_{\pi^\infty(x_0|y_{-1})}{\bf E}^{\pi^{\infty}}\bigg\{\log\Big(\frac{q(Y_0|Y_{-1},X_0)}{\nu^{\pi^{\infty}}(Y_0|Y_{-1})}\Big)\bigg\}\\
=&\sup_{\pi^\infty(x_0|y_{-1})}\sum_{y_{-1}}\Big(\sum_{x_0,y_0}\log\Big(\frac{q(y_0|y_{-1},x_0)}{\nu^{\pi^{\infty}}(y_0|y_{-1})}\Big)q(y_0|y_{-1},x_0)\pi^\infty(x_0|y_{-1})\Big)\nu^{\pi^{\infty}}(y_{-1})
.\end{align}\label{chen-burger_2} 
\end{subequations}
\item[(c)] The previous results extend to the case of feedback capacity with average transmission cost as follows.
\begin{subequations}
\begin{align}
C^{FB, A.1}(\kappa)
=&\lim_{n\longrightarrow\infty}\sup_{{\cal P}^{A.1}_{0,n}(\kappa)}\frac{1}{n+1}{\bf E}^{\pi}\bigg\{\sum_{t=0}^n\log\Big(\frac{q(Y_t|Y_{t-1},X_t)}{\nu_t^{\pi}(Y_t|Y_{t-1})}\Big)\bigg\}\\
=&\sup_{{\cal P}^{A.1,\infty}(\kappa)}\lim_{n\longrightarrow\infty}\frac{1}{n+1}{\bf E}^{\pi^{\infty}}\bigg\{\sum_{t=0}^n\log\Big(\frac{q(Y_t|Y_{t-1},X_t)}{\nu_t^{\pi}(Y_t|Y_{t-1})}\Big)\bigg\}\\
=&\sup_{\bar{\cal P}^{A.1,\infty}(\kappa)}{\bf E}^{\pi^{\infty}}\bigg\{\log\Big(\frac{q(Y_0|Y_{-1},X_0)}{\nu^{\pi^{\infty}}(Y_0|Y_{-1})}\Big)\bigg\}\\
=&\sup_{\bar{\cal P}^{A.1,\infty}(\kappa)}\sum_{y_{-1},x_0,y_0}\log\Big(\frac{q(y_0|y_{-1},x_0)}{\nu^{\pi^{\infty}}(y_0|y_{-1})}\Big)q(y_0|y_{-1},x_0)\pi^\infty(x_0|y_{-1})\nu^{\pi^{\infty}}(y_{-1})
\end{align}\label{chen-burger_3} 
\end{subequations}
where
\begin{align*}
{\cal P}^{A.1,\infty}(\kappa)=&\Big\{\pi^{\infty}(x_t|y_{t-1}),~t\in\mathbb{N}_0:~\lim_{n\longrightarrow\infty}\frac{1}{n+1}{\bf E}^{\pi^\infty}\big\{\sum_{t=0}^n \gamma(X_t,Y_{t-1}) \big\}\leq\kappa\Big\}\\
\bar{\cal P}^{A.1,\infty}(\kappa)=&\Big\{\pi^{\infty}(x_0|y_{-1}):~{\bf E}^{\pi^\infty}\big\{\gamma(X_0,Y_{-1}) \big\}\leq\kappa\Big\}.
\end{align*}
\end{itemize}
\end{remark}

The results derived in \cite{chen-berger2005} can be extended to channels  of class $A$. However, we do not proceed to do so, because for all application examples presented in this paper, we can show that $\frac{1}{n+1}C^{FB}_{X^n\rightarrow{Y^n}}$ \big(or $\frac{1}{n+1}C^{FB}_{X^n\rightarrow{Y^n}}(\kappa)$\big) corresponds to feedback capacity by investigating the ergodic asymptotic properties of the FTFI capacity.

\begin{remark}(Generalizations){\ \\}
The analysis presented in this subsection extends naturally to any combination of channels of classes $A$, $B$ and transmission cost constraint of classes $A$, $B$. This is shown in Section~\ref{section:generalizations}.
\end{remark}

%
%
%
%

\section{Application Examples}\label{section:application:examples}

\par In this section, we derive {\it closed form expressions of the optimal (nonstationary) channel input conditional distributions and the corresponding channel output transition probability distributions} of the characterization of the FTFI capacity, for the following channels. 
\begin{itemize}
\item[(a)] The time-varying Binary Unit Memory Channel Output (BUMCO)  channel defined by \eqref{introduction:section:example:matrix:general:bumco} with and without transmission cost constraint. 
\item[(b)] The time-varying Binary Erasure Unit Memory Channel Output (BEUMCO) channel defined by \eqref{introduction:example:conjecture:matrix:beumco}. 
\item[(c)] The time-varying Binary Symmetric Two Memory Channel Output (BSTMCO) channel defined by \eqref{introduction:section:example:matrix:general:bstmco}. 
\end{itemize}
For the time-invariant BUMCO channel and the BEUMCO channel, we also investigate the asymptotic properties of the optimal channel input conditional distribution via the per unit time limit of the characterization of FTFI capacity.

%
%
%
%

\subsection{The FTFI Capacity of Time-Varying BUMCO Channel and Feedback Capacity}\label{subsection:applications:ftfi:capacity:bumco}

\par In this subsection,  we give the derivation of equations \eqref{section:example:bumco:equation1}-\eqref{introduction:theorem:value:function:eq2}, \eqref{introduction:example:ergodic:feedback:capacity:bumco}-\eqref{section:example:bumco:time-invariant:equation1} of Theorem~\ref{introduction:theorem:bumco:optimal:solutions}, and we present numerical evaluations based on the closed form expressions for various scenarios.\\ 

\subsubsection{Proof of Equations \eqref{section:example:bumco:equation1}-\eqref{introduction:theorem:value:function:eq2}}\label{proof:ftfi:feedback:capacity} We provide the derivation of the backward recursive equations \eqref{section:example:bumco:equation1}-\eqref{introduction:theorem:value:function:eq2}.\\
 Denote the optimal distributions as follows. 
\begin{equation}
 \nu^{\pi^*}_t(y_t|y_{t-1})\triangleq\bordermatrix{&0&1\cr
            0&c_0(t)&1-c_1(t)\cr
            1&1-c_0(t)&c_1(t)\cr},\\
~~\pi^*_t(x_t|y_{t-1})\triangleq\bordermatrix{&0&1\cr
            0&d_0(t)&1-d_1(t)\cr
            1&1-d_0(t)&d_1(t)\cr},~t\in\mathbb{N}_0^n.\label{section:example:matrix:nominal:equation1}
\end{equation} 
We shall derive recursive expressions for $\{c_0(t),c_1(t),d_0(t),d_1(t):~t\in\mathbb{N}_0^n\}$.\\
Define 
\begin{align}
\Delta{C}_{t}\triangleq{C}_t(1)-C_t(0),~{t}\in\mathbb{N}_0^{n+1},~\Delta{C}_{n+1}(0)=\Delta{C}_{n+1}(1)=0.\label{section:example:time-n:equation13a}
\end{align}
$\bullet\underline{\mbox{Time t=n:}}$\\
By Theorem~\ref{introduction:baa:sequential:theorem:lmco:necessary:sufficient}, the necessary and sufficient condition for $\pi^*_n(x_n|y_{n-1})\neq{0}$ to achieve the supremum of the FTFI capacity of BUMCO channel is the following.
\begin{align}
C_n(y_{n-1})=\sum_{y_n\in\{0,1\}}\log\Big(\frac{q_n(y_n|x_n,y_{n-1})}{\nu^{\pi^*}_n(y_n|y_{n-1})}\Big)q_n(y_n|x_n,y_{n-1}),~\forall{x_n}.\label{section:example:bumco:ns:time-n}
\end{align}
Next, we evaluate $C_n(y_{n-1})$ for $x_n\in\{0,1\}$, for fixed $y_{n-1}$.\\
\underline{$y_{n-1}=0,~x_n=0$:}
\begin{align}
C_n(0)&=\sum_{y_n\in\{0,1\}}\log\Big(\frac{q_n(y_n|0,0)}{\nu^{\pi^*}_n(y_n|0)}\Big)q_n(y_n|0,0)=\log\Big(\frac{q_n(0|0,0)}{\nu^{\pi^*}_n(0|0)}\Big)q_n(0|0,0)+\log\Big(\frac{q_n(1|0,0)}{\nu^{\pi^*}_n(1|0)}\Big)q_n(1|0,0)\nonumber\\
&=\alpha_n\log\big(\frac{1-c_0(n)}{c_0(n)}\big)+\log\big(\frac{1}{1-c_0(n)}\big)-H(\alpha_n).\label{section:example:time-n:equation1}
\end{align}
\underline{$y_{n-1}=0,~x_n=1$:}
\begin{align}
C_n(0)&=\sum_{y_n\in\{0,1\}}\log\Big(\frac{q_n(y_n|1,0)}{\nu^{\pi^*}_n(y_n|0)}\Big)q_n(y_n|1,0)=\log\Big(\frac{q_n(0|1,0)}{\nu^{\pi^*}_n(0|0)}\Big)q_n(0|1,0)+\log\Big(\frac{q_n(1|1,0)}{\nu^{\pi^*}_n(1|0)}\Big)q_n(1|1,0)\nonumber\\
&=\gamma_n\log\big(\frac{1-c_0(n)}{c_0(n)}\big)+\log\big(\frac{1}{1-c_0(n)}\big)-H(\gamma_n).\label{section:example:time-n:equation2}
\end{align}
Since \eqref{section:example:time-n:equation1}=\eqref{section:example:time-n:equation2}, we obtain
\begin{align}
\nu^{\pi^*}_n(0|0)\equiv{c}_0(n)=\frac{1}{1+2^{\mu_0(n)}},~\mu_0(n)\triangleq\frac{H(\gamma_n)-H(\alpha_n)}{\gamma_n-\alpha_n}.\label{section:example:time-n:equation4}
\end{align}
The channel output transition probability at time $t=n$ is given by
\begin{align}
\nu^{\pi^*}_n(y_n|y_{n-1})=\sum_{x_n\in\{0,1\}}q_n(y_n|x_n,y_{n-1})\pi^*_n(x_n|y_{n-1}).\label{section:example:time-n:equation5}
\end{align}
We use \eqref{section:example:time-n:equation5} to find the values $\pi^*_n(0|0)\equiv{d}_0(n)$.\\
\underline{$y_{n-1}=0,~y_n=0$:}
\begin{align}
\nu^{\pi^*}_n(0|0)=\sum_{x_n\in\{0,1\}}q_n(0|x_n,0)\pi^*_n(x_n|0)=q_n(0|0,0)\pi_n(0|0)+q_n(0|1,0)\pi^*_n(1|0).\label{section:example:time-n:equation6a}
\end{align}
Substituting \eqref{section:example:time-n:equation4} into  \eqref{section:example:time-n:equation6a} we obtain
\begin{align}
\pi^*_n(0|0)\equiv{{d_0}(n)=\frac{1-\gamma_n(1+2^{\mu_0(n)})}{(\alpha_n-\gamma_n)(1+2^{\mu_0(n)})}}.\label{section:example:time-n:equation6}
\end{align}
We repeat the above procedure to compute the expressions of $C_n(1)$, $\nu^{\pi^*}_n(0|1)$, $\nu^{\pi^*}_n(1|1)$, $\pi^*_n(0|1)$ and $\pi^*_n(1|1)$. After some algebra, we obtain
\begin{align}
\nu^{\pi^*}_n(1|1)\equiv{{c}_1(n)=\frac{2^{\mu_1(n)}}{1+2^{\mu_1(n)}}},~\pi^*_n(1|1)\equiv{d_1(n)}=\frac{\beta_n(1+2^{\mu_1(n)})-1}{(\beta_n-\delta_n)(1+2^{\mu_1(n)})},~\mu_1(n)\triangleq\frac{H(\beta_n)-H(\delta_n)}{\beta_n-\delta_n}.\label{section:example:time-n:equation10}
\end{align}
Finally, we substitute \eqref{section:example:time-n:equation4}, \eqref{section:example:time-n:equation6} and \eqref{section:example:time-n:equation10}, in \eqref{section:example:matrix:nominal:equation1} to obtain \eqref{section:example:bumco:equation1} evaluated at $t=n$. Next, we evaluate $C_n(0)$, $C_n(1)$, since these are required in the next time step. After some algebra, we obtain the following expressions.
\begin{align}
C_n(0)=\mu_0(n)(\alpha_n-1)+\log(1+2^{\mu_0(n)})-H(\alpha_n),~C_n(1)=\mu_1(n)(\beta_n-1)+\log(1+2^{\mu_1(n)})-H(\beta_n).\label{section:example:time-n:equation13}
\end{align} 
Using \eqref{section:example:time-n:equation13} in \eqref{section:example:time-n:equation13a} we obtain \eqref{section:example:bumco:equation2b} at $t=n$ as follows.
\begin{align}
{\Delta{C}_{n}={C}_n(1)-C_n(0)=\big(\mu_1(n)(\beta_n-1)-\mu_0(n)(\alpha_n-1)\big)+H(\alpha_n)-H(\beta_n)+\log\Big(\frac{1+2^{\mu_1(n)}}{1+2^{\mu_0(n)}}\Big)}.  \label{section:example:time-n:equation14}
\end{align}
We proceed with the computation at the next time step.\\
$\bullet\underline{\mbox{Time t={n-1}:}}$\\
By Theorem~\ref{introduction:baa:sequential:theorem:lmco:necessary:sufficient}, 
\begin{align}
C_{n-1}(y_{n-2})=\sum_{y_{n-1}\in\{0,1\}}\Big(\log\big(\frac{q_{n-1}(y_{n-1}|x_{n-1},y_{n-2})}{\nu^{\pi^*}_{n-1}(y_{n-1}|y_{n-2})}\big)+C_n(y_{n-1})\Big)q_{n-1}(y_{n-1}|x_{n-1},y_{n-2}),~\forall{x_{n-1}}.\label{section:example:bumco:ns:time-n-1}
\end{align}
Next, we evaluate $C_{n-1}(y_{n-2})$ for $x_{n-1}\in\{0,1\}$, for fixed $y_{n-2}$.\\
\underline{$y_{n-2}=0,~x_{n-1}=0$:}
\begin{align}
&C_{n-1}(0)=\sum_{y_{n-1}\in\{0,1\}}\Big(\log\big(\frac{q_{n-1}(y_{n-1}|0,0)}{\nu^{\pi^*}_{n-1}(y_{n-1}|0)}\big)+C_n(y_{n-1})\Big)q_{n-1}(y_{n-1}|0,0)\nonumber\\
&=\Big(\log\big(\frac{q_{n-1}(0|0,0)}{\nu^{\pi^*}_{n-1}(0|0)}\big)+C_n(0)\Big)q_{n-1}(0|0,0)+\Big(\log\big(\frac{q_{n-1}(1|0,0)}{\nu^{\pi^*}_{n-1}(1|0)}\big)+C_n(1)\Big)q_{n-1}(1|0,0)\nonumber\\
&=\alpha_{n-1}\log\big(\frac{1-c_0(n-1)}{c_0(n-1)}\big)+\log\big(\frac{1}{1-c_0(n-1)}\big)-H(\alpha_{n-1})-\alpha_{n-1}{C}_n(0)+(1-\alpha_{n-1})C_n(1).\label{section:example:time-n-1:equation1}
\end{align}
\underline{$y_{n-2}=0,~x_{n-1}=1$:}
\begin{align}
&C_{n-1}(0)=\sum_{y_{n-1}\in\{0,1\}}\Big(\log\big(\frac{q_{n-1}(y_{n-1}|1,0)}{\nu^{\pi^*}_{n-1}(y_{n-1}|0)}\big)+C_n(y_{n-1})\Big)q_{n-1}(y_{n-1}|1,0)\nonumber\\
&=\Big(\log\big(\frac{q_{n-1}(0|1,0)}{\nu^{\pi^*}_{n-1}(0|0)}\big)+C_n(0)\Big)q_{n-1}(0|1,0)+\Big(\log\big(\frac{q_{n-1}(1|1,0)}{\nu^{\pi^*}_{n-1}(1|0)}\big)+C_n(1)\Big)q_{n-1}(1|1,0)\nonumber\\
&=\gamma_{n-1}\log\big(\frac{1-c_0(n-1)}{c_0(n-1)}\big)+\log\big(\frac{1}{1-c_0(n-1)}\big)-H(\gamma_{n-1})-\gamma_{n-1}{C}_n(0)+(1-\gamma_{n-1})C_n(1).\label{section:example:time-n-1:equation2}
\end{align}
Since \eqref{section:example:time-n-1:equation1}=\eqref{section:example:time-n-1:equation2}, we obtain
\begin{align}
\nu^{\pi^*}_{n-1}(0|0)\equiv{{c}_0(n-1)=\frac{1}{1+2^{\mu_0(n-1)+\Delta{C}_n}}},~\mu_0(n-1)\triangleq\frac{H(\gamma_{n-1})-H(\alpha_{n-1})}{\gamma_{n-1}-\alpha_{n-1}}.\label{section:example:time-n-1:equation4}
\end{align}
The channel output transition probability at time $t=n-1$ is given by
\begin{align}
\nu^{\pi^*}_{n-1}(y_{n-1}|y_{n-2})=\sum_{x_{n-1}\in\{0,1\}}q_{n-1}(y_{n-1}|x_{n-1},y_{n-2})\pi^*_{n-1}(x_{n-1}|y_{n-2}).\label{section:example:time-n-1:equation5}
\end{align}
We use \eqref{section:example:time-n-1:equation5} to find the values of $\pi^*_{n-1}(0|0)$ and $\pi^*_{n-1}(1|0)$.\\
\underline{$y_{n-2}=0,~y_{n-1}=0$:}
\begin{align}
\nu^{\pi^*}_{n-1}(0|0)&=\sum_{x_{n-1}\in\{0,1\}}q_{n-1}(0|x_{n-1},0)\pi^*_{n-1}(x_{n-1}|0)=q_{n-1}(0|0,0)\pi^*_{n-1}(0|0)+q_{n-1}(0|1,0)\pi^*_{n-1}(1|0)\label{section:example:time-n-1:equation5a}
\end{align}
Substituting \eqref{section:example:time-n-1:equation4} into \eqref{section:example:time-n-1:equation5a} we obtain
\begin{align}
\pi^*_{n-1}(0|0)\equiv{{d_0(n-1)}=\frac{1-\gamma_{n-1}(1+2^{\mu_0(n-1)+\Delta{C}_n})}{(\alpha_{n-1}-\gamma_{n-1})(1+2^{\mu_0(n-1)+\Delta{C}_n})}}.\label{section:example:time-n-1:equation6}
\end{align}
Repeating the above procedure we obtain the expressions for $C_{n-1}(1)$, $\nu^{\pi^*}_{n-1}(0|1)$, $\nu^{\pi^*}_{n-1}(1|1)$, $\pi^*_n(0|1)$ and $\pi^*_{n-1}(1|1)$. After some algebra, we obtain
\begin{align}
\nu^{\pi^*}_{n-1}(1|1)\equiv{{c}_1(n-1)=\frac{2^{\mu_1(n-1)}}{2^{\mu_1(n-1)+\Delta{C}_n}}},~\pi^*_{n-1}(1|1)\equiv{{d_1(n-1)}=\frac{\beta_{n-1}(1+2^{\mu_1(n-1)+\Delta{C}_n})-1}{(\beta_{n-1}-\delta_{n-1})(1+2^{\mu_1(n-1)+\Delta{C}_n})}}\label{section:example:time-n-1:equation10}
\end{align}
where
\begin{align}
\mu_1(n-1)\triangleq\frac{H(\beta_{n-1})-H(\delta_{n-1})}{\beta_{n-1}-\delta_{n-1}}.
\end{align}
Finally, we substitute \eqref{section:example:time-n-1:equation4}, \eqref{section:example:time-n-1:equation6} and \eqref{section:example:time-n-1:equation10} in \eqref{section:example:matrix:nominal:equation1} to obtain \eqref{section:example:bumco:equation1} evaluated at $t=n-1$.
Similarly as before, we evaluate $C_{n-1}(0)$, $C_{n-1}(1)$, which are required in the next time step. After some algebra, we obtain the following expressions.
\begin{align}
C_{n-1}(0)&=\mu_0(n-1)(\alpha_{n-1}-1)+C_n(0)+\log(1+2^{\mu_0(n-1)+\Delta{C}_n})-H(\alpha_{n-1}),\nonumber\\
~C_{n-1}(1)&=\mu_1(n-1)(\beta_{n-1}-1)+{C}_n(0)+\log(1+2^{\mu_1(n-1)+\Delta{C}_n})-H(\beta_{n-1}).\label{section:example:time-n-1:equation13}
\end{align} 
Finally, using \eqref{section:example:time-n-1:equation13} in \eqref{section:example:time-n:equation13a} we obtain \eqref{section:example:bumco:equation2b} at $t=n-1$.\\
To complete the derivation we need to apply induction hypothesis, i.e., to show validity of the solution for $t=n-k$, provided it is valid for $t=n,n-1,n-2,\ldots,n-k+1$. This is done precisely as the derivation of the time step $t=n-1$, hence we omit it.  This completes the derivation.

\subsubsection{Proof of Equations \eqref{introduction:example:ergodic:feedback:capacity:bumco}-\eqref{section:example:bumco:time-invariant:equation1}}\label{proof:ergodic:feedback:capacity} Next, we address the asymptotic convergence of the optimal channel input conditional distribution and the corresponding channel output transition probability distribution given in \eqref{section:example:bumco:equation1}, by investigating the convergence properties of the value functions $\{C_t(0),~C_t(1),~t\in\mathbb{N}_0^n\}$ in terms of their difference $\{\Delta{C}_t:~t\in\mathbb{N}_0^n\}$. Conditions for convergence of the sequence $\{\Delta{C}_t:~t\in\mathbb{N}_0^n\}$, can be expressed in terms of parameters $\{\alpha_t, \beta_t, \gamma_t, \delta_t:~t\in\mathbb{N}_0^n\}$. From \eqref{section:example:bumco:equation2b}, it follows by contradiction, that the sequence $\{\Delta{C}_t:~t\in\mathbb{N}_0^n\}$ cannot diverge, i.e., it is bounded. \\
\noi Consider the time-invariant version of BUMCO $\{q_t(y_t|y_{t-1},x_t)=q(y_t|y_{t-1},x_t):~t\in\mathbb{N}^n_0\}$, denoted by BUMCO$(\alpha,\beta,\gamma,\delta)$. First, recall that recursion \eqref{section:example:bumco:equation2b} is expressed as follows
\begin{align}
\Delta{C}_t=&\big(\mu_1(\beta-1)-\mu_0(\alpha-1)\big)+H(\alpha)-H(\beta)+\log\Big(\frac{1+2^{\mu_1+\Delta{C}_{t+1}}}{1+2^{\mu_0+\Delta{C}_{t+1}}}\Big),~\Delta{C}_{n+1}=0,\label{discussion:equation1}\\
=&f(\alpha,\beta,\mu_0,\mu_1,\Delta{C}_{t+1}),~t\in\{n,\ldots,0\}\nonumber
\end{align}
where
\begin{align*}
\mu_0(\alpha_t,\gamma_t)\longmapsto\mu_0(\alpha,\gamma)=\frac{H(\gamma)-H(\alpha)}{\gamma-\alpha}\equiv{\mu}_0,~~~\mu_1(\beta_t,\delta_t)\longmapsto\mu_1(\beta,\delta)=\frac{H(\beta)-H(\delta)}{\beta-\delta}\equiv\mu_1,~\forall{t}.
\end{align*}
Define $\{\Delta\bar{C}_t=\Delta{C}_{n-t}:~t\in\mathbb{N}_0^{n+1}\}$. Then by \eqref{discussion:equation1} we obtain the following forward recursions
\begin{align}
\Delta\bar{C}_t=&\big(\mu_1(\beta-1)-\mu_0(\alpha-1)\big)+H(\alpha)-H(\beta)+\log\Big(\frac{1+2^{\mu_1+\Delta\bar{C}_{t-1}}}{1+2^{\mu_0+\Delta{C}_{t-1}}}\Big),~\Delta\bar{C}_{-1}=0,~t\in\mathbb{N}_0^n.\label{discussion:equation2}
\end{align}

\noi Since $\Big{|}\frac{\partial}{\partial{\Delta\bar{C}_t}}f(\alpha,\beta,\mu_0,\mu_1,\Delta\bar{C}_{t-1})\Big{|}<1$, then $\lim_{t\longrightarrow\infty}\Delta\bar{C}_t=\Delta\bar{C}^{\infty}\equiv\Delta{C}^{\infty}$, where $\Delta{C}^{\infty}$ satisfies the following algebraic equation.
\begin{align}
\Delta{C}^\infty=&\big(\mu_1(\beta-1)-\mu_0(\alpha-1)\big)+H(\alpha)-H(\beta)+\log\Big(\frac{1+2^{\mu_1+\Delta{C}^{\infty}}}{1+2^{\mu_0+\Delta{C}^{\infty}}}\Big).\label{remark:example:steady:state:updated:factor}
\end{align}
The real solution of the nonlinear equation \eqref{remark:example:steady:state:updated:factor} is 
\begin{align}
\Delta{C}^\infty=\log\Big((2^{\ell_1}-1)+\sqrt{(1-2^{\ell_1})^2+2^{\ell_0+2}}\Big)-\mu_0-1\label{solutions:steady:state:dc}
\end{align}
where
\begin{align*}
\ell_0\equiv\ell_0(\alpha,\beta,\gamma,\delta)\triangleq&\mu_1(\beta-1)-\mu_0(\alpha-2)+H(\alpha)-H(\beta),\\
\ell_1\equiv\ell_1(\alpha,\beta,\gamma,\delta)\triangleq&\mu_1\beta-\mu_0(\alpha-1)+H(\alpha)-H(\beta).
\end{align*}
Hence, by \eqref{solutions:steady:state:dc}, the optimal channel input conditional distribution and the corresponding output transition probability distribution converge asymptotically to the time-invariant transition probabilities given by \eqref{section:example:bumco:time-invariant:equation1}. It remains to show that the channel output transition probability distribution given by \eqref{section:example:bumco:time-invariant:equation1}, has a unique invariant distribution $\{\nu^{{\pi^{*,\infty}}}(y):~y\in\{0,1\}\}$.\\
Solving the equation
\begin{align}
\bordermatrix{&\cr
            &\nu^{{\pi^{*,\infty}}}(0)\cr
            &\nu^{{\pi^{*,\infty}}}(1)\cr}=\bordermatrix{&\cr
            &\nu^{\pi^{*,\infty}}(0|0)&\nu^{\pi^{*,\infty}}(0|1)\cr
            &\nu^{\pi^{*,\infty}}(1|0)&\nu^{\pi^{*,\infty}}(1|1)\cr}\bordermatrix{&\cr
            &\nu^{{\pi^{*,\infty}}}(0)\cr
            &\nu^{{\pi^{*,\infty}}}(1)\cr}\label{remark:section:example:matrix:general:invariant:distribution}
\end{align}
we obtain the unique solution
\begin{align}
\nu^{{\pi^{*,\infty}}}(0)=\frac{1+2^{\mu_0+\Delta{C}^{\infty}}}{1+2^{\mu_0+\mu_1+2\Delta{C}^{\infty}}+2^{\mu_0+1+\Delta{C}^{\infty}}},~\nu^{{\pi^{*,\infty}}}(1)=\frac{2^{\mu_0+\Delta{C}^{\infty}}(1+2^{\mu_1+\Delta{C}^{\infty}})}{1+2^{\mu_0+\mu_1+2\Delta{C}^{\infty}}+2^{\mu_0+1+\Delta{C}^{\infty}}}.\nonumber
\end{align}
Since $\nu^{{\pi^{*,\infty}}}$ is unique, then the feedback capacity of time-invariant BUMCO($\alpha,\beta,\gamma,\delta$) is given by the following expression.
\begin{align}
C^{FB,A.1}=\sum_{y\in\{0,1\}}\bigg(\sum_{x\in\{0,1\},z\in\{0,1\}}\log\Big(\frac{q(z|y,x)}{\nu^{*,\infty}(z|y)}\Big)q(z|y,x)\pi^{*,\infty}(x|y)\bigg)\nu^{{\pi^{*,\infty}}}(y).
\end{align}
After some algebra, we obtain \eqref{introduction:example:ergodic:feedback:capacity:bumco}.

\subsubsection{Numerical evaluations} Fig.~\ref{fig:bsumco:without:cost} depicts numerical simulations of the optimal (nonstationary) channel input conditional distribution and the corresponding channel output transition probability distribution given by \eqref{section:example:bumco:equation1}, for a time-invariant channel $$BUMCO(\alpha_t,\beta_t,\gamma_t,\delta_t)=BUMCO(0.9,0.1,0.2,0.4),$$ for $n=1000$.\\
\noi Fig.~\ref{fig:bumco:ftfi:capacity} depicts the corresponding value of $\frac{1}{n+1}C_{X^n\rightarrow{Y^n}}^{FB,A.1}=\frac{1}{n+1}{\bf E}^{\pi^*}\Big\{\sum_{t=0}^n\log\big(\frac{q(y_t|y_{t-1},x_t)}{\nu^{\pi^*}(y_t|y_{t-1})}\big)\Big\}$ where $\{\pi_t^*(x_t|y_{t-1}):~t=0,1,\ldots,n\}$ is given by \eqref{section:example:bumco:equation1}, for $n=1000$. From Fig.~\ref{fig:bumco:ftfi:capacity}, at $n\approx{1000}$, the characterization of FTFI capacity is $\frac{1}{n+1}C_{X^n\rightarrow{Y^n}}^{FB,A.1}=0.2148~\mbox{bits/channel use}$, while the actual ergodic feedback capacity evaluated from \eqref{introduction:example:ergodic:feedback:capacity:bumco} is $C^{FB,A.1}=0.215~\mbox{bits/channel use}$. \\
Based on our simulations, it is interesting to point out the fact that the optimal channel input conditional distribution and the corresponding channel output transition probability converge to their asymptotic values at $n\approx{400}$, with respect to an error tolerance of $10^{-3}$.

\begin{figure}
    \centering
    \begin{subfigure}[b]{0.49\textwidth}
        \centering
        \includegraphics[width=\textwidth]{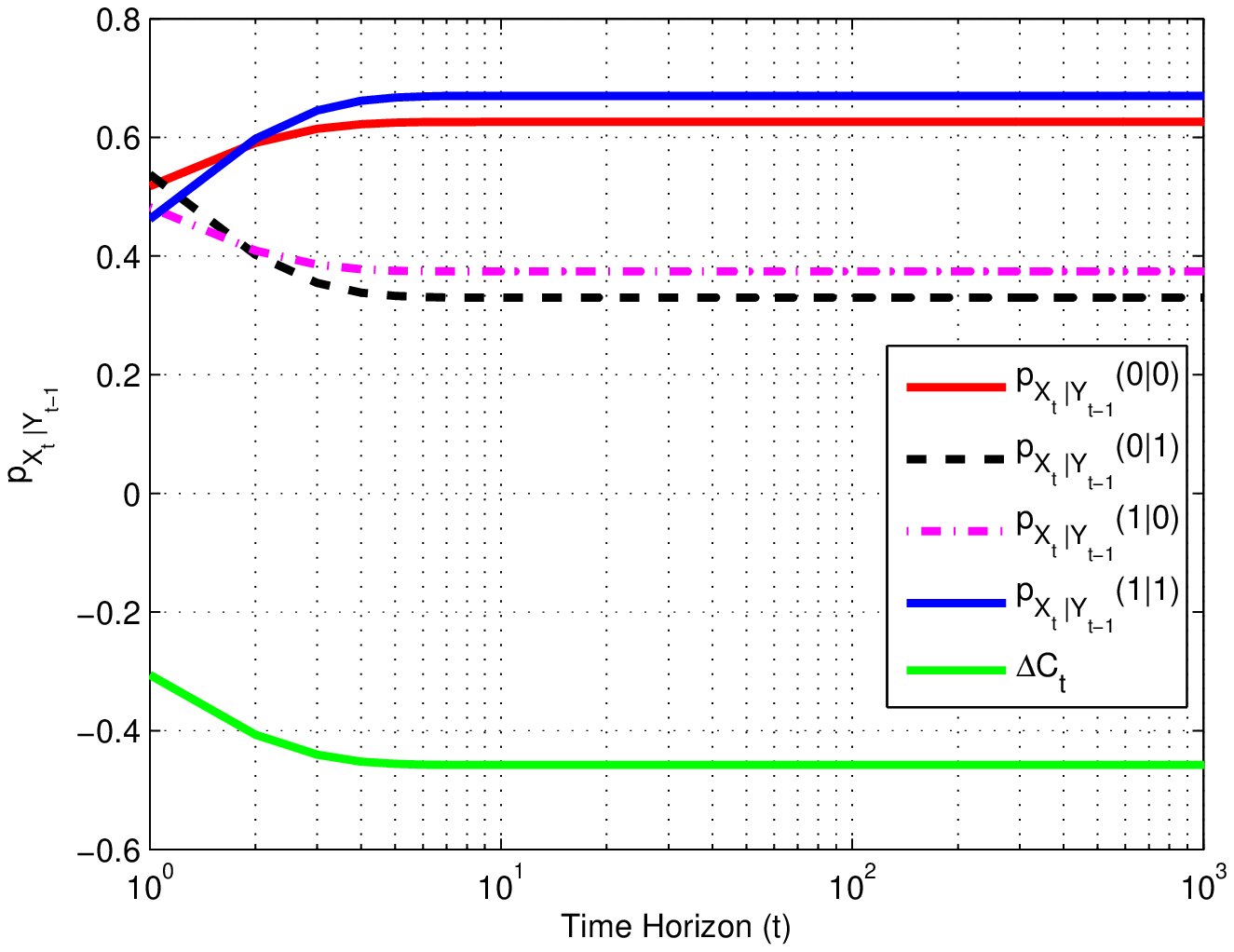}
        \caption{Optimal Distributions $\pi_t^*(x_t|y_{t-1})$ and $\Delta{C}_t$.}\label{fig:bsumco:optimal:input:distribution}
    \end{subfigure}
    \hfill
    \begin{subfigure}[b]{0.49\textwidth}
        \centering
        \includegraphics[width=\textwidth]{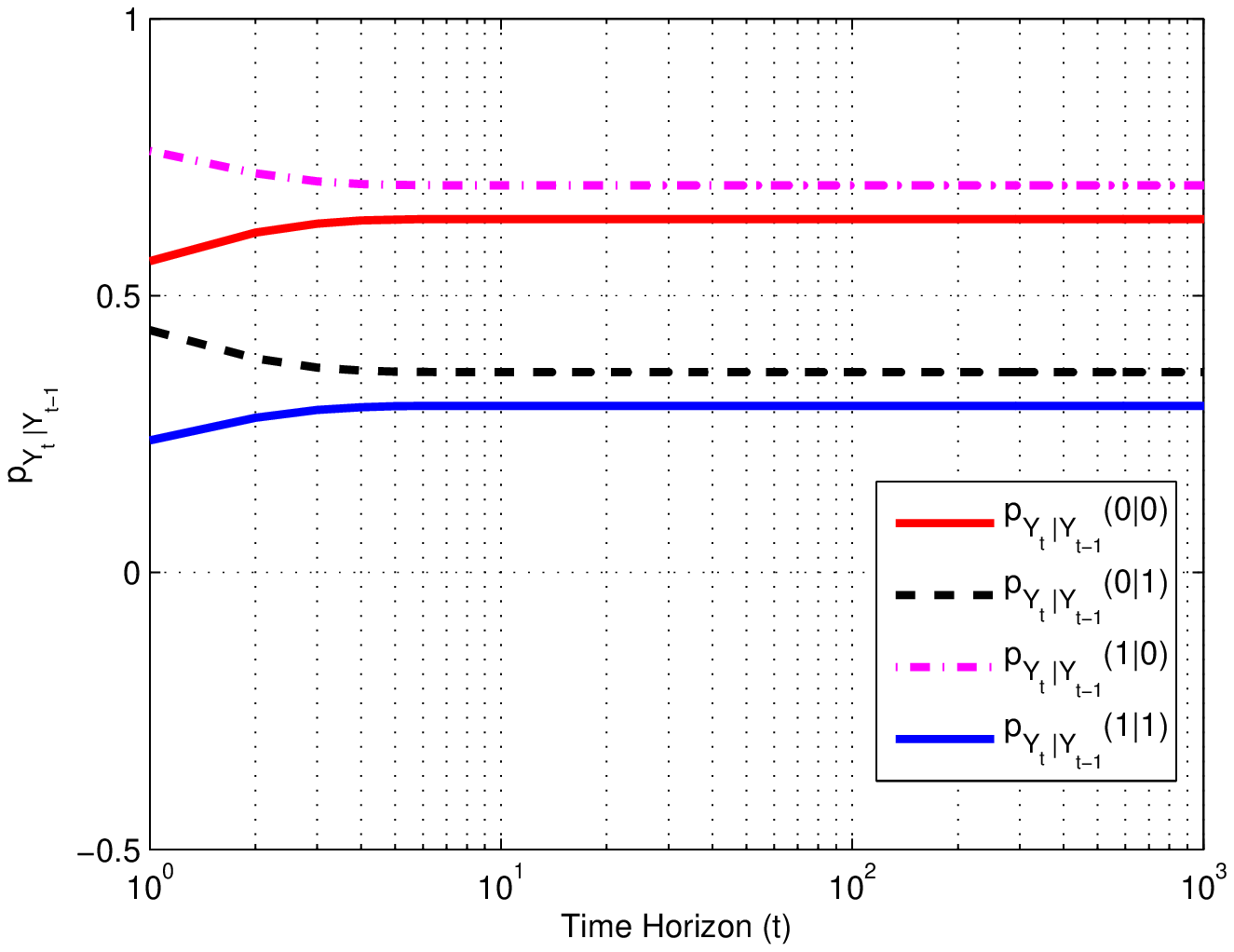}      
        \caption{Optimal Distributions $\nu_t^{\pi^*}(y_t|y_{t-1})$.}\label{fig:bsumco:optimal:output:distribution}
    \end{subfigure}
    \caption{Optimal distributions of $BUMCO(0.9,0.1,0.2,0.4)$ for $n=1000$.}\label{fig:bsumco:without:cost}
\end{figure}
\begin{figure}[h]
\centering
\includegraphics[width=0.49\textwidth]{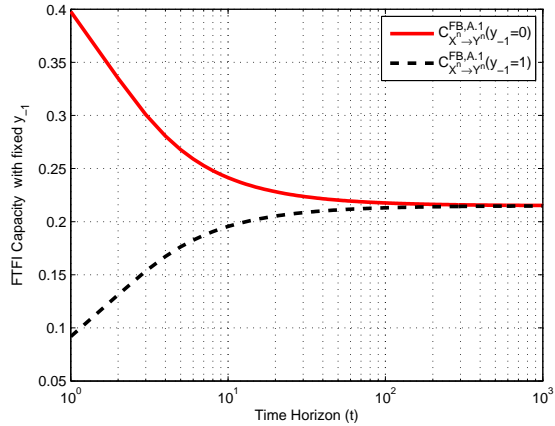}
\caption{ $\frac{1}{n+1}C^{FB,A.1}_{X^n\rightarrow{Y^n}}$ of BUMCO ($0.9$, $0.1$, $0.2$, $0.4$) for $n=1000$ with a choice of the initial distribution ${\bf P}_{Y_{-1}}(y_{-1}=0)=0$ with its complement ${\bf P}_{Y_{-1}}(y_{-1}=1)=1$.}
\label{fig:bumco:ftfi:capacity}
\end{figure}

\subsubsection{Special Cases of Equations \eqref{section:example:bumco:equation1}-\eqref{section:example:bumco:equation2b}}\label{example:special case} Next, we discuss special cases of $BUMCO(\alpha,\beta,\gamma,\delta)$.
\begin{itemize}
\item[$\bullet$] The POST channel investigated in \cite{permuter-asnani-weissman2014ieeeit} corresponds to the degenerated channel BUMCO$(\alpha,1-\beta,\beta,1-\alpha)$. The authors in \cite{permuter-asnani-weissman2014ieeeit} derived the expression of feedback capacity $C^{FB,A.1}$ and the optimal channel output distribution using known expressions of the so called $Z$ and $S$ channels without, however, determining the capacity achieving input distribution.
\item[$\bullet$] The BSCC investigated in \cite{kourtellaris-charalambous2015itw}, corresponds to the degenerated channel BUMCO$(\alpha,\beta,1-\beta,1-\alpha)$. The authors in \cite{kourtellaris-charalambous2015itw} derived the feedback capacity and the corresponding channel input conditional distribution with and without transmission cost constraint, and they have also shown that feedback does not increase the capacity. Our general expressions  \eqref{section:example:bumco:equation1}-\eqref{section:example:bumco:equation2b} give, as degenerated cases, the expressions obtained in \cite{permuter-asnani-weissman2014ieeeit,kourtellaris-charalambous2015itw}.
\item[$\bullet$] For the special case of BUMCO$(\alpha,\alpha,1-\alpha,1-\alpha)$, the channel is memoryless, and the recursive equations \eqref{section:example:bumco:equation1}-\eqref{section:example:bumco:equation2b} degenerate to the well-known results of memoryless Binary Symmetric Channels (BSC), where the optimal channel input distribution is uniform \cite{cover-thomas2006}.
\end{itemize}

%
%
%
%

\subsection{The FTFI Capacity of Time-Varying BUMCO Channel  with Transmission Cost and Feedback Capacity}\label{subsection:applications:ftfi:capacity:bumco:cost}

\par In this subsection, we apply Theorem~\ref{theorem:necessary:sufficient:ftfi:lmco:cost}, for $M=1$ and $N=1$, to derive closed form expressions for the optimal channel input and output distributions of BUMCO given by \eqref{introduction:section:example:matrix:general:bumco}. \\
We consider a transmission cost function $c^{A.1}(x^n,y^{n-1})\triangleq\sum_{t=0}^n\gamma_t(x_t,y_{t-1})$, where 
\begin{align}
 \gamma_t(x_t,y_{t-1})\triangleq\bordermatrix{&0&1\cr
            0&1&0\cr
            1&0&1\cr},\quad~t\in\mathbb{N}_0.\label{section:example:general:bumco:cost:function}
\end{align}  
The optimal solution of the characterization of FTFI capacity is given in the next theorem.

\begin{theorem}(Optimal solution of the characterization of FTFI capacity of time-varying BUMCO with transmission cost)\label{section:example:theorem:bumco:cost}{\ \\}
Consider the BUMCO($\alpha_t$,$\beta_t$,$\gamma_t$,$\delta_t$) defined in \eqref{introduction:section:example:matrix:general:bumco}, when the cost function \eqref{section:example:general:bumco:cost:function} is imposed. 
\begin{itemize}
\item[(a)] The optimal channel input distribution and corresponding channel output transition probability distribution corresponding to $C^{FB,A.1}_{X^n\rightarrow{Y^n}}(\kappa)$, defined by \eqref{characterization:ftfi:capacity:cost:lmco1}, when $\{\pi^*_t(x_t|y_{t-1})\neq{0},~\forall{x_t\in{\cal X}_t},~t\in\mathbb{N}_0^{n}\}$ and $s\geq{0}$, are the following. 
\begin{subequations}
\begin{align}
\pi^*_t(0|0)&=\frac{1-\gamma_t(1+2^{\mu^s_0(t)+\Delta{K}^s_{t+1}})}{(\alpha_t-\gamma_t)(1+2^{\mu^s_0(t)+\Delta{K}^s_{t+1}})},~&~\pi^*_t(0|1)&=\frac{1-\delta_t(1+2^{\mu^s_1(t)+\Delta{K}^s_{t+1}})}{(\beta_t-\delta_t)(1+2^{\mu^s_1(t)+\Delta{K}^s_{t+1}})},\\
\pi^*_t(1|0)&=1-\pi^*_t(0|0),~&~\pi^*_t(1|1)&=1-\pi^*_t(0|1),\\
\nu^{\pi^*}_t(0|0)&=\frac{1}{1+2^{\mu^s_0(t)+\Delta{K}^s_{t+1}}},~&~\nu^{\pi^*}_t(0|1)&=\frac{1}{1+2^{\mu^s_1(t)+\Delta{K}^s_{t+1}}},\\~\nu^{\pi^*}_t(1|0)&=1-\nu^{\pi^*}_t(0|0),~&~\nu^{\pi^*}_t(1|1)&=1-\nu^{\pi^*}_t(0|1)
\end{align}\label{section:example:bumco:cost:equation1}
\end{subequations}
where $\{\Delta{K}^s_t(\alpha_t,\beta_t,\gamma_t,\delta_t,s)\equiv{\Delta}K^s_t\triangleq{K}^s_{t}(0)-K^s_{t}(1):~ t\in\mathbb{N}_0^{n+1}\}$ is the difference of the value functions at each time, satisfying the backward recursions
\begin{subequations}
\begin{align}
\Delta{K}^s_{n+1}&=0\\
\Delta{K}^s_{t}&=\big(\mu^s_1(t)(\beta_t-1)-\mu^s_0(t)(\alpha_t-1)\big)+H(\alpha_t)-H(\beta_t)\nonumber\\
&\qquad\qquad+\log\Big(\frac{1+2^{\mu^s_1(t)+\Delta{K}^s_{t+1}}}{1+2^{\mu^s_0(t)+\Delta{K}^s_{t+1}}}\Big)+s,~t\in\{n,\ldots,0\}.  
\end{align}\label{section:example:bumco:cost:equation2b}
\end{subequations}
and
\begin{align*}
\mu_0(\alpha_t,\gamma_t,s)\triangleq\frac{H(\gamma_t)-H(\alpha_t)-s}{\gamma_t-\alpha_t}\equiv{\mu}^s_0(t),~\mu_1(\beta_t,\delta_t,s)\triangleq\frac{H(\beta_t)-H(\delta_t)-s}{\beta_t-\delta_t}\equiv\mu^s_1(t).
\end{align*}
\item[(b)] The solution of the value functions is given recursively by the following expressions.
\begin{align}
K^s_{t}(0)&=\mu_0(t)(\alpha_{t}-1)+K^s_{t+1}(0)+\log(1+2^{\mu_0(t)+\Delta{K}^s_{t+1}})-H(\alpha_{t}),~K^s_{n+1}(0)=0,\label{introduction:theorem:value:bumco:cost:function:eq1}\\
~K^s_{t}(1)&=\mu_1(t)(\beta_{t}-1)+{K}^s_{t+1}(0)+\log(1+2^{\mu_1(t)+\Delta{K}^s_{t+1}})-H(\beta_{t}),~K^s_{n+1}(1)=0,~t\in\{n,\ldots,0\}.\label{introduction:theorem:value:bumco:cost:function:eq2}
\end{align}
\item[(c)] The characterization of the FTFI capacity is given by
\begin{align*}
C^{FB,A.1}_{X^n\rightarrow{Y^n}}(\kappa)=\inf_{s\geq{0}}\sum_{y_{-1}\in\{0,1\}}\Big(K^s_0(y_{-1})\mu(y_{-1})+(n+1)\kappa\Big),~\mu(y_{-1})~\mbox{is fixed}.
\end{align*}
\end{itemize}
\end{theorem}
\begin{proof}
The derivation is similar to the one of subsubsection~\ref{proof:ftfi:feedback:capacity}, hence we omit it.
\end{proof}

\noi Next, we comment on the time-invariant version of Theorem~\ref{section:example:theorem:bumco:cost}. 

\subsubsection{Time-Invariant BUMCO with Transmission Cost} Consider the steady state version of \eqref{section:example:bumco:cost:equation2b}, defined by the following algebraic equation.
\begin{align}
\Delta{K}^{s,\infty}=&\big(\mu^s_1(\beta-1)-\mu^s_0(\alpha-1)\big)+H(\alpha)-H(\beta)+s+\log\Big(\frac{1+2^{\mu^s_1+\Delta{K}^{s,\infty}}}{1+2^{\mu^s_0+\Delta{K}^{s,\infty}}}\Big).\label{discussion:example:steady:state:cost:updated:factor}
\end{align}
where
\begin{align*}
\mu^s_0(\alpha_t,\gamma_t)\longmapsto\mu^s_0(\alpha,\gamma)=\frac{H(\gamma)-H(\alpha)}{\gamma-\alpha}\equiv{\mu}^s_0,~~~\mu^s_1(\beta_t,\delta_t)\longmapsto\mu^s_1(\beta,\delta)=\frac{H(\beta)-H(\delta)}{\beta-\delta}\equiv\mu^s_1,~\forall{t}.
\end{align*}
The real solution of the nonlinear equation \eqref{discussion:example:steady:state:cost:updated:factor} is 
\begin{align}
\Delta{K}^{s,\infty}=\log\Big((2^{\ell_1}-1)+\sqrt{(1-2^{\ell_1})^2+2^{\ell_0+2}}\Big)-\mu_0-1\label{solutions:steady:state:dc:cost}
\end{align}
where
\begin{align*}
\ell_0\equiv\ell_0(\alpha,\beta,\gamma,\delta)\triangleq&\mu_1(\beta-1)-\mu_0(\alpha-2)+H(\alpha)-H(\beta)+s,\\
\ell_1\equiv\ell_1(\alpha,\beta,\gamma,\delta)\triangleq&\mu_1\beta-\mu_0(\alpha-1)+H(\alpha)-H(\beta)+s.
\end{align*}
By \eqref{solutions:steady:state:dc:cost}, the optimal time-invariant channel input conditional distribution and the corresponding output transition probability distribution are the following. 
\begin{subequations}
\begin{align}
\pi^{*,\infty}(0|0)&=\frac{1-\gamma(1+2^{\mu^s_0+\Delta{K}^{s,\infty}})}{(\alpha-\gamma)(1+2^{\mu^s_0+\Delta{K}^{s,\infty}})},~&~\pi^{*,\infty}(0|1)&=\frac{1-\delta(1+2^{\mu^s_1+\Delta{K}^{s,\infty}})}{(\beta-\delta)(1+2^{\mu^s_1+\Delta{K}^{s,\infty}})},\\
\pi^{*,\infty}(1|0)&=1-\pi^{*,\infty}(0|0),~&~\pi^{*,\infty}(1|1)&=1-\pi^{*,\infty}(0|1),\\
\nu^{\pi^{*,\infty}}(0|0)&=\frac{1}{1+2^{\mu^s_0+\Delta{K}^{S,\infty}}},~&~\nu^{\pi^{*,\infty}}(0|1)&=\frac{1}{1+2^{\mu^s_1+\Delta{K}^{s,\infty}}},\\
~\nu^{\pi^{*,\infty}}(1|0)&=1-\nu^{\pi^{*,\infty}}(0|0),~&~\nu^{\pi^{*,\infty}}(1|1)&=1-\nu^{\pi^{*,\infty}}(0|1).
\end{align}
\label{remark:section:example:matrix:general:input-output:steady:state:cost}
\end{subequations}
Utilizing the channel output transition probability distribution given by \eqref{remark:section:example:matrix:general:input-output:steady:state:cost}, we obtain the following unique invariant distribution $\{\nu^{{\pi^{*,\infty}}}(y):~y\in\{0,1\}\}$ corresponding to $\{\nu^{{\pi^{*,\infty}}}(z|y):~(z,y)\in\{0,1\}\times\{0,1\}\}$.
\begin{align}
\nu^{{\pi^{*,\infty}}}(0)=\frac{1+2^{\mu^s_0+\Delta{K}^{s,\infty}}}{1+2^{\mu^s_0+\mu^s_1+2\Delta{K}^{s,\infty}}+2^{\mu^s_0+1+\Delta{K}^{s,\infty}}},~\nu^{{\pi^{*,\infty}}}(1)=\frac{2^{\mu^s_0+\Delta{K}^{s,\infty}}(1+2^{\mu^s_1+\Delta{K}^{s,\infty}})}{1+2^{\mu^s_0+\mu^s_1+2\Delta{K}^{s,\infty}}+2^{\mu^s_0+1+\Delta{K}^{s,\infty}}}.\label{unique:invariant:distribution:cost}
\end{align}
The feedback capacity of time-invariant BUMCO($\alpha,\beta,\gamma,\delta$) with transmission cost $\kappa$, is given by the following expression (following \eqref{remark:section:example:matrix:general:input-output:steady:state:cost} and \eqref{unique:invariant:distribution:cost}).
\begin{align}
C^{FB,A.1}(\kappa)=&\nu_0\Big(H(\nu_{0|0})-H(\gamma)\Big)+(1-\nu_0)\Big(H(\nu_{0|1})-H(\delta)\Big)+\xi_0\Big(H(\gamma)-H(\alpha)\Big)\nonumber\\
&+\xi_1\Big(H(\delta)-H(\beta)\Big)\label{remark:section:example:ergodic:feedback:capacity:cost}
\end{align}
where
\begin{align}
\nu_0&=\nu^{{\pi^{*,\infty}}}(0),~~\qquad\qquad\xi_0=\frac{1-\gamma(1+2^{\mu^s_0+\Delta{K}^{s,\infty}})}{(\alpha-\gamma)\big(1+2^{\mu^s_0+\mu^s_1+2\Delta{K}^{s,\infty}}+2^{\mu^s_0+1+\Delta{K}^{s,\infty}}\big)},\nonumber\\
\xi_1&=\frac{2^{\mu^s_0+\Delta{K}^{s,\infty}}\big(1-\delta(1+2^{\mu^s_1+\Delta{K}^{s,\infty}})\big)}{(\beta-\delta)\big(1+2^{\mu^s_0+\mu^s_1+2\Delta{K}^{s,\infty}}+2^{\mu^s_0+1+\Delta{K}^{s,\infty}}\big)},~\quad\nu_{0|0}=\nu^{\pi^{*,\infty}}(0|0),~~~\nu_{0|1}=\nu^{\pi^{*,\infty}}(0|1).\nonumber
\end{align}
Note that by Theorem~\ref{thm-pr_fb}, at $s=0$, $\kappa=\kappa_{max}$, and $C^{FB,A.1}(\kappa)=C^{FB,A.1}$.
Utilizing \eqref{remark:section:example:matrix:general:input-output:steady:state:cost} and \eqref{unique:invariant:distribution:cost} we can find ($s(\kappa),\kappa$) from the following expression.
\begin{align}
&\lim_{n\longrightarrow\infty}\frac{1}{n+1}{\bf E}\big\{\sum_{t=0}^n\gamma(X_t,Y_{t-1})\big\}={\bf E}\big\{\gamma(X_0,Y_{-1})\big\},~(x_0,y_{-1})\in{\cal X}\times{\cal Y}\nonumber\\
&=\frac{1-\gamma(1+2^{\mu^s_0+\Delta{K}^{s,\infty}})}{(\alpha-\gamma)\big(1+2^{\mu^s_0+\mu^s_1+2\Delta{K}^{s,\infty}}+2^{\mu^s_0+1+\Delta{K}^{s,\infty}}\big)}+\frac{2^{\mu^s_0+\Delta{K}^{s,\infty}}\big(\beta(1+2^{\mu^s_1+\Delta{K}^{s,\infty}})-1\big)}{(\beta-\delta)\big(1+2^{\mu^s_0+\mu^s_1+2\Delta{K}^{s,\infty}}+2^{\mu^s_0+1+\Delta{K}^{s,\infty}}\big)}\nonumber\\
&=\kappa,~\kappa\in[0,\kappa_{max}].\nonumber
\end{align}

\subsubsection{Numerical Evaluations} Fig.~\ref{fig:bsumco:cost} depicts numerical simulations of the optimal (nonstationary) channel input conditional distribution and the corresponding channel output transition probability distribution given by \eqref{section:example:bumco:cost:equation1}-\eqref{section:example:bumco:cost:equation2b}, for a time-invariant channel $$BUMCO(\alpha_t,\beta_t,\gamma_t,\delta_t)=BUMCO(0.9,0.1,0.2,0.4)$$, with transmission cost given by \eqref{section:example:general:bumco:cost:function}, $s=0.05$, i.e., $\kappa=0.5992$, for $n=1000$.\\
\noi Fig.~\ref{fig:bumco:ftfi:capacity:cost} depicts the corresponding value of $\frac{1}{n+1}C_{X^n\rightarrow{Y^n}}^{FB,A.1}(\kappa)=\frac{1}{n+1}{\bf E}^{\pi^*}\Big\{\sum_{t=0}^n\log\big(\frac{q(y_t|y_{t-1},x_t)}{\nu^{\pi^*}(y_t|y_{t-1})}\big)\Big\}$, where $\{\pi_t^*(x_t|y_{t-1}):~t=0,1,\ldots,n\}$ is given by \eqref{section:example:bumco:cost:equation1}, for $n=1000$. From Fig.~\ref{fig:bumco:ftfi:capacity}, at $n\approx{1000}$, the constrained FTFI capacity for $s=0.05, \kappa=0.5992$ is $\frac{1}{n+1}C_{X^n\rightarrow{Y^n}}^{FB,A.1}(\kappa)=0.2135~\mbox{bits/channel use}$, while the actual constrained feedback capacity evaluated by \eqref{remark:section:example:ergodic:feedback:capacity:cost} for $s=0.05$ and $\kappa=0.5992$ is $C^{FB,A.1}(\kappa)=0.2137~\mbox{bits/channel use}$.

\begin{figure}
    \centering
    \begin{subfigure}[b]{0.49\textwidth}
        \centering
        \includegraphics[width=\textwidth]{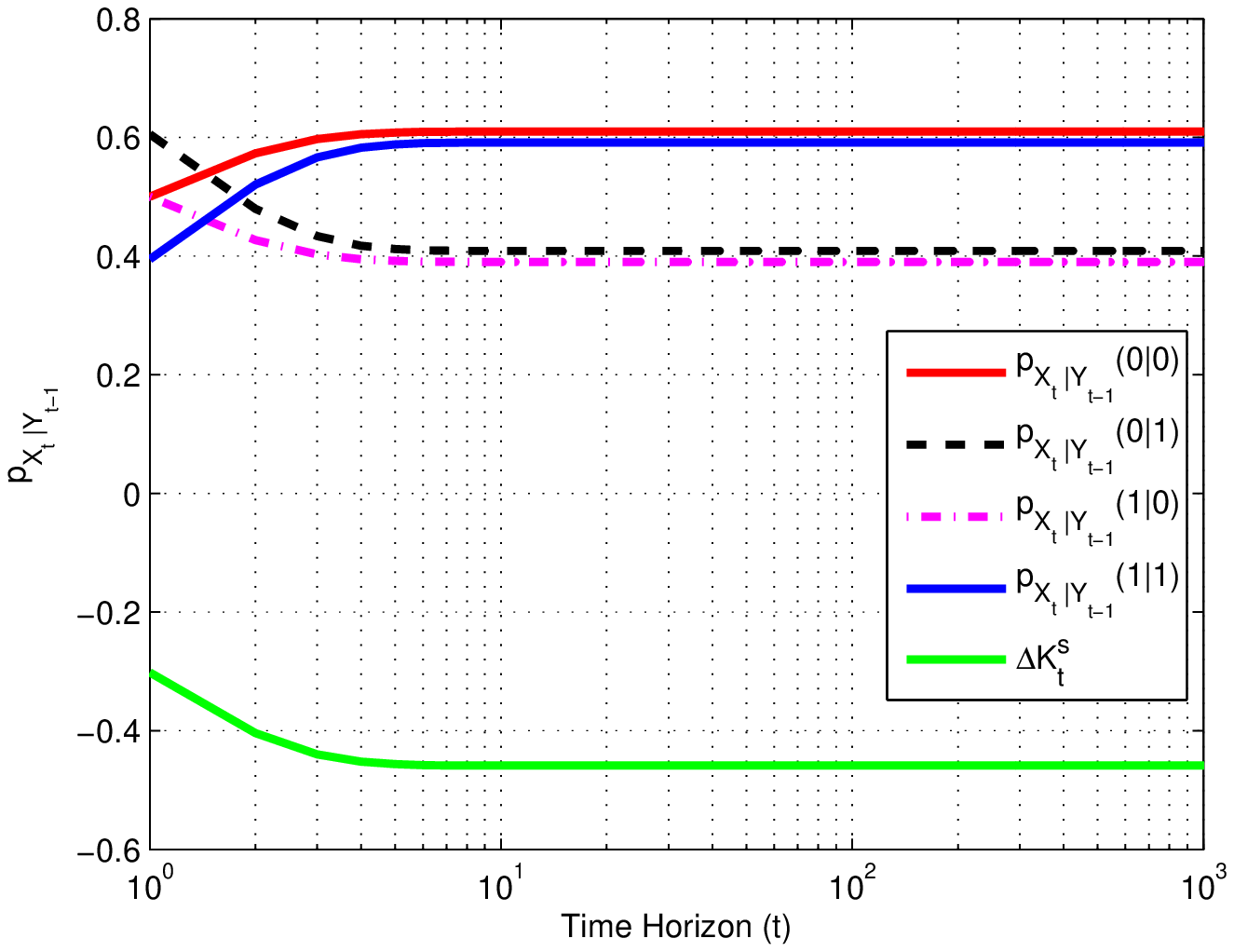}
        \caption{Optimal Distributions $\pi_t^*(x_t|y_{t-1})$ and $\Delta{K}^s_t$.}\label{fig:bsumco:optimal:input:distribution:cost}
    \end{subfigure}
    \hfill
    \begin{subfigure}[b]{0.49\textwidth}
        \centering
        \includegraphics[width=\textwidth]{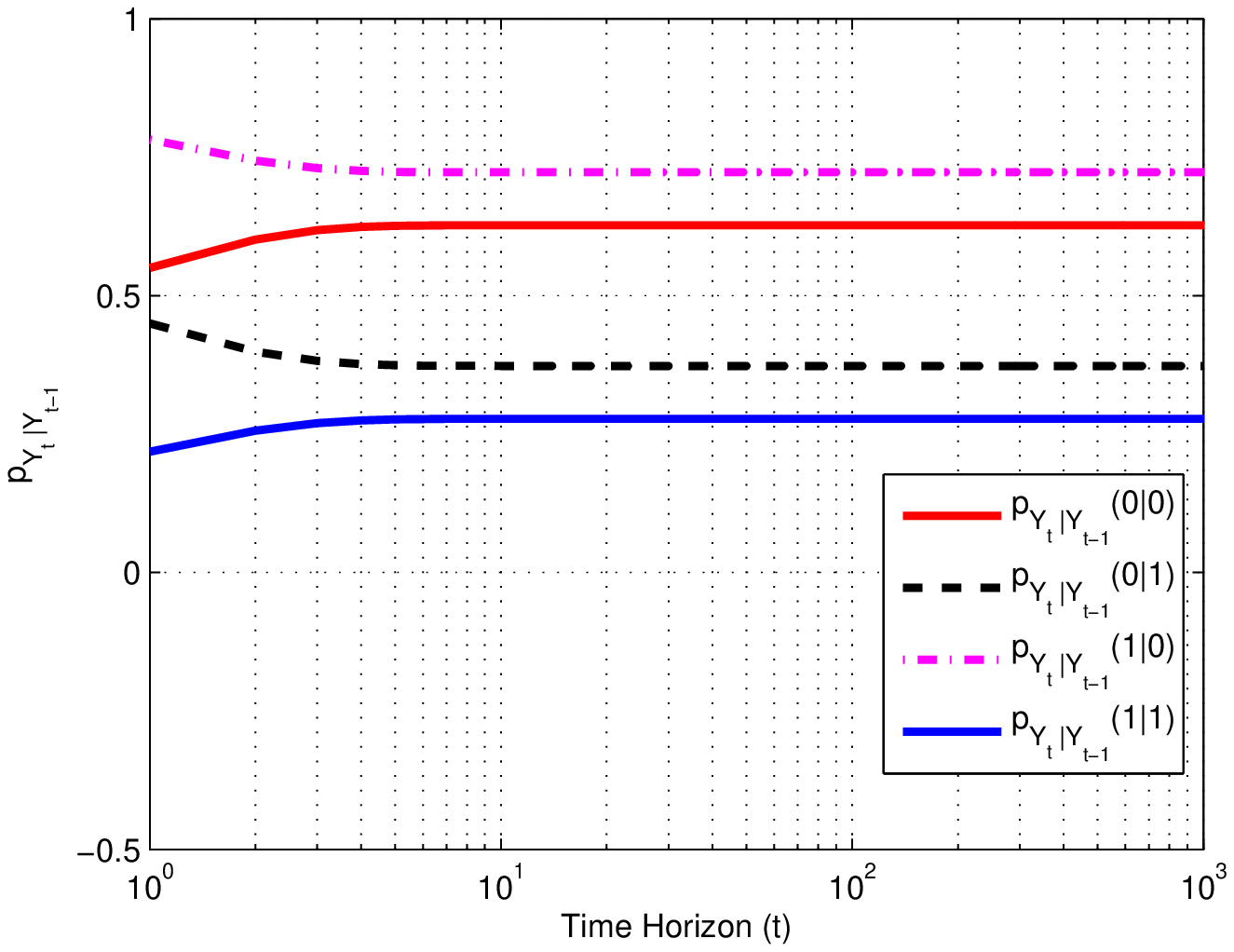}      
        \caption{Optimal Distributions $\nu_t^{\pi^*}(y_t|y_{t-1})$.}\label{fig:bsumco:cost:optimal:output:distribution}
    \end{subfigure}
    \caption{Optimal transition probability distributions of $BUMCO(0.9,0.1,0.2,0.4)$ with transmission cost function given by \eqref{section:example:general:bumco:cost:function}, $s=0.05$, for $n=1000$.}\label{fig:bsumco:cost}
\end{figure}
\begin{figure}[h]
\centering
\includegraphics[width=0.49\textwidth]{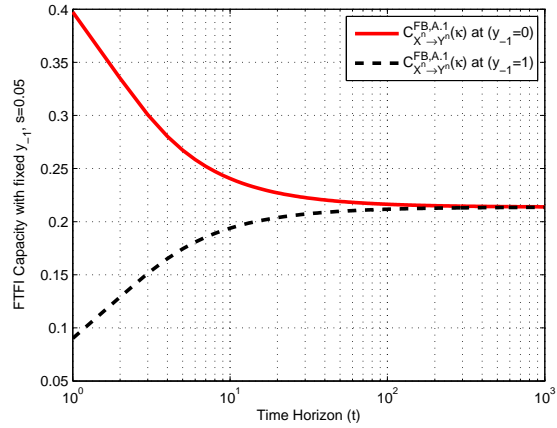}
\caption{$\frac{1}{n+1}C^{FB,A.1}_{X^n\rightarrow{Y^n}}(\kappa)$ of BUMCO ($0.9$, $0.1$, $0.2$, $0.4$), $s=0.05$, $\kappa=0.5992$, for $n=1000$ with a choice of the initial distribution ${\bf P}_{Y_{-1}}(y_{-1}=0)=0$ with its complement ${\bf P}_{Y_{-1}}(y_{-1}=1)=1$.}
\label{fig:bumco:ftfi:capacity:cost}
\end{figure}

%
%
%

\subsection{The FTFI Capacity of Time-Varying BEUMCO}\label{subsection:applications:ftfi:capacity:beumco}

\par In this subsection, we apply Theorem~\ref{introduction:baa:sequential:theorem:lmco:necessary:sufficient}, for $M=1$, to derive closed form expressions for the optimal channel input conditional distribution and the corresponding output transition probability distribution of time-varying $\{BEUMCO(\alpha_t,\gamma_t,\beta_t):t~\in\mathbb{N}_0^n\}$ channel defined by
\begin{align}
q_t(dy_t|y_{t-1},x_t)=\bordermatrix{&0,0&e,0&1,0&0,1&e,1&1,1\cr
            0&\alpha_t&\gamma_t&\beta_t&0&0&0\cr
            e&1-\alpha_t&1-\gamma_t&1-\beta_t&1-\alpha_t&1-\gamma_t&1-\beta_t\cr
            1&0&0&0&\alpha_t&\gamma_t&\beta_t\cr},~ \alpha_t, \beta_t, \gamma_t\in[0,1].\label{introduction:example:conjecture:matrix:beumco}
\end{align}
\noi The results given in the next theorem, state that feedback does not increase the FTFI capacity of this channel.

\begin{theorem}(Optimal solution of the characterization of FTFI capacity of time-varying BEMCO)\label{section:example:theorem:beumco}{\ \\}
Consider the $\{$BEUMCO$(\alpha_t$, $\gamma_t,\beta_t):~t\in\mathbb{N}_0^n\}$ defined in \eqref{introduction:example:conjecture:matrix:beumco}.
\begin{itemize}
\item[(a)] The optimal channel input conditional distribution and the corresponding output transition probability distribution of the characterization of FTFI capacity $C^{FB,A.1}_{X^n\rightarrow{Y^n}}$, i.e., \eqref{introduction:cor-ISR_B.2} with $M=1$,  when $\{{\pi}^*_{t}(x_t|y_{t-1})\neq{0},~\forall{x_t\in{\cal X}_t},~t\in\mathbb{N}_0^n\}$, are given by the following expressions. 
\begin{subequations}\label{introduction:example:matrix:beumco:input-output}
\begin{align}
&{\pi}^*_{t}(x_t|y_{t-1})\equiv{\pi}^*_{t}(x_t)=\bordermatrix{&\cr
            0&{\pi}^*_{t}(0)\cr
            1&{\pi}^*_{t}(1)\cr}, \forall{y_{t-1}\in{\cal Y}_{t-1}},~t\in\mathbb{N}_0^n,\label{theorem:conjecture:example:beumco:equation1}\\
&\nu_{t}^{\pi^*}(y_t|y_{t-1})=\bordermatrix{&0&e&1\cr
            0&\nu_{t}^{\pi^*}(0|0)&\nu_{t}^{\pi^*}(0|e)&\nu_{t}^{\pi^*}(0|1)\cr
            e&\nu_{t}^{\pi^*}(e|0)&\nu_{t}^{\pi^*}(e|e)&\nu_{t}^{\pi^*}(e|1)\cr
            1&\nu_{t}^{\pi^*}(1|0)&\nu_{t}^{\pi^*}(1|e)&\nu_{t}^{\pi^*}(1|1)\cr},~t\in\mathbb{N}_0^n\label{theorem:conjecture:example:beumco:equation1a}
\end{align}
\end{subequations} 
where
\begin{subequations}\label{introduction:conjecture:example:beumco:equation2a}
\begin{align}
{\pi}^*_{t}(0)&=\frac{2^{\Delta{C}^1_{t+1}}}{1+2^{\Delta{C}^1_{t+1}}},~&~{\pi}^*_{t}(1)&=\frac{1}{1+2^{\Delta{C}^1_{t+1}}},\\
\nu_{t}^{\pi^*}(0|0)&=\frac{\alpha_t2^{\Delta{C}^1_{t+1}}}{1+2^{\Delta{C}^1_{t+1}}},~&~\nu_{t}^{\pi^*}(0|e)&=\frac{\gamma_t2^{\Delta{C}^1_{t+1}}}{1+2^{\Delta{C}^1_{t+1}}},~&~\nu_{t}^{\pi^*}(0|1)&=\frac{\beta_t2^{\Delta{C}^1_{t+1}}}{1+2^{\Delta{C}^1_{t+1}}},\\
\nu_{t}^{\pi^*}(e|0)&=1-\alpha_t,~&~\nu_{t}^{\pi^*}(e|e)&=1-\gamma_t,~&~\nu_{t}^{\pi^*}(e|1)&=1-\beta_t,\\
\nu_{t}^{\pi^*}(1|0)&=\frac{\alpha_t}{1+2^{\Delta{C}^1_{t+1}}},~&~\nu_{t}^{\pi^*}(1|e)&=\frac{\gamma_t}{1+2^{\Delta{C}^1_{t+1}}},~&~\nu_{t}^{\pi^*}(1|1)&=\frac{\beta_t}{1+2^{\Delta{C}^1_{t+1}}}
\end{align}
\end{subequations}
and $\{\Delta{C}^1_t(\alpha_t,\gamma_t,\beta_t)\equiv{\Delta}C^1_t\triangleq{C}_t(0)-C_t(1):~t\in\mathbb{N}_0^{n+1}\}$ is the difference of the value functions $\{{C}_t(0),~C_t(1):~t\in\mathbb{N}_0^{n+1}\}$ at each time, satisfying the following backward recursions.
\begin{align}
\Delta{C}^1_{t}=(\alpha_t-\beta_t)\Big(\Delta{C}^2_{t+1}+\log\big(1+2^{\Delta{C}^1_{t+1}}\big)\Big),~\Delta{C}^1_{n+1}=0,~t\in\{n,\ldots,0\},\label{introduction:conjecture:example:beumco:equation2}
\end{align}
with $\{\Delta{C}^2_t(\alpha_t,\gamma_t,\beta_t)\equiv{\Delta}C^2_t\triangleq{C}_t(1)-C_t(e):~t\in\mathbb{N}_0^{n+1}\}$ is the difference of the value functions $\{{C}_t(1),~C_t(e):~t\in\mathbb{N}_0^{n+1}\}$ at each time, satisfying the following backward recursions
\begin{align}
\Delta{C}^2_{t}=(\beta_t-\gamma_t)\Big(\Delta{C}^2_{t+1}+\log\big(1+2^{\Delta{C}^1_{t+1}}\big)\Big),~\Delta{C}^2_{n+1}=0,~t\in\{n,\ldots,0\}.  \label{introduction:conjecture:example:beumco:equation3}
\end{align}
\item[(b)] The solution of the value functions is given recursively by the following expressions.
\begin{align}
C_{t}(0)&=\alpha_tC_{t+1}(1)+(1-\alpha_t)C_{t+1}(e)+\alpha_t\log(1+2^{\Delta{C}^1_{t+1}})-H(\alpha_{t}),~C_{n+1}(0)=0,\label{introduction:theorem:beumco:value:function:eq1}\\
C_{t}(e)&=\gamma_tC_{t+1}(1)+(1-\gamma_t)C_{t+1}(e)+\gamma_t\log(1+2^{\Delta{C}^1_{t+1}})-H(\alpha_{t}),~C_{n+1}(e)=0,\label{introduction:theorem:beumco:value:function:eq2}\\
~C_{t}(1)&=\beta_tC_{t+1}(1)+(1-\beta_t)C_{t+1}(e)+\beta_t\log(1+2^{\Delta{C}^1_{t+1}})-H(\alpha_{t}),~C_{n+1}(1)=0,~t\in\{n,\ldots,0\}.\label{introduction:theorem:beumco:value:function:eq3}
\end{align}
\item[(c)] The characterization of the FTFI capacity is given by
\begin{align*}
C^{FB,A.1}_{X^n\rightarrow{Y^n}}=\sum_{y_{-1}\in\{0,e,1\}}C_0(y_{-1})\mu(y_{-1}),~\mu(y_{-1})~\mbox{is fixed}.
\end{align*}
\end{itemize}
\end{theorem}
\begin{proof}
The derivation is similar to the one of subsubsection~\ref{proof:ftfi:feedback:capacity}, hence we omit it.
\end{proof}

For Theorem~\ref{section:example:theorem:beumco}, \eqref{theorem:conjecture:example:beumco:equation1}, it follows that feedback does not increase the characterization of FTFI capacity, and consequently feedback capacity.

\subsubsection{Time-Invariant BEUMCO}\label{subsubsection:time-invariant:beumco}  Here, we discuss the results of Theorem~\ref{section:example:theorem:beumco}, when the channel is time-invariant, i.e., $BEUMCO(\alpha_t,\gamma_t,\beta_t)=BEUMCO(\alpha,\gamma,\beta)$. The steady state versions of \eqref{introduction:conjecture:example:beumco:equation2}, \eqref{introduction:conjecture:example:beumco:equation3}, are defined by the following algebraic equations.
\begin{align}
\Delta{C}^{1,\infty}=&(\alpha-\beta)\Big(\Delta{C}^{2,\infty}+\log\big(1+2^{\Delta{C}^{1,\infty}}\big)\Big)\label{section:steady:state:dc1:beumco:equation1}\\
\Delta{C}^{2,\infty}=&(\beta-\gamma)\Big(\Delta{C}^{2,\infty}+\log\big(1+2^{\Delta{C}^{1,\infty}}\big)\Big).  \label{section:steady:state:dc2:beumco:equation1}
\end{align}
After some algebra, it can be shown that the solutions of the nonlinear equation \eqref{section:steady:state:dc1:beumco:equation1} is given by 
\begin{align}
\Delta{C}^{1,\infty}=\Big({\frac{\alpha-\beta}{1-(\beta-\gamma)}}\Big)\log(1+2^{\Delta{C}^{1,\infty}}).\label{section:steady:state:beumco:equation1}
\end{align}
Moreover, the time-invariant versions of \eqref{theorem:conjecture:example:beumco:equation1}-\eqref{theorem:conjecture:example:beumco:equation1a} denoted by ${\pi}_t^{*}(x_t)\equiv{\pi}^{*,\infty}(x_t)$ and $\nu_t^{\pi^{*}}(y_t|y_{t-1})\equiv\nu^{\pi^{*,\infty}}(y_t|y_{t-1})$, are given as follows. 
\begin{subequations}
\begin{align}
{\pi}^{*,\infty}(0)&=\frac{2^{\Delta{C}^{1,\infty}}}{1+2^{\Delta{C}^{1,\infty}}},~&~{\pi}^{*,\infty}(1)&=1-{\pi}^{*,\infty}(0),\label{section:steady-state:beumco:input-output}\\
\nu^{\pi^{*,\infty}}(0|0)&=\frac{\alpha2^{\Delta{C}^{1,\infty}}}{1+2^{\Delta{C}^{1,\infty}}},~&~\nu^{\pi^{*,\infty}}(0|e)&=\frac{\gamma2^{\Delta{C}^{1,\infty}}}{1+2^{\Delta{C}^{1,\infty}}},~&~\nu^{\pi^{*,\infty}}(0|1)&=\frac{\beta2^{\Delta{C}^{1,\infty}}}{1+2^{\Delta{C}^{1,\infty}}},\label{section:steady-state:beumco:input-output1}\\
\nu^{\pi^{*,\infty}}(e|0)&=1-\alpha,~&~\nu^{\pi^{*,\infty}}(e|e)&=1-\gamma,~&~\nu^{\pi^{*,\infty}}(e|1)&=1-\beta,\label{section:steady-state:beumco:input-output2}\\
\nu^{\pi^{*,\infty}}(1|0)&=\frac{\alpha}{1+2^{\Delta{C}^{1,\infty}}},~&~\nu^{\pi^{*,\infty}}(1|e)&=\frac{\gamma}{1+2^{\Delta{C}^{1,\infty}}},~&~\nu^{\pi^{*,\infty}}(1|1)&=\frac{\beta}{1+2^{\Delta{C}^{1,\infty}}}\label{section:steady-state:beumco:input-output3}
.\end{align}
\end{subequations}
It can be shown that the channel output transition probability distribution given by \eqref{section:steady-state:beumco:input-output1}-\eqref{section:steady-state:beumco:input-output3}, has a unique invariant distribution $\{\nu^{{\pi^{*,\infty}}}(y):~y\in\{0,e,1\}\}$ given by
\begin{align}
\nu^{{\pi^{*,\infty}}}(0)=&\frac{\gamma2^{\Delta{C}^{1,\infty}}}{1-(\beta-\gamma)+2^{\Delta{C}^{1,\infty}}(1-\alpha+\gamma)},~\nu^{{\pi^{*,\infty}}}(e)=\frac{1-\beta+2^{\Delta{C}^{1,\infty}}(1-\alpha)}{1-(\beta-\gamma)+2^{\Delta{C}^{1,\infty}}(1-\alpha+\gamma)},\nonumber\\
\nu^{{\pi^{*,\infty}}}(1)=&\frac{\gamma}{1-(\beta-\gamma)+2^{\Delta{C}^{1,\infty}}(1-\alpha+\gamma)}.\nonumber
\end{align}
Hence, the feedback capacity of time-invariant $BEUMCO(\alpha,\gamma,\beta$) is given by the following expression.
\begin{align}
C^{FB,A.1}=\sum_{y\in\{0,e,1\}}\bigg(\sum_{x\in\{0,1\},z\in\{0,e,1\}}\log\Big(\frac{q(z|y,x)}{\nu^{*,\infty}(z|y)}\Big)q(z|y,x)\pi^{*,\infty}(x|y)\bigg)\nu^{{\pi^{*,\infty}}}(y).
\end{align}
After some algebra, we obtain the following
\begin{align}
C^{FB,A.1}=(1-\nu_e)\log(1+2^{\Delta{C}^{1,\infty}})-\nu_0\Delta{C}^{1,\infty}\label{section:beumco:ergodic:feedback:capacity}
\end{align}
where
\begin{align}
\nu_e=\nu^{{\pi^{*,\infty}}}(e),~~~\nu_0=\nu^{{\pi^{*,\infty}}}(0).\nonumber
\end{align}

\subsubsection{Numerical evaluations} Fig.~\ref{fig:beumco:without:cost} depicts numerical simulations of the optimal (nonstationary) channel input conditional distribution and the corresponding channel output transition probability distribution given by \eqref{section:steady-state:beumco:input-output1}-\eqref{section:steady-state:beumco:input-output3}, for a time-invariant channel $BEUMCO(\alpha,\gamma,\beta)=BEUMCO(0.95,0.6,0.8)$, for $n=1000$.\\
\noi Fig.~\ref{fig:beumco:ftfi:capacity} depicts the corresponding value of $\frac{1}{n+1}C_{X^n\rightarrow{Y^n}}^{FB,A.1}=\frac{1}{n+1}{\bf E}^{\pi^*}\Big\{\sum_{t=0}^n\log\big(\frac{q(y_t|y_{t-1},x_t)}{\nu^{\pi^*}(y_t|y_{t-1})}\big)\Big\}$, where $\{\pi_t^*(x_t|y_{t-1})\equiv\pi_t^*(x_t):~t=0,1,\ldots,n\}$ is given by \eqref{section:steady-state:beumco:input-output1}-\eqref{section:steady-state:beumco:input-output3}, for $n=1000$. From Fig.~\ref{fig:beumco:ftfi:capacity}, at $n\approx{1000}$, the FTFI capacity is $\frac{1}{n+1}C_{X^n\rightarrow{Y^n}}^{FB,A.1}=0.8306~\mbox{bits/channel use}$, while the actual ergodic feedback capacity evaluated from \eqref{section:beumco:ergodic:feedback:capacity} is $C^{FB,A.1}=0.8307~\mbox{bits/channel use}$. \\
Based on our simulations, it is interesting to note that the optimal channel input conditional distribution and the corresponding channel output transition probability converge to their asymptotic limits at $n\approx{6}$, with respect to an error tolerance of $10^{-4}$.

\begin{figure}
    \centering
    \begin{subfigure}[b]{0.49\textwidth}
        \centering
        \includegraphics[width=\textwidth]{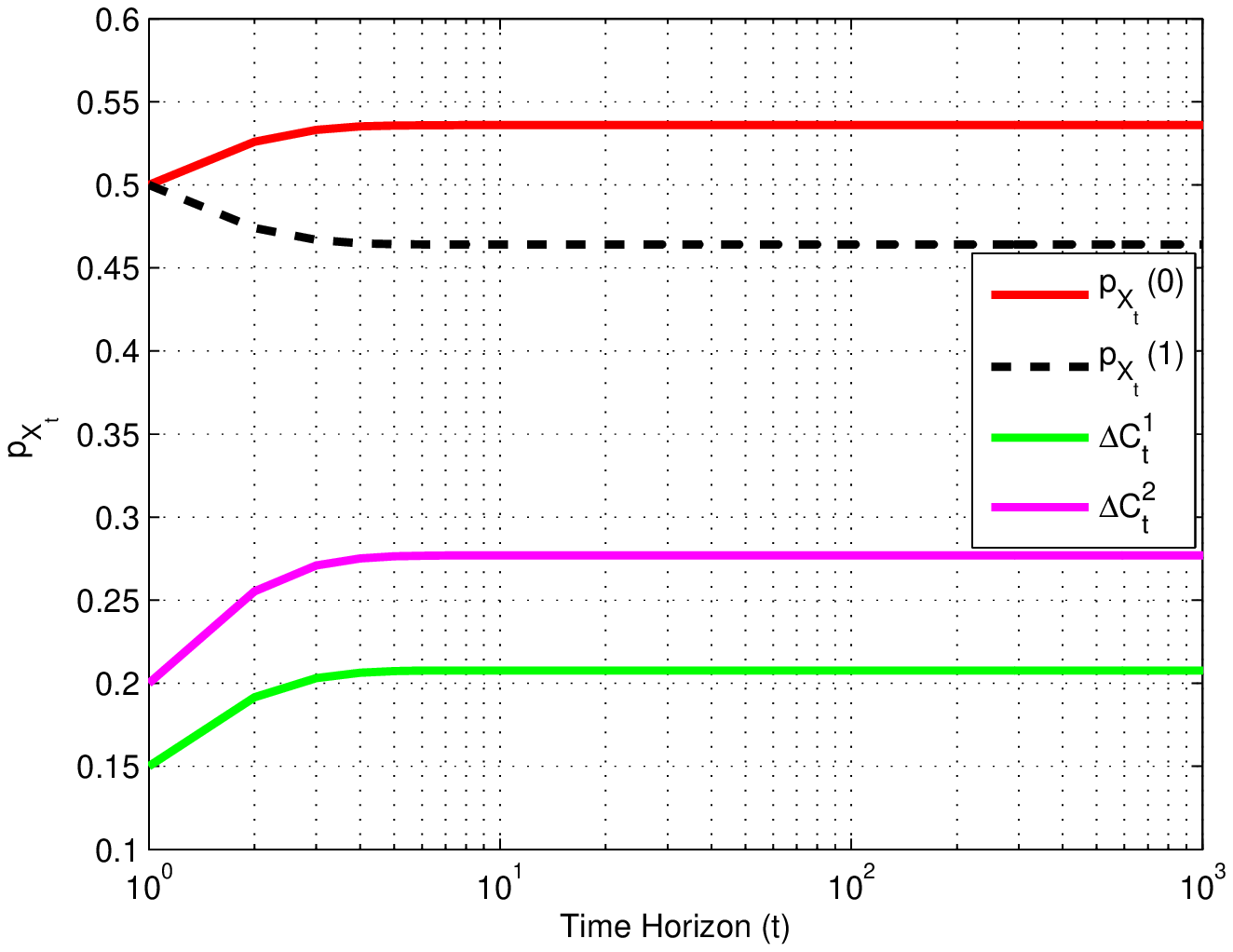}
        \caption{Optimal Distributions $\pi_t^*(x_t|y_{t-1})\equiv\pi_t^*(x_t)$ and $\Delta{C}^1_t, \Delta{C}^2_t$.}\label{fig:beumco:optimal:input:distribution}
    \end{subfigure}
    \hfill
    \begin{subfigure}[b]{0.49\textwidth}
        \centering
        \includegraphics[width=\textwidth]{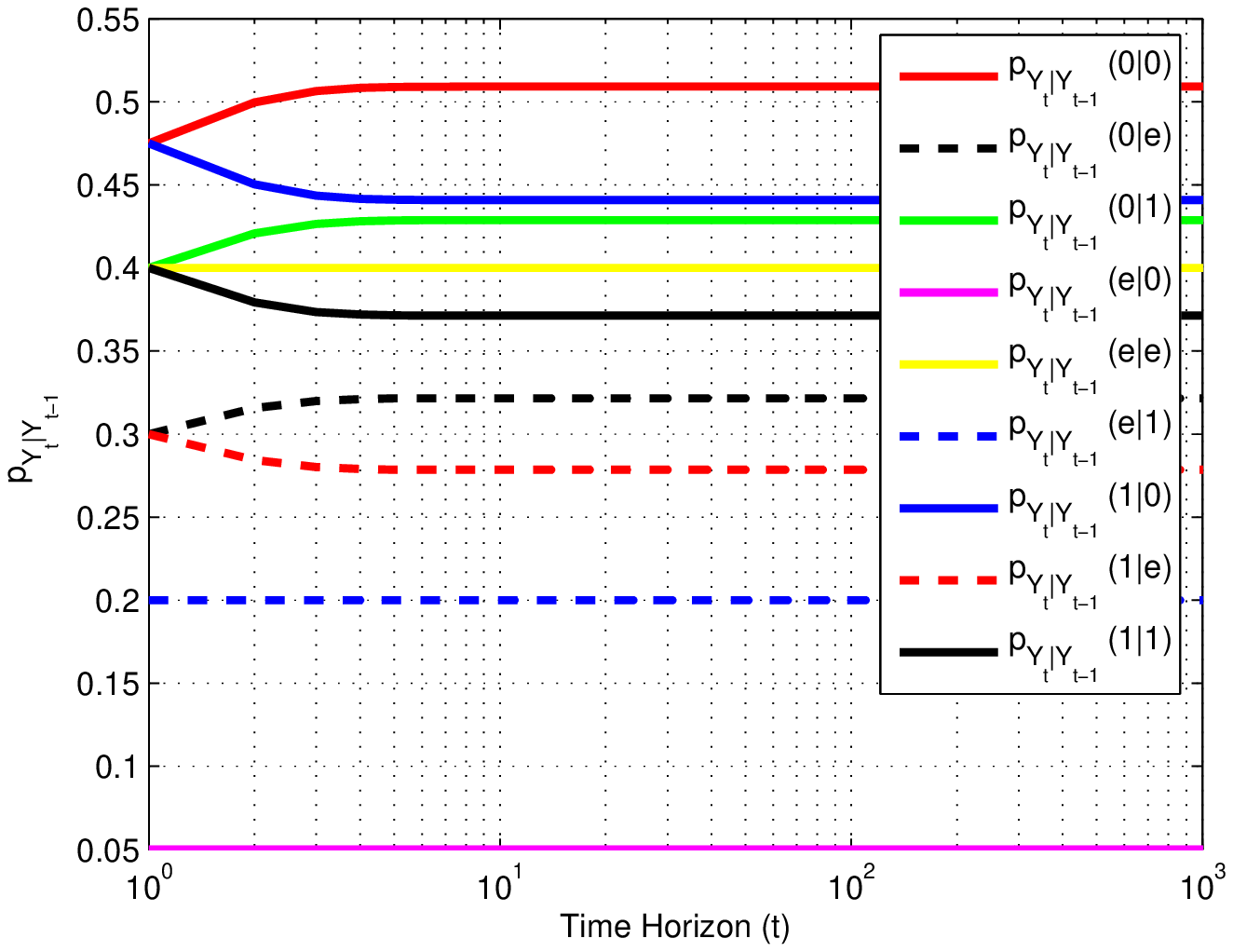}      
        \caption{Optimal Distributions $\nu_t^{\pi^*}(y_t|y_{t-1})$.}\label{fig:beumco:optimal:output:distribution}
    \end{subfigure}
    \caption{Optimal transition probability distributions of $BEUMCO(0.95,0.6,0.8)$ for $n=1000$.}\label{fig:beumco:without:cost}
\end{figure}
\begin{figure}[h]
\centering
\includegraphics[width=0.49\textwidth]{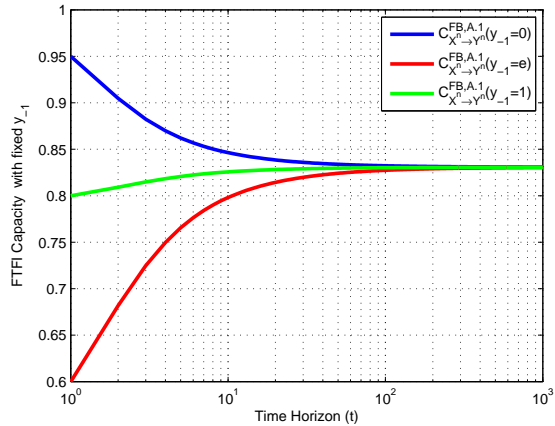}
\caption{ $\frac{1}{n+1}C^{FB,A.1}_{X^n\rightarrow{Y^n}}$ of BEUMCO ($0.95$, $0.6$, $0.8$) for $n=1000$ with a choice of the initial distribution ${\bf P}_{Y_{-1}}(y_{-1}=0)=1$ with its complements ${\bf P}_{Y_{-1}}(y_{-1}=e)=0$ ${\bf P}_{Y_{-1}}(y_{-1}=1)=0$.}
\label{fig:beumco:ftfi:capacity}
\end{figure}

\subsubsection{Special Cases of Theorem~\ref{section:example:theorem:beumco}}\label{example:beumco:special case} Next, we discuss certain degenerated cases.
\begin{itemize}
\item[$\bullet$] For the time-invariant channel $BEUMCO(1-\alpha,\gamma,1-\alpha)$, by \eqref{section:steady-state:beumco:input-output} the optimal channel input conditional distribution is uniform, the corresponding output transition probability distribution is stationary, and the ergodic feedback capacity is equal to the corresponding no-feedback capacity given by 
\begin{align}
C^{NFB,A.1}=C^{FB,A.1}=\frac{\gamma}{\alpha+\gamma}.\label{example:beumco:special:case}
\end{align}
\item[$\bullet$] For the channel $BEUMCO(1-\alpha,1-\alpha,1-\alpha)$, the channel is memoryless, and it degenerates to the well-known memoryless Binary Erasure Channel (BEC), where the optimal channel input distribution is uniform \cite{cover-thomas2006}. This follows from \eqref{example:beumco:special:case}, by setting $\gamma=1-\alpha$.
\end{itemize}

%
%
%

\subsection{The FTFI Capacity of Time-Varying BSTMCO}\label{subsection:applications:ftfi:capacity:bstmco}

\par In this subsection, we apply Theorem~\ref{introduction:baa:sequential:theorem:lmco:necessary:sufficient}, for $M=2$, to derive closed form expressions for the optimal channel input conditional distribution and the corresponding channel output transition probability distribution of the time-varying $\{BSTMCO(\alpha_t,\beta_t,\gamma_t,\delta_t):~t\in\mathbb{N}_0^n\}$ channel defined by 
\begin{align}
q_t(dy_t|y_{t-1},y_{t-2},x_t)=&\bordermatrix{&0,0,0&0,0,1&0,1,0&0,1,1&1,0,0&1,0,1&1,1,0&1,1,1\cr
            0&\alpha_t&\beta_t&\gamma_t&\delta_t&1-\delta_t&1-\gamma_t&1-\beta_t&1-\alpha_t\cr
            1&1-\alpha_t&1-\beta_t&1-\gamma_t&1-\delta_t&\delta_t&\gamma_t&\beta_t&\alpha_t\cr}, \label{introduction:section:example:matrix:general:bstmco} \\
            &~~\hspace*{5cm}\alpha_t, \beta_t, \gamma_t, \delta_t\in[0,1],~t=0,\ldots,n.\nonumber 
\end{align}
The results are given in the next theorem.
\begin{theorem}(Optimal solution of the characterization of time-varying BSTMCO)\label{section:example:theorem:bstmco}{\ \\}
Consider the $\{BSTMCO(\alpha_t,\beta_t,\gamma_t,\delta_t):~t\in\mathbb{N}_0^n\}$ defined in \eqref{introduction:section:example:matrix:general:bstmco}. Then the following hold.
\begin{itemize}
\item[(a)] The optimal channel input distribution and the corresponding channel output transition probability distribution, of the characterization of $C^{FB,A.2}_{X^n\rightarrow{Y^n}}$, i.e., \eqref{introduction:cor-ISR_B.2} with $M=2$,  denoted by  $\big\{{\pi}_t^*(x_t|y_{t-1},y_{t-2}):  (x_t, y_{t-1},y_{t-2})\in \{0,1\} \times \{0,1\}\times\{0,1\}, t\in\mathbb{N}_0^n\big\},~\big\{ {\nu}_t^{\pi^*}(y_t|y_{t-1},y_{t-2}):  (y_t, y_{t-1},y_{t-2})\in \{0,1\} \times \{0,1\}\times \{0,1\}, t\in\mathbb{N}_0^n\big\}$ are the following.
\begin{subequations}
\begin{align}
&{\pi}^*_t(0|0,0)={\pi}^*_t(1|1,1)=\frac{1-\beta_t(1+2^{\mu_0(t)+\Delta{C}_{t+1}})}{(\alpha_t-\beta_t)(1+2^{\mu_0(t)+\Delta{C}_{t+1}})},\\
&{\pi}^*_t(0|0,1)={\pi}^*_t(1|1,0)=\frac{1-\delta_t(1+2^{\mu_1(t)+\Delta{C}_{t+1}})}{(\gamma_t-\delta_t)(1+2^{\mu_1(t)+\Delta{C}_{t+1}})},\\
&{\pi}^*_t(0|1,0)={\pi}^*_t(1|0,1)=\frac{\gamma_t(1+2^{\mu_1(t)+\Delta{C}_{t+1}})-1}{(\gamma_t-\delta_t)(1+2^{\mu_1(t)+\Delta{C}_{t+1}})},\\ 
&{\pi}^*_t(0|1,1)={\pi}^*_t(1|0,0)=\frac{\alpha_t(1+2^{\mu_0(t)+\Delta{C}_{t+1})-1}}{(\alpha_t-\beta_t)(1+2^{\mu_0(t)+\Delta{C}_{t+1}})},\\
\nu_t^{\pi^*}(0|0,0)&=\nu_t^{\pi^*}(1|1,1)=\frac{1}{1+2^{\mu_0(t)+\Delta{C}_{t+1}}},~~\nu_t^{\pi^*}(0|0,1)=\nu_t^{\pi^*}(1|1,0)=\frac{1}{1+2^{\mu_1(t)+\Delta{C}_{t+1}}},\\
\nu_t^{\pi^*}(1|0,0)&=\nu_t^{\pi^*}(0|1,1)=\frac{2^{\mu_0(t)+\Delta{C}_{t+1}}}{1+2^{\mu_0(t)+\Delta{C}_{t+1}}},~~\nu_t^{\pi^*}(1|0,1)=\nu_t^{\pi^*}(0|1,0)=\frac{2^{\mu_1(t)+\Delta{C}_{t+1}}}{1+2^{\mu_1(t)+\Delta{C}_{t+1}}},\\
\mu_0(\alpha_t,\beta_t)&=\frac{H(\beta_t)-H(\alpha_t)}{\beta_t-\alpha_t}\equiv{\mu}_0(t),~~~\mu_1(\gamma_t,\delta_t)=\frac{H(\delta_t)-H(\gamma_t)}{\delta_t-\gamma_t}\equiv\mu_1(t),\end{align}\label{section:example:bstmco:equation1}
\end{subequations}
$\{\Delta{C}_t(\alpha_t,\beta_t,\gamma_t,\delta_t)\equiv{\Delta}C_t\triangleq{C}_t(1,1)-C_t(0,1):~t\in\mathbb{N}_0^{n+1}\}$ satisfies the following backward recursions.
\begin{subequations}
\begin{align}
\Delta{C}_{n+1}=&0,\\
\Delta{C}_{t}=&\big(\mu_1(t)(\gamma_t-1)-\mu_0(t)(\alpha_t-1)\big)+H(\alpha_t)-H(\gamma_t)\nonumber\\
&\qquad\qquad+\log\Big(\frac{1+2^{\mu_1(t)+\Delta{C}_{t+1}}}{1+2^{\mu_0(t)+\Delta{C}_{t+1}}}\Big),~~t\in\{n,\ldots,0\}.  
\end{align}\label{section:example:bstmco:equation2b}
\end{subequations}
\item[(b)] The solution of the value function is given recursively by the following expressions.
\begin{align}
&C_{t}(1,1)=C_t(0,0)=\mu_0(t)(\alpha_{t}-1)+C_{t+1}(0,0)+\log(1+2^{\mu_0(t)+\Delta{C}_{t+1}})\nonumber\\
&\hspace*{5cm}-H(\alpha_{t}),~C_{n+1}(1,1)=C_{n+1}(0,0)=0,\label{theorem:value:function:bstmco:eq1}\\
~&C_{t}(0,1)=C_t(1,0)=\mu_1(t)(\beta_{t}-1)+{C}_{t+1}(0,0)+\log(1+2^{\mu_1(t)+\Delta{C}_{t+1}})\nonumber\\
&\hspace*{5cm}-H(\beta_{t}),~C_{n+1}(0,1)=C_{n+1}(1,0)=0,~t\in\{n,\ldots,0\}.\label{theorem:value:function:bstmco:eq2}
\end{align}
\item[(c)] The characterization of the FTFI capacity is given by
\begin{align*}
C^{FB,A.2}_{X^n\rightarrow{Y^n}}=\sum_{y_{-1}\in\{0,1\},y_{-2}\in\{0,1\}}C_t(y^{-1}_{-2})\mu(y_{-2}^{-1}),~\mu(y_{-2}^{-1})~\mbox{is fixed}.
\end{align*}
\end{itemize}
\end{theorem}
\begin{proof}
The derivation is similar to the one of subsubsection~\ref{proof:ftfi:feedback:capacity}, hence we omit it.
\end{proof}

\subsubsection{Discussion on Theorem~\ref{section:example:theorem:bstmco}} 
Theorem~\ref{section:example:theorem:bstmco} illustrates that the channel symmetry, when $y_{t-2}=0$ or $y_{t-2}=1$,~$t\in\mathbb{N}^n_0$, imposes a symmetry on the structure of the optimal channel input conditional distribution. 

\begin{remark}(Discussion of the results){\ \\}
 Next, we make some observations regarding the results obtained in subsection~\ref{subsection:applications:ftfi:capacity:bumco} and in subsection~\ref{subsection:applications:ftfi:capacity:beumco}. \\
If  $card({\cal X})=T$ and $card({\cal Y})=S$, where $T, S \geq 3$ then it is very hard and sometimes impossible to find closed form expressions for the optimal channel input distributions corresponding to $C^{FB,A.M}_{X^n\rightarrow{Y^n}}$. However, the necessary and sufficient conditions of Theorem~\ref{theorem:necessary:sufficient:ftfi:lmco:cost} are simplified considerably, when the channel distribution has certain symmetry similar to the one in Theorem~\ref{section:example:theorem:bstmco}, and for such channels closed form expressions are expected.
\end{remark}

%
%
%

\section{Generalizations to Abstract Alphabet Spaces}\label{generalizations:abstract:alphabets}

\par The theorems of Section~\ref{section:generalizations:lmco} extend to abstract alphabet spaces (i.e., countable, continuous alphabets etc.). However, for these  extensions to hold,  it is necessary to impose sufficient conditions related to the existence of an optimal channel input conditional distribution, G\^ateaux differentiability of directed information functional, and continuity with respect to channel input conditional distribution. 

\noi Below, we state sufficient conditions for Theorem~ \ref{theorem:necessary:sufficient:ftfi:lmco:cost} to hold on abstract alphabet spaces.
\begin{itemize}
\item[(C1)] $\{X_t:~t\in\mathbb{N}_0\}$, $\{Y_t:~t\in\mathbb{N}_0\}$ are complete separable metric spaces.
\item[(C2)] The directed information functional $\mathbb{I}_{X^n\rightarrow{Y^n}}(\overleftarrow{P}_{0,n},\overrightarrow{Q}_{0,n})$  (see \eqref{eqdi5}) is continuous on $\overleftarrow{P}_{0,n}(\cdot|y^{n-1})\in{\cal M}({\cal X}^n)$ for a fixed $\overrightarrow{Q}_{0,n}(\cdot|x^n)\in{\cal M}({\cal Y}^n)$.
\item[(C3)] There exist an optimal input distribution $\overleftarrow{P}^*_{0,n}(\cdot|y^{n-1})\in{\cal M}({\cal X}^n)$, which achieves the supremum of directed information.
\item[(C4)] The value function $\{C_{t}(y^{t-1}_{t-J}):~t\in\mathbb{N}_0^n\}$ is G\^ateaux differentiable with respect to $\{\pi_t(dx_t|y^{t-1}_{t-J}):~t\in\mathbb{N}_0^n\}$.
\end{itemize}
General theorems for the validity of (C2) and (C3) are derived in \cite{charalambous-stavrou2015ieeeit}.

%
%
%

\subsection{Channels of Class A and Transmission Cost of Class A}\label{subsection:extensions}
\noi Let $C_t: {\cal Y}^{t-1}_{t-J}\longmapsto[0,\infty)$ represent the maximum expected total pay-off in (\ref{characterization:ftfi:capacity:cost:lmco1}) on the future time horizon $\{t,t+1,\ldots,n\}$, given $Y^{t-1}_{t-J}=y^{t-1}_{t-J}$ at time $t-1$, defined by 
\begin{align}
&C_t(y^{t-1}_{t-J})=\sup_{\big\{\pi_i(dx_i|y^{i-1}_{i-J}):~i=t,t+1,\ldots,n\big\}}{\bf E}^{\pi}\bigg\{\sum_{i=t}^n\log\Big(\frac{dq_i(\cdot|y^{i-1}_{i-M},X_i)}{d{\nu}^{\pi}_{i}(\cdot|y^{i-1}_{i-J})}(Y_i)\Big)\nonumber\\
&\qquad-s\Big(\sum_{i=t}^n\gamma_i(x_i,y^{i-1}_{i-N})-(n+1)\kappa\Big)\Big{|}Y^{t-1}_{t-J}=y^{t-1}_{t-J}\bigg\}\label{algorithms:abstract:generalizations:lmco:equation17}
\end{align}

\noi By \eqref{algorithms:abstract:generalizations:lmco:equation17} we obtain the following dynamic programming recursions.
\begin{align}
&C_n(y^{n-1}_{n-J})=\sup_{\pi_n(dx_n|y^{t-1}_{t-J})}\Bigg\{\int_{{\cal X}_n\times{\cal Y}_n}\log\Big(\frac{dq_n(\cdot|y^{n-1}_{n-M},x_n)}{d{\nu}^{\pi}_{n}(\cdot|y^{n-1}_{n-J})}(y_n)\Big)q_n(dy_n|y^{n-1}_{n-M},x_n)\otimes{\pi}_n(dx_n|y^{n-1}_{n-J})\nonumber\\
&-s\Big(\int_{{\cal X}_n}\gamma_n(x_n,y^{n-1}_{n-N})\pi_n(dx_n|y^{n-1}_{n-J})-(n+1)\kappa\Big)\Bigg\},\label{algorithms:abstract:generalizations:lmco:equation18}\\
&C_t(y^{t-1}_{t-J})=\sup_{\pi_t(dx_t|y^{t-1}_{t-J})}\Bigg\{\int_{{\cal X}_t\times{\cal Y}_t}\Big(\log\big(\frac{dq_t(\cdot|y^{t-1}_{t-M},x_t)}{d{\nu}^{\pi}_{t}(\cdot|y^{t-1}_{t-J})}(y_t)\big)+C_{t+1}(y^t_{t+1-J})\Big)\nonumber\\
&q_t(dy_t|y^{t-1}_{t-M},x_t)\otimes{\pi}_t(dx_t|y^{t-1}_{t-J})-s\Big(\int_{{\cal X}_t}\gamma_t(x_t,y^{t-1}_{t-N})\pi_t(dx_t|y^{t-1}_{t-J})-(n+1)\kappa\Big)\Bigg\},~t\in\mathbb{N}_0^{n-1}.\label{algorithms:abstract:generalizations:lmco:equation19}
\end{align}

\noi Then, we have the following generalization of Theorem~\ref{theorem:necessary:sufficient:ftfi:lmco:cost} on abstract alphabets. 

\begin{theorem}(Sequential necessary and sufficient conditions on abstract spaces)\label{theorem:abstract:necessary:sufficient:ftfi:lmco:cost}{\ \\}
Suppose conditions (C1)-(C4) hold. The necessary and sufficient conditions for any input distribution $\{\pi_{t}(dx_t|y^{t-1}_{t-J}):~t\in\mathbb{N}_0^{n}\}$,~$J=\max\{M,N\}$, to achieve the supremum of the characterization of FTFI capacity given by \eqref{characterization:ftfi:capacity:cost:lmco1} are the following.\\
\noi{(a)}  For each $y^{n-1}_{n-J}\in{\cal Y}^{n-1}_{n-J}$, there exist a ${K}^s_n(y^{n-1}_{n-J})$, which depends on $s\geq{0}$, such that the following hold.
\begin{align}
&\int_{{\cal Y}_{n}}\Big(\log\big(\frac{dq_n(\cdot|y^{n-1}_{n-M},x_n)}{d\nu^{\pi}_{t}(\cdot|y^{n-1}_{n-J})}(y_n)\big)\Big)q_n(dy_n|y^{n-1}_{n-M},x_n)\nonumber\\
&\qquad\qquad-s\gamma_n(x_n,y^{n-1}_{n-N})={K}^s_n(y^{n-1}_{n-J}),~\forall{x_n},~\mbox{if}~\pi_n(dx_n|y^{n-1}_{n-J})\neq{0},\label{theorem:abstract:ftfi:lmco:cost:necessary:sufficient:equation1}\\
&\int_{{\cal Y}_{n}}\Big(\log\big(\frac{dq_n(\cdot|y^{n-1}_{n-M},x_n)}{d\nu^{\pi}_{n}(\cdot|y^{n-1}_{n-J})}(y_n)\big)\Big)q_n(dy_n|y^{n-1}_{n-M},x_n)\nonumber\\
&\qquad\qquad-s\gamma_n(x_n,y^{n-1}_{n-N})\leq{K}^s_n(y^{n-1}_{n-J}),~\forall{x_n},~\mbox{if}~\pi_n(dx_n|y^{n-1}_{n-J})={0}.\label{theorem:abstract:ftfi:lmco:cost:necessary:sufficient:equation2}
\end{align}
Moreover, $C_t(y^{t-1}_{t-J})={K}^s_n(y^{n-1}_{n-J})+s(n+1)\kappa$ corresponds to the value function $C_t(y^{t-1}_{t-J})$, defined by \eqref{algorithms:abstract:generalizations:lmco:equation17}, evaluated at $t=n$.\\
\noi{(b)} For each $t$, $y^{t-1}_{t-J}\in{\cal Y}^{t-1}_{t-J}$, there exist a ${K}^s_t(y^{t-1}_{t-J})$, which depends on $s\geq{0}$, such that the following hold. \begin{align}
&\int_{{\cal Y}_{t}}\Big(\log\big(\frac{dq_t(\cdot|y^{t-1}_{t-M},x_t)}{d\nu^{\pi}_{t}(\cdot|y^{t-1}_{t-J})}(y_t)\big)+K^s_{t+1}(y^{t}_{t+1-J})\Big)q_t(dy_t|y^{t-1}_{t-M},x_t)\nonumber\\
&\qquad\qquad-s\gamma_t(x_t,y^{t-1}_{t-N})=K^s_t(y^{t-1}_{t-J}),~\forall{x_t},~\mbox{if}~\pi_t(dx_t|y^{t-1}_{t-J})\neq{0},\label{theorem:abstract:ftfi:lmco:cost:necessary:sufficient:equation3}\\
&\int_{{\cal Y}_{t}}\Big(\log\big(\frac{dq_t(\cdot|y^{t-1}_{t-M},x_t)}{d\nu^{\pi}_{t}(\cdot|y^{t-1}_{t-J})}(y_t)\big)+K^s_{t+1}(y^{t}_{t+1-J})\Big)q_t(dy_t|y^{t-1}_{t-M},x_t)\nonumber\\
&\qquad\qquad-s\gamma_t(x_t,y^{t-1}_{t-N})\leq{K}^s_t(y^{t-1}_{t-J}),~\forall{x_t},~\mbox{if}~\pi_t(dx_t|y^{t-1}_{t-J})={0}\label{theorem:abstract:ftfi:lmco:cost:necessary:sufficient:equation4}
\end{align}
for $t=n-1,\ldots,0$. Moreover, $C_t(y^{t-1}_{t-J})={K}^s_t(y^{t-1}_{t-J})+s(n+1)\kappa$ corresponds to the value function $C_t(y^{t-1}_{t-J})$, defined by \eqref{algorithms:abstract:generalizations:lmco:equation17}, evaluated at $t=n-1,\ldots,0$.
\end{theorem}
\begin{proof}
Since we assume conditions (C1)--(C4), we can repeat the derivation of Theorem~\ref{theorem:necessary:sufficient:ftfi:lmco:cost} for abstract alphabets.
\end{proof}

%
%
%

\subsection{Necessary and Sufficient Conditions for Channels of Class $B$ with Transmission Cost of Classes $A$ or $B$}\label{section:generalizations}

\par In this subsection, we illustrate how the main results of this paper extend to channels of class $B$ with transmission cost of classes $A$ or $B$. 
\subsubsection{Channels of class $A$ with transmission cost $B$} Consider the channel distributions of class $A$ given by \eqref{introduction:class:channel:equation1}, and a transmission cost function of class $B$ given by \eqref{TC_2}. By \cite{kourtellaris-charalambous2015aieeeit}, the characterization of FTFI capacity with average transmission cost constraint is given by 
\begin{align}
{C}_{X^n \rightarrow Y^n}^{FB,A.B}(\kappa) 
= \sup_{{\cal P}_{0,n}^{B}(\kappa)} \sum_{t=0}^n {\bf E}^{ \pi}\left\{
\log\Big(\frac{q_t(\cdot|Y_{t-M}^{t-1},X_t)}{\nu_t^{{\pi}}(\cdot|Y^{t-1})}(Y_t)\Big)
\right\}, \label{subsection:extension:characterization:ftfi:capacity:cost:lmco1}
\end{align}
where
\begin{align}
{\cal P}_{0,n}^{B}(\kappa)\triangleq\Big\{\pi_t(x_t|y^{t-1}), ~t=0, \ldots, n: \frac{1}{n+1} {\bf E}^{\pi} \Big( c^{B}_{0,n}(X^n, Y^{n-1}) \Big)\leq  \kappa\Big\},~ \kappa \in [0,\infty)\label{subsection:extension:characterization:ftfi:fmco:transmission:cost}
\end{align}
and the joint and transition probabilities are given by
\begin{align}
{\bf P}^{\pi}(dy^t, dx^t) =&\prod_{i=0}^tq_i(dy_i|y_{i-M}^{i-1}, x_i)\pi_i(dx_i|y^{i-1}), \label{subsection:extension:section:lmco:equation1} \\
\nu_t^{\pi}(dy_t|y^{t-1}) =&\int_{{\cal X}_t} q_t(dy_t|y_{t-M}^{t-1}, x_t)\pi_t(dx_t|y^{t-1}),~t\in\mathbb{N}_0^{n}.\label{subsection:extension:section:umco:equation2}
\end{align}
From \eqref{subsection:extension:characterization:ftfi:capacity:cost:lmco1}
-\eqref{subsection:extension:section:umco:equation2}, the analogue of Theorem~\ref{theorem:abstract:necessary:sufficient:ftfi:lmco:cost} is obtained by setting 
\begin{align*}
\gamma_t(x_t,y^{t-1}_{t-N}) \longmapsto \gamma_t(x_t,y^{t-1}),~~~\pi_t(dx_t|y^{t-1}_{t-J}) \longmapsto  \pi_t(dx_t|y^{t-1}),~~~\nu^{\pi}_t(dy_t|y^{t-1}_{t-J})\longmapsto \nu^{\pi}_t(dy_t|y^{t-1})
\end{align*}
Similarly, from \cite{kourtellaris-charalambous2015aieeeit} it follows than if the channel is of class $B$ and the transmission cost function is of classes $A$, or $B$, the analogue of Theorem~\ref{theorem:abstract:necessary:sufficient:ftfi:lmco:cost} is obtained by setting
\begin{align*}
q_t(dy_t|y^{t-1}_{t-M},x_t)\longmapsto q_t(dy_t|y^{t-1},x_t),~~~\pi_t(dx_t|y^{t-1}_{t-J}) \longmapsto  \pi_t(dx_t|y^{t-1}),~~~\nu^{\pi}_t(dy_t|y^{t-1}_{t-J})\longmapsto \nu^{\pi}_t(dy_t|y^{t-1}).
\end{align*}

\section{Conclusions and Future Directions}\label{conclusion}

\par In this paper, we derived sequential necessary and sufficient conditions for any channel input conditional distribution to maximize the finite-time horizon directed information with or without transmission cost constraints. We applied the necessary and sufficient conditions to several application examples and we derived recursive closed form expressions for the optimal channel input conditional distributions,  which maximize the finite-time horizon directed information. For the investigated application examples, we also illustrated how to derive the closed form expressions of feedback capacity and capacity achieving distributions. The methodology introduced in this paper is general and can be applied to a variety of general channels with memory, such as, the Gaussian channels with memory investigated in \cite{charalambous-kourtellaris-loyka2016ieeeit}.\\
The future research directions are focused on addressing the following issues.
\begin{description}
\item[(a)] Apply the necessary and sufficient conditions to other application examples.
\item[(b)] Derive necessary and sufficient conditions for general channels of the form $\{{\bf P}_{Y_t|Y^{t-1}_{t-M},X^{t-1}_{t-L}}:~t\in\mathbb{N}_0^n\}$, when $\{M,L\}$ are nonnegative finite integers.
\end{description}

%
%
%
%
\appendices

\section{Feedback Codes}\label{section:feedback:codes}

A sequence of feedback codes $\{(n, { M}_n, \epsilon_n):n=0, 1, \dots\}$ is defined by the following elements.\\
(a)  A set of messages ${\cal M}_n \triangleq \{ 1,  \ldots, M_n\}$ and a set of encoding maps,  mapping source messages  into channel inputs of block length $(n+1)$, defined by
\begin{align}
{\cal E}_{[0,n]}^{FB}(\kappa) \triangleq & \Big\{g_t: {\cal M}_n \times {\cal Y}^{t-1}  \longmapsto {\cal X}_t,~~ x_0=g_0(w, y^{-1}), x_t=e_t(w, y^{i-1}),~~  w\in {\cal M}_n, ~t=0, \ldots, n:\nonumber \\
& \frac{1}{n+1} {\bf E}^g\Big(c_{0,n}(X^n,Y^{n-1})\Big)\leq \kappa  \Big\}. \label{block-code-nf-non}
\end{align}
The codeword for any $w \in {\cal M}_n$  is $u_w\in{\cal X}^n$, $u_w=(g_0(w,y^{-1}), g_1(w, y^0),
,\dots,g_n(w, y^{n-1}))$, and ${\cal C}_n=(u_1,u_2,\dots,u_{{M}_n})$ is  the code for the message set ${\cal M}_n$. In general, the code  depends on the initial data $Y^{-1}=y^{-1}$ ( unless it can be shown that in  the limit, as $n \longrightarrow \infty$, the induced channel output process  has a unique invariant distribution).  \\
(b)  Decoder measurable mappings $d_{0,n}:{\cal Y}^n\longmapsto {\cal M}_n$, ${Y}^n= d_{0,n}(Y^{n})$, such that the average
probability of decoding error satisfies 
\begin{align}
{\bf P}_e^{(n)} \triangleq \frac{1}{M_n} \sum_{w \in {\cal M}_n} {\bf  P}^g \Big\{d_{0,n}(Y^{n}) \neq w |  W=w\Big\}\equiv {\bf  P}^g\Big\{d_{0,n}(Y^n) \neq W \Big\} \leq \epsilon_n\nonumber
\end{align}
where  $r_n\triangleq \frac{1}{n+1} \log M_n$ is  the coding rate or transmission rate (and the messages are uniformly distributed over ${\cal M}_n$), and $Y^{-1}=y^{-1}$ is known to the decoder. Alternatively, both the encoder and decoder assume no information, i.e., $Y^{-1}=\{\emptyset\}$. \\  
A rate $R$ is said to be an achievable rate, if there exists  a  code sequence satisfying
$\lim_{n\longrightarrow\infty} {\epsilon}_n=0$ and $\liminf_{n \longrightarrow\infty}\frac{1}{n+1}\log{{M}_n}\geq R$. The feedback capacity is defined by $C\triangleq \sup \{R: R \: \mbox{is achievable}\}$.\\

\noi By invoking standard techniques often applied in deriving coding theorems, $C_{X^\infty \rar Y^\infty}^{FB}$ is the supremum of all achievable feedback codes, provided the following conditions hold. \\
(C1) The messages $w\in{\cal M}_n$ to be encoded and transmitted over the channel satisfy the following conditional independence.
\begin{align}
{\bf P}_{Y_t|Y^{t-1},X^t,W}(dy_t|y^{t-1},x^t,w)={\bf P}_{Y_t|Y^{t-1},X^t}(dy_t|y^{t-1},x^t),~t\in\mathbb{N}^n_0.\label{introduction:condition:equation1}
\end{align}
If \eqref{introduction:condition:equation1} is violated, then  $I(X^n\rightarrow{Y^n})$ is no longer a tight bound on any achievable code rate \cite{massey1990}. \\
(C2) There exists a channel input distribution denoted by $\{{\bf P}_{X_t|X^{t-1},Y^{t-1}}^*:~t\in\mathbb{N}_0^n\}\in{\cal P}_{0,n}$ which achieves the supremum in $C^{FB}_{X^n\rightarrow{Y^n}}$, and the per unit time limit $\lim_{n\longrightarrow\infty}\frac{1}{n+1}C^{FB}_{X^n\rightarrow{Y^n}}$ exists and it is finite. \\
If any one of theses conditions is violated, then the arguments of the  converse coding theorem, which are based on Fano's inequality  do not apply.\\
(C3) The optimal channel input distribution $\{{\bf P}^*_{X_t|X^{t-1},Y^{t-1}}:~t\in\mathbb{N}_0^n\}\in{\cal P}_{0,n}$, which achieves the supremum in $C^{FB}_{X^n\rightarrow{Y^n}}$ induces stability in the sense of Dobrushin \cite{dobrushin1959}, of the directed information density, that is,
\begin{align*}
\lim_{n\longrightarrow\infty}{\bf P}_{X^n, Y^n}^{{\bf P}^*}\Big\{(X^n,Y^n)\in{\cal X}^n\times{\cal Y}^n:~\frac{1}{n+1}\left|{\bf E}^{{\bf P}^*}\{i^{{\bf P}^{*}}(X^n,Y^n)\}-i^{{\bf P}^{*}}(X^n,Y^n)\right|>\epsilon\Big\}=0
\end{align*}
where $i^{{\bf P}^{*}}(X^n,Y^n)$ is the directed information density,  defined by 
\begin{align}
\sum_{t=0}^n\log\Big(\frac{d{\bf P}_{Y_t|Y^{t-1},X^t}(\cdot|y^{t-1},x^t)}{d{\bf P}_{Y_t|Y^{t-1}}^{{\bf P}^*}(\cdot|y^{t-1})}(Y_t)\Big).\nonumber
\end{align}
and the superscript notation indicates the dependence of the distributions on the optimal distribution $\{{\bf P}^*_{X_t|X^{t-1},Y^{t-1}}:~t\in\mathbb{N}_0^n\}\in{\cal P}_{0,n}$. \\
 This condition is sufficient to  show achievability.

\section{Proofs of Section~\ref{section:generalizations:lmco}}
\label{proofs:section:ftfi:lmco}

\subsection{Proof of Theorem~\ref{theorem:generalization:ftfi:lmco:cost:double maximization}}\label{proof:theorem:generalization:ftfi:lmco:cost:double maximization}

\par {(a)} Expressions (\ref{algorithms:generalizations:lmco:equation21}), (\ref{algorithms:generalizations:lmco:equation22}) can be easily obtained from \eqref{algorithms:generalizations:lmco:equation17} and \eqref{applications:variational:equalities:lmco:equation1}. (i) (\ref{algorithms:generalizations:lmco:equation25aa}) follows from Corollary~\ref{corollary:application:variational:equalities:lmco}, \eqref{applications:variational:equalities:lmco:equation2}. We show (\ref{algorithms:generalizations:lmco:equation26}), by performing the maximization in (\ref{algorithms:generalizations:lmco:equation21}), using the fact that the problem is convex. For a fix $r_n(x_n|y_{n-M}^{n-1},y_n)$, we calculate the derivative of the right hand side of (\ref{algorithms:generalizations:lmco:equation21}) with respect to each of the elements of the probability vector $\{\pi_n(x_n|y_{n-J}^{n-1}):~{x_n}\in{\cal X}_n\}$ for a fixed $y_{n-J}^{n-1}\in{\cal Y}_{n-J}^{n-1}$ in (\ref{algorithms:generalizations:lmco:equation21}), by introducing the Lagrange multiplier $\lambda_n(y_{n-J}^{n-1})$ of the constraint $\sum_{x_n}\pi_n(x_n|y_{n-J}^{n-1})=1$, and imposing another Lagrange multiplier $s\geq{0}$ for the transmission cost constraint as follows.
\begin{align}
&\frac{\partial}{\partial \pi_n}\Big\{\sum_{x_n,y_{n}}\log\Big(\frac{r_{n}(x_n|y_n,y^{n-1}_{n-M})}{\pi_{n}(x_n|y^{n-1}_{n-J})}\Big)q_n(y_n|y^{n-1}_{n-M},x_n)\pi_n(x_n|y^{n-1}_{n-J})-s\sum_{x_n}\gamma_n(x_n,y^{n-1}_{n-N})\pi_n(x_n|y^{n-1}_{n-J})
\nonumber\\
&+\lambda_n(y^{n-1}_{n-J})\Big{(}\sum_{x_n}\pi_{n}(x_n|y^{n-1}_{n-J})-1\Big{)}\Big\}=0,~\forall{x_n}\in{\cal X}_n,~y_{n-J}^{n-1}\in{\cal Y}_{n-J}^{n-1}~\mbox{is fixed}\label{section:ftfi:lmco:derivative:equation1}
\end{align}
where $\frac{\partial}{\partial \pi_n}$ denotes the derivative with respect to a specific element of $\{\pi_n(x_n|y_{n-J}^{n-1}):~{x_n}\in{\cal X}_n\}$, and $y_{n-J}^{n-1}\in{\cal Y}_{n-J}^{n-1}$ is fixed. From \eqref{section:ftfi:lmco:derivative:equation1}, we obtain
\begin{align}
&\pi_{n}(x_n|y^{n-1}_{n-J})\nonumber\\
&=\exp{\Big\{\sum_{y_{n}}\log\big(r_{n}(x_n|y_n,y^{n-1}_{n-M}\big)q_n(y_n|y^{n-1}_{n-M},x_n)-1-s\gamma_n(x_n,y^{n-1}_{n-N})+\lambda_n(y^{n-1}_{n-J})\Big\}},~\forall{x_n}\in{\cal X}_n.
\end{align}\label{section:ftfi:lmco:derivative:equation2}
From \eqref{section:ftfi:lmco:derivative:equation2}, in view of $\sum_{x_n}\pi_n(x_n|y^{n-1}_{n-J})=1$, we obtain
\begin{align}
&\lambda(y^{n-1}_{n-J})\nonumber\\
&=-\log\Big(\sum_{x_n}\exp{\Big\{\sum_{y_{n}}\log\big(r_{n}(x_n|y_n,y^{n-1}_{n-M}\big)q_n(y_n|y^{n-1}_{n-M},x_n)-1-s\gamma_n(x_n,y^{n-1}_{n-N})\Big\}}\Big).\label{section:ftfi:lmco:derivative:equation3}
\end{align}
Substituting \eqref{section:ftfi:lmco:derivative:equation3} in \eqref{section:ftfi:lmco:derivative:equation2} we obtain (\ref{algorithms:generalizations:lmco:equation26}). (ii) (\ref{algorithms:generalizations:lmco:equation28}) follows from Corollary~\ref{corollary:application:variational:equalities:lmco}, \eqref{applications:variational:equalities:lmco:equation2}. To show (\ref{algorithms:generalizations:lmco:equation29}), we repeat the derivation of (\ref{algorithms:generalizations:lmco:equation26}), by tracking the additional second RHS term in (\ref{algorithms:generalizations:lmco:equation22}), to obtain the following expression.
\begin{align}
&\frac{\partial}{\partial \pi^{}_t}\Big\{\sum_{x_t,y_{t}}\log\Big(\frac{r_t(x_t|y^{t-1}_{t-M},y_t)}{{\pi}_{t}(x_t|y^{t-1}_{t-J})}\Big)+C_{t+1}(y^t_{t+1-J})\Big)q_t(y_t|y^{t-1}_{t-M},x_t){\pi}_t(x_t|y^{t-1}_{t-J})\nonumber\\
&-s\sum_{x_t}\gamma_t(x_t,y^{t-1}_{t-N}){\pi}_t(x_t|y^{t-1}_{t-J})+\lambda_t(y^{t-1}_{t-J})\Big{(}\sum_{x_t}\pi^r_{t}(x_t|y^{t-1}_{t-J})-1\Big{)}\Big\}=0,~\forall{x_t}\in{\cal X}_t,~t\in\mathbb{N}_0^{n-1}.\label{section:ftfi:lmco:derivative:equation5}
\end{align}
From \eqref{section:ftfi:lmco:derivative:equation5} we obtain
\begin{align}
&\pi_{t}(x_t|y_{t-J}^{t-1})\nonumber\\
&=\exp{\Big\{\sum_{y_{t}}\Big(\frac{r_t(x_t|y^{t-1}_{t-M},y_t)}{{\pi}_{t}(x_t|y^{t-1}_{t-J})}\Big)+C_{t+1}(y^t_{t+1-J})\Big)q_t(y_t|y^{t-1}_{t-M},x_t)-1-s\gamma_t(x_t,y^{t-1}_{t-N})+\lambda_t(y^{t-1}_{t-J})\Big\}},\nonumber\\
&\hspace{9.5cm}~\forall{x_t}\in{\cal X}_t,~t\in\mathbb{N}_0^{n-1}.\label{section:ftfi:lmco:derivative:equation6}
\end{align}
Using $\sum_{x_t}\pi_{t}(x_t|y_{t-J}^{t-1})=1,~t\in\mathbb{N}_0^{n-1}$ and \eqref{section:ftfi:lmco:derivative:equation6} we obtain
\begin{align}
&\lambda_t(y^{t-1}_{t-J})\nonumber\\
&=-\log\Big(\sum_{x_t}\exp{\Big\{\sum_{y_{t}}\Big(\frac{r_t(x_t|y^{t-1}_{t-M},y_t)}{{\pi}_{t}(x_t|y^{t-1}_{t-J})}\Big)+C_{t+1}(y^t_{t+1-J})\Big)q_t(y_t|y^{t-1}_{t-M},x_t)-1-s\gamma_t(x_t,y^{t-1}_{t-N})\Big\}}\Big),\nonumber\\
&\hspace{12.5cm}~t\in\mathbb{N}_0^{n-1}.\label{section:ftfi:lmco:derivative:equation7}
\end{align}
Substituting \eqref{section:ftfi:lmco:derivative:equation7} in \eqref{section:ftfi:lmco:derivative:equation6} we obtain (\ref{algorithms:generalizations:lmco:equation29}). (iii)  \eqref{algorithms:generalizations:lmco:equation27} follows by substituting (\ref{algorithms:generalizations:lmco:equation25aa}) into (\ref{algorithms:generalizations:lmco:equation26}). \eqref{algorithms:generalizations:lmco:equation30} follows by substituting (\ref{algorithms:generalizations:lmco:equation28}) into (\ref{algorithms:generalizations:lmco:equation29}).\\
\noi{(c)} Since $\mu(dy_{-J}^{-1})$ is fixed, then (\ref{algorithms:generalizations:lmco:equation31}) follows directly from {(a)}, by evaluating $C_t(y^{t-1}_{t-J})$ given by (\ref{algorithms:generalizations:lmco:equation29}) at $t=0$, and taking the expectation.\qed

\subsection{Proof of Theorem~\ref{theorem:necessary:sufficient:ftfi:lmco:cost}}\label{proof:theorem:necessary:sufficient:ftfi:lmco:cost}

\par {(a)} Recall that the optimization problem given by (\ref{algorithms:generalizations:lmco:equation18}) is convex. Hence, we can apply Kuhn-Tucker theorem \cite{boyd-vandenberghe2004} to find  necessary and sufficient conditions for $\{\pi_{n}(x_n|y^{t-1}_{t-J}):~x_n\in{\cal X}_n\}$, to maximize $C_n(y^{t-1}_{t-J})$ by introducing the Lagrange multiplier $\lambda_n(y^{t-1}_{t-J})$ as follows.
\begin{align}
&\frac{\partial}{\partial{\pi_{n}}}\Bigg\{\sum_{x_n,y_n}\Big(\log\big(\frac{q_n(y_n|y^{n-1}_{n-M},x_n)}{\nu^{\pi}_{n}(y_n|y^{n-1}_{n-J})}\big)\Big)q_n(y_n|y^{n-1}_{n-M},x_n){\pi}_n(x_n|y^{n-1}_{n-J})\nonumber\\
&-s\sum_{x_n}\gamma_n(x_n,y_{n-N}^{n-1}){\pi}_n(x_n|y^{n-1}_{n-J})+\lambda_{n}(y^{n-1}_{n-J})\Big(\sum_{x_n}\pi_{n}(x_n|y^{n-1}_{n-J})-1\Big)\Bigg\}\leq{0}.\nonumber
\end{align}
By performing the differentiation, we obtain
\begin{align}
&\sum_{x_n,y_n}\Big(\frac{1}{\frac{q_n(y_n|y^{n-1}_{n-M},x_n)}{\nu^{\pi}_{n}(y_n|y^{n-1}_{n-J})}}\Big)\Big(\frac{-q_n(y_n|y^{n-1}_{n-M},x_n)\frac{\partial}{\partial{\pi_{n}}}\big(\nu_n^\pi(y_n|y^{n-1}_{n-J})\big)}{(\nu_n^\pi(dy_n|y^{n-1}_{n-J})^2})\Big)q_n(y_n|y^{n-1}_{n-M},x_n){\pi}_n(x_n|y^{n-1}_{n-J})\nonumber\\
&\qquad+\sum_{y_n}\log\big(\frac{q_n(y_n|y^{n-1}_{n-M},x_n)}{\nu^{\pi}_{n}(y_n|y^{n-1}_{n-J})}\big)\Big)q_n(y_n|y^{n-1}_{n-M},x_n)-s\gamma_n(x_n,y_{n-N}^{n-1})+\lambda_n(y^{n-1}_{n-J})\leq{0}.\label{proof:necessary:sufficient:lmco:equation1}
\end{align}
Further simplification of \eqref{proof:necessary:sufficient:lmco:equation1} gives
\begin{align}
\sum_{y_n}\log\big(\frac{q_n(y_n|y^{n-1}_{n-M},x_n)}{\nu^{\pi}_{n}(y_n|y^{n-1}_{n-J})}\big)\Big)q_n(y_n|y^{n-1}_{n-M},x_n)-s\gamma_n(x_n,y_{n-N}^{n-1})\leq{1}-\lambda_n(y^{n-1}_{n-J}).\label{proof:necessary:sufficient:lmco:equation2}
\end{align}
Multiplying both sides of \eqref{proof:necessary:sufficient:lmco:equation2} by $\pi_n(x_n|y_{n-J}^{n-1})$ and summing over $x_n$, for which $\pi_n(x_n|y^{n-1}_{n-J})\neq{0}$, gives the necessary and sufficient conditions for maximizing over $\pi_n(x_n|y^{n-1}_{n-J})$ given by \eqref{theorem:ftfi:lmco:cost:necessary:sufficient:equation1}-\eqref{theorem:ftfi:lmco:cost:necessary:sufficient:equation2}, which then implies that $K^s_n(y^{n-1}_{n-J})={C}_n(y^{n-1}_{n-J})-s(n+1)\kappa$ given by \eqref{theorem:ftfi:lmco:cost:necessary:sufficient:equation1}. \\
\noi{(b)} Consider the time $t=n-1$. Then by (\ref{algorithms:generalizations:lmco:equation19}), $C_n(y_{n-J}^{n-1})$ is a function of $\pi_n(x_n|y^{n-1}_{n-J})$ which is not subjected to optimization. Applying the Kuhn-Tucker conditions to (\ref{algorithms:generalizations:lmco:equation19}) we have the following.
\begin{align}
&\frac{\partial}{\partial{\pi_{n-1}}}\Bigg\{\sum_{x_{n-1},y_{n-1}}\Big(\log\big(\frac{q_{n-1}(y_{n-1}|y^{n-2}_{n-1-M},x_{n-1})}{\nu^{\pi}_{n-1}(y_n-1|y^{n-2}_{n-1-J})}\big)+C_{n}(y^{n-1}_{n-J})\Big)q_{n-1}(y_{n-1}|y^{n-2}_{n-1-M},x_{n-1})\nonumber\\
&{\pi}_{n-1}(x_{n-1}|y^{n-2}_{n-1-J})-s\sum_{x_{n-1}}\gamma_{n-1}(x_{n-1},y_{n-1-N}^{n-2}){\pi}_{n-1}(x_{n-1}|y^{n-2}_{n-1-J})\nonumber\\
&+\lambda_{n-1}(y^{n-2}_{n-1-J})\Big(\sum_{x_{n-1}}{\pi}_{n-1}(x_{n-1}|y^{n-2}_{n-1-J})-1\Big)\Bigg\}\leq{0}.\nonumber
\end{align}
By performing differentiation  we obtain 
\begin{align}
&\sum_{x_{n-1},y_{n-1}}\Big(\frac{1}{\frac{q_{n-1}(y_{n-1}|y^{n-2}_{n-1-M},x_{n-1})}{\nu^{\pi}_{n-1}(y_{n-1}|y^{n-2}_{n-1-J})}}\Big)\Big(\frac{-q_{n-1}(y_{n-1}|y^{n-2}_{n-1-M},x_{n-1})\frac{\partial}{\partial{\pi_{n-1}}}\big(\nu_{n-1}^\pi(y_{n-1}|y^{n-2}_{n-1-J})\big)}{(\nu_{n-1}^\pi(y_{n-1}|y^{n-2}_{n-1-J})^2})\Big)\nonumber\\
&\qquad\qquad{q}_{n-1}(y_{n-1}|y^{n-2}_{n-1-M},x_{n-1}){\pi}_{n-1}(x_{n-1}|y^{n-2}_{n-1-J})\nonumber\\
&+\sum_{y_{n-1}}\log\big(\frac{q_{n-1}(y_{n-1}|y^{n-2}_{n-1-M},x_{n-1})}{\nu^{\pi}_{n-1}(y_{n-1}|y^{n-2}_{n-1-J})}\big)\Big)q_{n-1}(y_{n-1}|y^{n-2}_{n-1-M},x_{n-1})\nonumber\\
&+\sum_{y_{n-1}}C_{n}(y^{n-1}_{n-J})q_{n-1}(y_{n-1}|y^{n-2}_{n-1-M},x_{n-1})-s\gamma_{n-1}(x_{n-1},y_{n-1-N}^{n-2})+\lambda_{n-1}(y^{n-2}_{n-1-J})\leq{0}.\label{proof:necessary:sufficient:lmco:equation3}
\end{align}
After simplifications, \eqref{proof:necessary:sufficient:lmco:equation3} gives the following.
\begin{align}
\sum_{y_{n-1}}&\Big(\log\big(\frac{q_{n-1}(y_{n-1}|y^{n-2}_{n-1-M},x_{n-1})}{\nu_{n-1}^\pi(y_{n-1}|y^{n-2}_{n-1-J})}\big)+C_{n}(y^{n-1}_{n-J})\Big)q_{n-1}(y_{n-1}|y^{n-2}_{n-1-M},x_{n-1})\nonumber\\
&\hspace{6cm}-s\gamma_{n-1}(x_{n-1},y_{n-1-N}^{n-2})\leq{1}-\lambda_{n-1}(y^{n-2}_{n-1-J}).\label{proof:necessary:sufficient:lmco:equation4}
\end{align}
To verify that ${1}-\lambda_{t}(y^{n-2}_{n-1-J})=C_{n-1}(y^{n-2}_{n-1-J})-s(n+1)\kappa\equiv{K}^s_{n-1}(y^{n-2}_{n-1-J})$, we multiply both sides of \eqref{proof:necessary:sufficient:lmco:equation4} by $\pi_{n-1}(x_{n-1}|y^{n-2}_{n-1-J})$ and sum over $x_{n-1}$, for which $\pi_{n-1}(x_{n-1}|y^{n-2}_{n-1-J})\neq{0}$, to obtain the necessary and sufficient conditions for ${\pi}_{n-1}(x_{n-1}|y^{n-2}_{n-1-J})$ to maximize $C_{n-1}(y^{n-2}_{n-1-J})-s(n+1)\kappa\equiv{K}^s_{n-1}(y^{n-2}_{n-1-J})$ given the necessary and sufficient conditions at $t=n$. Repeating this derivation for $t=n-2,n-3,\ldots,0$, or by induction, we obtain \eqref{theorem:ftfi:lmco:cost:necessary:sufficient:equation3}, \eqref{theorem:ftfi:lmco:cost:necessary:sufficient:equation4}. This completes the proof.\qed

\subsection{Alternative proof of Theorem~\ref{theorem:necessary:sufficient:ftfi:lmco:cost}}\label{remark:necessary:sufficient:alternative:proof}

\noi Here, we give an alternative proof to Theorem~\ref{theorem:necessary:sufficient:ftfi:lmco:cost} using Theorem~\ref{theorem:generalization:ftfi:lmco:cost:double maximization}. Recall that by Theorem~\ref{theorem:generalization:ftfi:lmco:cost:double maximization}, {(a)}, we have
\begin{align}
C_n(y^{n-1}_{n-J})&=\sup_{\pi_n(x_n|y^{n-1}_{n-J})}\sup_{r_n(x_n|y^{n-1}_{n-M},y_n)}\Bigg\{\sum_{x_n,y_n}\log\Big(\frac{r_n(x_n|y^{n-1}_{n-M},y_n)}{{\pi}_{n}(x_n|y^{n-1}_{n-J})}\Big){q}_n(y_n|y^{n-1}_{n-M},x_n)\pi_n(x_n|y^{n-1}_{n-J})\nonumber\\
&-s\Big(\sum_{x_n}\gamma_n(x_n,y^{n-1}_{n-N})\pi_n(x_n|y^{n-1}_{n-J})-(n+1)\kappa\Big)\Bigg\},~\forall{y^{n-1}_{n-J}\in{\cal Y}^{n-1}_{n-J}}.\label{baa:sequential:applications:equation12ii}
\end{align}
By \eqref{baa:sequential:applications:equation12ii}, for a fixed $r_n(x_n|y^{n-1}_{n-M},y_n)$, we calculate the derivative with respect to each of the elements of the probability vector $\{\pi_n(x_n|y^{n-1}_{n-J}):~{x_n}\in{\cal X}_n\}$, we incorporate the pointwise constraint $\sum_{x_n}\pi_n(x_n|y^{n-1}_{n-J})=1$, by introducing the Lagrange multiplier $\lambda_n(y^{n-1}_{n-J})$, and we also include a second Lagrange multiplier $s\geq{0}$ to encompass the transmission cost constraint as follows.
\begin{align}
&\frac{\partial}{\partial \pi_n}\Big\{\sum_{x_n,y_{n}}\log\Big(\frac{r_{n}(x_n|y^{n-1}_{n-M},y_n)}{\pi_{n}(x_n|y^{n-1}_{n-J})}\Big)q_n(y_n|y^{n-1}_{n-M},x_n)\pi_n(x_n|y^{n-1}_{n-J})\nonumber\\
&-s\sum_{x_n}\gamma_n(x_n,y^{n-1}_{n-N})\pi_n(x_n|y^{n-1}_{n-J})+\lambda_n(y^{n-1}_{n-J})\Big{(}\sum_{x_n}\pi_{n}(x_n|y^{n-1}_{n-J})-1\Big{)}\Big\}=0,~\forall{x_n}\in{\cal X}_n\label{necessary:sufficient:ftfi:umco:derivative:equation1ii}
\end{align}
where $\frac{\partial}{\partial \pi_n}$ denotes derivative with respect to a specific coordinate of the probability vectors $\{\pi_{n}(x_n|y^{n-1}_{n-J}):~x_n\in{\cal X}^n\}$. From \eqref{necessary:sufficient:ftfi:umco:derivative:equation1ii} we obtain
\begin{align}
\sum_{y_{n}}\log\Big(\frac{r_{n}(x_n|y^{n-1}_{n-M},y_n)}{\pi_{n}(x_n|y^{n-1}_{n-J})}\Big)q_n(y_n|y^{n-1}_{n-M},x_n)-s\gamma_n(x_n,y^{n-1}_{n-N})=1-\lambda_n(y^{n-1}_{n-J}),~\forall{x_n\in{\cal X}_n}.\label{necessary:sufficient:ftfi:umco:derivative:equation2ii}
\end{align}
By \eqref{algorithms:generalizations:lmco:equation25aa}, for a fixed $\pi_n(x_n|y^{n-1}_{n-J})$, the maximization with respect to $r_n(x_n|y^{n-1}_{n-J},y_n)$ is given by
\begin{align}
r^{*,\pi}_n(x_n|y^{n-1}_{n-M},y_n)=\Big(\frac{q_n(y_n|y^{n-1}_{n-M},x_n)}{\nu^{\pi}_{n}(y_n|y^{n-1}_{n-J})}\Big){\pi}_{n}(x_n|y^{n-1}_{n-J}).\label{necessary:sufficient:ftfi:umco:equation14ii}
\end{align}
Substituting \eqref{necessary:sufficient:ftfi:umco:equation14ii} in \eqref{necessary:sufficient:ftfi:umco:derivative:equation2ii} we obtain
\begin{align}
\sum_{y_{n}}\log\Big(\frac{q_n(y_n|y^{n-1}_{n-M},x_n)}{\nu^{\pi}_{n}(y_n|y^{n-1}_{n-J})}\Big)q_n(y_n|y^{n-1}_{n-M},x_n)-s\gamma_n(x_n,y^{n-1}_{n-N})=1-\lambda_n(y^{n-1}_{n-J}),~\forall{x_n\in{\cal X}_n}.\label{necessary:sufficient:ftfi:umco:15ii}
\end{align}
Summing both sides in \eqref{necessary:sufficient:ftfi:umco:15ii} with respect to $\pi_n(x_n|y^{n-1}_{n-J})$ we obtain \eqref{theorem:ftfi:lmco:cost:necessary:sufficient:equation1}.\\
Similarly, by Theorem~\ref{theorem:generalization:ftfi:lmco:cost:double maximization},~{(a)}, we have
\begin{align}
C_t(y^{t-1}_{t-J})&=\sup_{\pi_t(x_t|y^{t-1}_{t-J})}\sup_{r_t(x_t|y^{t-1}_{t-M},y_t)}\Bigg\{\sum_{x_t,y_t}\Big(\log\Big(\frac{r_t(x_t|y^{t-1}_{t-M},y_t)}{{\pi}_{t}(x_t|y^{t-1}_{t-J})}\Big)+C_{t+1}(y^t_{t+1-J})\Big)q_t(y_t|y^{t-1}_{t-M},x_t){\pi}_t(x_t|y^{t-1}_{t-J})\nonumber\\
&-s\Big(\sum_{x_t}\gamma_t(x_t,y^{t-1}_{t-N})\pi_t(x_t|y^{t-1}_{t-J})-(n+1)\kappa\Big)\Bigg\},~\forall{y^{t-1}_{t-J}\in{\cal Y}^{t-1}_{t-J}},~t\in\mathbb{N}_0^{n-1}.\label{baa:sequential:applications:equation13ii}
\end{align}
By \eqref{baa:sequential:applications:equation13ii}, for each $t$, and a fixed $r_t(x_t|y_{t-M}^{t-1},y_t)$, we calculate the derivative with respect to each of the elements of the probability vector $\{\pi_t(x_t|y_{t-J}^{t-1}):~{x_t}\in{\cal X}_t\}$, and we incorporate the constraints to obtain
\begin{align}
\sum_{y_{t}}\Big(\log\Big(\frac{r_t(x_t|y^{t-1}_{t-M},y_t)}{{\pi}_{t}(x_t|y^{t-1}_{t-J})}\Big)+C_{t+1}(y^t_{t+1-J})\Big)q_t(y_t|y^{t-1}_{t-M},x_t)-s\gamma_t(x_t,y^{t-1}_{t-N})=1-\lambda_t(y^{t-1}_{t-J}),~\forall{x_t}\in{\cal X}_t.\label{necessary:sufficient:ftfi:umco:derivative:equation3iii}
\end{align}
By \eqref{algorithms:generalizations:lmco:equation28}, for fixed $\pi_t(x_t|y^{t-1}_{t-J})$, the maximization with respect to $r_t(x_t|y^{t-1}_{t-M},y_t)$ is given by
\begin{align}
r^{*,\pi}_t(x_t|y^{t-1}_{t-M},y_t)=\Big(\frac{q_t(y_t|y^{t-1}_{t-M},x_t)}{\nu^{\pi}_{t}(y_t|y^{t-1}_{t-J})}\Big){\pi}_{t}(x_t|y^{t-1}_{t-J}),~\forall{x_t}\in{\cal X}_t,~t\in\mathbb{N}_0^{n-1}.
\label{necessary:sufficient:ftfi:umco:equation16iii}
\end{align}
By substituting \eqref{necessary:sufficient:ftfi:umco:equation16iii} in \eqref{necessary:sufficient:ftfi:umco:derivative:equation3iii} we obtain
\begin{align}
\sum_{y_{t}}\Big(\log\Big(\frac{q_t(y_t|y^{t-1}_{t-M},x_t)}{\nu^{\pi}_{t}(y_t|y^{t-1}_{t-J})}\Big)+C_{t+1}(y^t_{t+1-J})\Big)q_t(y_t|y^{t-1}_{t-M},x_t)-s\gamma_t(x_t,y^{t-1}_{t-N})=1-\lambda_t(y^{t-1}_{t-J}),~\forall{x_t}\in{\cal X}_t.\label{necessary:sufficient:ftfi:umco:15iii}
\end{align}
By summing both sides in \eqref{necessary:sufficient:ftfi:umco:15iii} with respect to $\pi_t(x_t|y^{t-1}_{t-J})$, we obtain
 \eqref{theorem:ftfi:lmco:cost:necessary:sufficient:equation3}, for $t=n-1,n-2,\ldots,0$. Inequalities in \eqref{theorem:ftfi:lmco:cost:necessary:sufficient:equation2}, \eqref{theorem:ftfi:lmco:cost:necessary:sufficient:equation4} can be obtained similarly from Kuhn-Tucker conditions. This completes the proof.\qed

\bibliographystyle{IEEEtran}
\bibliography{bibliography}

\end{document}